\documentclass{article}

\usepackage[english]{babel}

\usepackage[letterpaper,top=1in,bottom=1in,left=1in,right=1in]{geometry}

\usepackage{xcolor}
\usepackage{enumerate}
\usepackage{amsmath,amsthm,amssymb}
\usepackage{graphicx}
\usepackage{hyperref}
\usepackage[linesnumbered,ruled,vlined,algo2e]{algorithm2e}
\usepackage{algorithm}
\usepackage{algpseudocode}
\usepackage{svg}
\usepackage[most]{tcolorbox}
\usepackage{thm-restate}
\hypersetup{
    colorlinks = true,
    allcolors = black
}

\usepackage[nameinlink]{cleveref}

\newtheorem{theorem}{Theorem}
\newtheorem{corollary}[theorem]{Corollary}
\newtheorem{definition}[theorem]{Definition}
\newtheorem{lemma}[theorem]{Lemma}

\newtheorem{observation}[theorem]{Observation}

\newcommand{\A}{\mathcal{A}}
\newcommand{\C}{\mathcal{C}}
\newcommand{\E}{\mathcal{E}}

\newcommand{\V}{\mathcal{V}}
\newcommand{\La}{\mathcal{L}}
\newcommand{\T}{\mathcal{T}}
\newcommand{\Hp}{\mathcal{H}}
\newcommand{\I}{\mathcal{I}}
\newcommand{\Type}{\mathbb{T}}
\newcommand{\Pow}{\mathcal{P}}

\newcommand{\hE}{\bar{E}}
\newcommand{\cen}{\mathfrak{c}}
\newcommand{\Nh}{\mathcal{N}}

\newcommand{\Sin}{\Sigma_{\text{in}}}
\newcommand{\Sout}{\Sigma_{\text{out}}}
\newcommand{\poly}{\text{poly}}

\newcommand{\poLFL}{\mathbf{poLFL}}

\newcommand{\N}{\mathbb{N}}

\newcommand{\set}[1]{\left\{#1\right\}}

\newcommand\utimes{\mathbin{\ooalign{$\cup$\cr%
   \hfil\raise0.42ex\hbox{$\scriptscriptstyle\times$}\hfil\cr}}}
\newcommand\bigutimes{\mathop{\ooalign{$\bigcup$\cr%
   \hfil\raise0.36ex\hbox{$\scriptscriptstyle\boldsymbol{\times}$}\hfil\cr}}}

\newcommand\blfootnote[1]{
    \begingroup
    \renewcommand\thefootnote{}\footnote{#1}
    \addtocounter{footnote}{-1}
    \endgroup
}

\newenvironment{myabstract}
{\list{}{\listparindent 1.5em%
		\itemindent    \listparindent
		\leftmargin    1cm
		\rightmargin   1cm
		\parsep        0pt}%
	\item\relax}
{\endlist}

\newenvironment{mycover}
{\list{}{\listparindent 0pt
		\itemindent    \listparindent
		\leftmargin    1cm
		\rightmargin   1cm
		\parsep        0pt}
	\raggedright
	\item\relax}
{\endlist}

\begin{document}

\begin{mycover}
    \vspace*{10mm}
    \begin{center}
    \begin{minipage}{0.75\textwidth}
    \begin{center}
    {\Huge LCLs Beyond Bounded Degrees\par}
    \bigskip
    {\Large Gustav Schmid}

    { University of Freiburg, Germany}  
    \bigskip
    \bigskip

    \textbf{Abstract}
    \end{center}
    \end{minipage} 
    \end{center}
    \blfootnote{schmidg@informatik.uni-freiburg.de}
  \end{mycover}

\begin{myabstract}
    The study of Locally Checkable Labelings (LCLs) has led to a remarkably precise characterization of the distributed time complexities that can occur on bounded-degree trees \cite{BBOS18almostGlobal, BHOS19HomogeneousLCL, brandt21trees, chang20, CKP19exponential, CP19timeHierarchy}. A central feature of this \emph{complexity landscape} is the existence of strong gap results, which rule out large ranges of intermediate complexities. While it was initially hoped that these gaps might extend to more general graph classes \cite{chang20}, this has turned out not to be the case \cite{BHKLOS18lclComplexity}. In this work, we investigate a different direction: we remain in the class of trees, but allow arbitrarily large degrees.

    We focus on the \emph{polynomial regime}, i.e. complexities of the form $\Theta(n^{1/k})$ for $k \in \mathbb{N}$, and show that whether polynomial gap results persist in the unbounded-degree setting crucially depends on how LCLs are generalized beyond bounded degrees. 
    
    A complex construction in \cite{LaurentDiamCheckers} shows that the polynomial gaps also vanish for LCLs on unbounded-degree trees. Rather than stopping at this negative result, we give a much simpler set of problems that already contradicts the existence of any polynomial gaps. The insight obtained from this cleaner construction is that for gap results to exist, we cannot allow problem definitions to distinguish infinitely many local cases.

    This guides us to a natural class of problems for which polynomial gap results can still be recovered. We introduce \emph{Locally Finite Labelings} (LFLs), which formalize the intuition that \emph{every node must fall into one of finitely many local cases}, even in the presence of unbounded degrees.

    Our main result shows that this restriction is sufficient to restore the polynomial gaps: for any LFL $\Pi$ on trees of unbounded degrees, the deterministic LOCAL complexity of $\Pi$ is either
    \begin{itemize}
        \item $\Theta(n^{1/k})$ for some integer $k \geq 1$, or
        \item $O(\log n)$.
    \end{itemize}
\end{myabstract}

\clearpage

\section{Introduction}

Locally Checkable Labeling (LCL) problems, introduced by Naor and Stockmeyer in 1995~\cite{NaorStockmeyer95}, form a central framework for studying the locality of distributed graph labeling problems. An LCL is specified by a finite set of labeled $r$-hop neighborhoods ($r \in O(1)$), and a labeling is valid if the $r$-hop neighborhood of every node is isomorphic to one of these configurations. The finiteness of this set inherently restricts LCLs to bounded-degree graphs: if node degrees were unbounded, the number of possible local neighborhoods would no longer be finite, even for $r=1$.

LCLs are typically studied in the LOCAL model of distributed computing. In this model, the input is a communication network $G=(V,E)$ whose nodes represent processors and whose edges represent communication links. Computation proceeds in synchronous rounds, and in each round, nodes may exchange messages with their neighbors and perform local computation. After some number of rounds, each node must output a label, and these labels must collectively satisfy the constraints of the LCL. We focus on the LOCAL model in this work, though many of our ideas are compatible with the bandwidth restricted CONGEST model as well (using techniques from \cite{smallMessagesbcmos21}).

Over the past decade, the study of LCLs has led to remarkably clean and surprising complexity classifications. On bounded-degree trees, deterministic LCLs admit exactly the following round complexities:
\[
O(1),\ \Theta(\log^* n),\ \Theta(\log n),\ \text{and } \Theta(n^{1/k}) \text{ for any integer } k \ge 1,
\]
with randomness only helping in the $\Theta(\log n)$ class, reducing some problems in this class to $\Theta(\log\log n)$~\cite{BBOS18almostGlobal, BHOS19HomogeneousLCL, brandt21trees, chang20, CKP19exponential, CP19timeHierarchy}. Crucially, no other complexities are possible. Results ruling out intermediate complexities, so-called \emph{gap results}, stand in stark contrast to classical sequential complexity theory, where finer-grained hierarchies are ubiquitous. Many of these gaps arise from constructive speedup theorems, which transform any sufficiently fast algorithm into an even faster one.

These results give rise to what is often called the \emph{complexity landscape} of LCLs on bounded-degree trees, illustrated in \Cref{fig:Landscape}. The existence of large forbidden regions is what shapes the landscape, but perhaps more importantly these gap results often unveil a deeper understanding of why the complexity landscape is the way it is.

\begin{figure}[!ht]
    \centering
    \includesvg[width=0.7\linewidth]{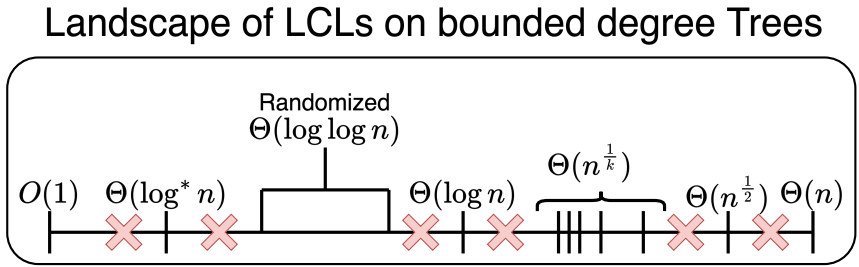}
    \caption{The deterministic complexity landscape of LCLs on bounded-degree trees. Every such LCL has one of the depicted complexities; no others are possible. In the randomized setting, the class $\Theta(\log\log n)$ also appears.}
    \label{fig:Landscape}
\end{figure}

Given the strength of these results, it was natural to hope that similarly clean landscapes might extend beyond trees. Indeed, it was conjectured, for instance, in the work of Chang and Pettie~\cite{chang20}, that such gaps might extend to general graphs. However, this hope was quickly dispelled. On bounded-degree general graphs, the complexity landscape becomes much denser: LCLs can have complexities such as $\poly(\log n)$, or arbitrary $\poly(n)$, and most of the gap results that hold on trees disappear~\cite{BHKLOS18lclComplexity}. Thus, even under the bounded-degree assumption, moving from trees to general graphs fundamentally changes the nature of the landscape.

This raises a different question: rather than generalizing from trees to general graphs, what happens if we generalize in the \emph{other} direction? That is, what if we retain the tree structure, which is responsible for many of the clean results, but drop the bounded-degree assumption?

Some evidence already suggests that unbounded degrees significantly alter the picture. Prior work has shown that, even on trees, the complexity of certain local problems depends explicitly on the maximum degree $\Delta$ (e.g.,~\cite{Balliu2019}), and that this dependence can lead to new asymptotic behaviors not present in the bounded-degree setting. In particular, by choosing $\Delta$ as a function of $n$, one can obtain deterministic complexities that do not fit into the classical bounded-degree tree landscape. This indicates that allowing unbounded degrees introduces genuinely new phenomena, even without cycles.

Unfortunately, investigating LCLs in the unbounded-degree tree setting also leads to a negative result. Bousquet, Feuilloley, and Pierron~\cite{LaurentDiamCheckers} show that on unbounded-degree trees any complexity $f(n) \in O(n/\log n)$ can occur\footnote{This is slightly imprecise, please refer to \cite{LaurentDiamCheckers} for the precise statement.} for what they call \emph{generalized LCLs}. 

However, the two negative results in bounded-degree general graphs and unbounded-degree trees differ in an important way.
The construction for general graphs in \cite{BHKLOS18lclComplexity} yielding arbitrary complexities for bounded-degree general graphs allows a meaningful interpretation, namely that we can abuse cycles to define easier versions of problems in a structured way.

In contrast the result of \cite{LaurentDiamCheckers} for the unbounded-degree tree setting, offers little to no insight into \emph{why} such generalized LCLs can behave so erratically. Their work is primarily geared toward locally verifying global graph properties, such as the diameter, via a framework of \emph{local checkers}. The corresponding LCL result is obtained through a reduction into this framework rather than through an explicit construction of LCLs. As a consequence, it is unclear which aspects are responsible for the collapse of structure in the complexity landscape of generalized LCLs. Do we only run into trouble, because we allow esoteric constructions in our generalization, or is this a problem fundamental to the unbounded-degree setting?

\vspace{0.15cm}
\begin{tcolorbox}%[innertopmargin=8pt,innerbottommargin=5pt,innermargin=5pt]
	%\center
	\textbf{Open Question} \noindent What causes the vanishing of gap results for LCLs in unbounded-degree trees?
\end{tcolorbox}

When considering LCLs in the unbounded-degree setting, a fundamental obstacle arises immediately: LCLs, as traditionally defined, do not make sense on unbounded-degree graphs, because they rely on a finite set of local configurations. Thus, before asking meaningful complexity questions, we must first decide how to generalize the notion of an LCL itself.
We believe that by carefully examining how exactly we define our class of problems we can get to the bottom of this issue of gap results. 

The notion of a generalized LCL, used in \cite{LaurentDiamCheckers}, just asks that the labeling can be verified by a $O(1)$ round LOCAL algorithm. That is, if the labeling is accepted by this constant round algorithm at every node, then the labeling is correct. We call such problems \emph{locally verifiable} problems in this work.

We give a clean construction of a family of \emph{locally verifiable} problems for unbounded-degree trees that have complexity $\Theta(n^r)$, for any real number $0<r\le1$. This recovers a weaker version of the result from \cite{LaurentDiamCheckers}, but is still enough to rule out many of the gaps in the complexity landscape. Since our construction directly deals with local problems, without involving complex machinery, it offers some insight into what causes these degenerate constructions. 

On some high level our construction just does the following: Define some very large number of neighborhoods such that each one can be next to each other, then order them arbitrarily by assigning them numbers 1, 2, ... 
Since we allow the degree to be arbitrarily large, there are much more than $n$ distinct radius 2 neighborhoods, even when each uses just a small number of nodes.

The problem now simply consists of checking whether these neighborhoods appear, adjacent to each other and in the correct order, starting with the first.

What determines the distributed complexity of these problems is that the last node in such a path must make sure, that the sequence of neighborhoods actually started with the correct neighborhood with number 1. The complexity is then the distance from the last node to the first node, or in other words the number of the last neighborhood in the chain.

By requiring nodes to only count up to a certain threshold, we can control the complexity of the problem and obtain any polynomial complexity we desire. Thus ruling out the existence of any gap results for such problems.

We think the main issue is the fact that our problem definition is able to make an infinite number of case distinctions, which leads us to the following question.

\begin{tcolorbox}
\textbf{Open Question.}
Is there a suitable (and natural) generalization of LCLs to unbounded-degree graphs for which a rich complexity landscape exists?
\end{tcolorbox}

We attempt to answer this question by giving a positive example of such a generalization in this work. 
We want to retain the essence of LCLs, that is we still just give a set of allowed configurations, for some fixed radius $r$. 
What is done in the bounded-degree setting is the following: The number of possible $r$-hop neighborhoods is finite, so we can just go through all of them and specify which ones we allow and which ones we do not allow. This is exactly the definition of LCLs.

However, in the unbounded-degree setting enumerating all possibilities necessarily means moving away from a finite description. In fact, writing down just a single configuration for each possible degree would result in an infinite number of such configurations.

In \cite{LaurentDiamCheckers} the authors define their generalized-LCLs in exactly this way. They just say that all possible labeled neighborhoods are classified into those that are allowed and those which are not allowed. While this captures all possible local labeling problems, it suffers exactly from the problem outlined above. Namely, that the infinitely many case distinctions that we can make will result in very unnatural problems as discussed above. 

Some other generalizations have been proposed in prior work. For example, the node-edge checkable framework, underlying the powerful round elimination technique, describes correctness via constraints on half-edge labels, and can be extended to unbounded degrees by specifying constraints for every possible degree \cite{brandt-2019-an-automatic-speedup-theorem-for}. This formalism can be viewed as a generalization, where nodes and edges have to only consider their radius 0.5 neighborhood to verify correctness. On Trees, this formulation is just as strong as the one above and so again suffers from the same problem\footnote{Although the node-edge formalism would also require an infinite number of outputlabels to be as expressive as generalized-LCLs.}.

Another approach appears in the study of binary labeling problems, where we are allowed only the labels 0 and 1. In \cite{24BinaryHighDegree} the authors look at these problems in the unbounded-degree setting, but restrict to a structurally simple class.

\subsection{Our Contributions}
A common feature of the above generalizations is that they allow an unbounded number of local configurations. 

Our first result shows that if such generalizations are permitted, then no meaningful polynomial gap results can exist: the complexity landscape becomes arbitrarily dense.

We do so by defining a family of problems that are locally checkable and that admit arbitrary polynomial complexities. Here, locally $4$-checkable refers to the fact that correctness can be verified by checking every node's 4-hop neighborhood. The solution is then globally correct if it is locally correct at every node.

\begin{restatable}{theorem}{PolyProblems}\label{thm:polyProblems}
    For any real $0 < r \le 1$, there exists a $4$-checkable problem using only labels $\{0,1\}$ with distributed complexity $\Theta(n^r)$ in unbounded-degree trees. 
\end{restatable}

This family of problems is defined by specifying an infinite number of $4$-hop neighborhoods that are allowed, much in the same way as LCLs, but no longer requiring the set of allowed configurations to be finite. 

Motivated by this observation, we introduce \emph{Locally Finite Labelings} (LFLs), which formalize the idea that \emph{every node must fall into one of finitely many local cases}. As in LCLs, an LFL is specified by a finite set of labeled $r$-hop configurations. However, unlike LCLs, the $r$-hop neighborhood of a node $v$ is no longer required to be identical (up to isomorphisms) to a configuration. Instead, we mark each edge as either \emph{optional} or \emph{required}. Required edges must appear exactly once in $v$'s neighborhood, while optional edges may appear arbitrarily often or not at all.

\begin{figure}[!ht]
    \centering
    \includesvg[width=0.5\linewidth]{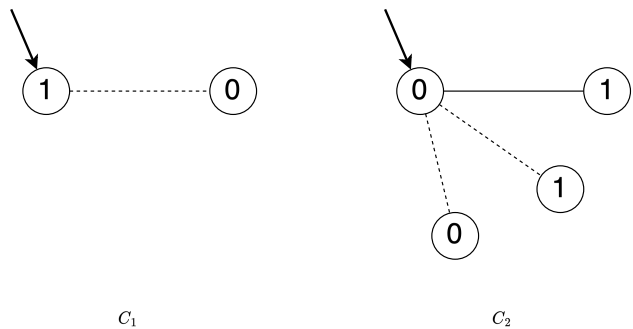}
    \caption{These two configurations encode the problem of finding an MIS. Optional edges are illustrated using dashed lines and required edges are solid lines. Any node that outputs 1 cannot have another neighbor that outputs 1, but arbitrarily many neighbors that output 0. On the other side any node that outputs 0 must have at least one neighbor that outputs 1 and can have arbitrarily many more. Therefore this correctly encodes the problem of finding an MIS.}
    \label{fig:MISIntroduction}
\end{figure}

As a first example we give the problem of finding a Maximal Independent Set, in the LFL formalism. We require nodes that are in the MIS to output 1 and nodes not in the MIS to output 0. There are now only 2 configurations\footnote{Defining an LCL technically requires writing up all possible neighborhoods, a very tedious task already for the simple problem of MIS. In contrast the description as an LFL captures exactly the idea behind MIS in just two simple configurations.}  required to define this problem as an LFL as can be seen in \Cref{fig:MISIntroduction}.

LFLs generalize LCLs to unbounded degree graphs while preserving a finiteness condition that rules out the constructions from \Cref{thm:polyProblems}. 
Perhaps surprisingly, this finiteness restriction already suffices to show that LFLs adhere to the same polynomial gap results as LCLs in the unbounded-degree setting.

\begin{restatable}{theorem}{GapResults}\label{thm:GapResults}
    Let $\Pi$ be an LFL on trees. Then one of the following two statements is true:
    \begin{itemize}
        \item There exists a positive integer $k$, such that $\Pi$ has complexity $\Theta(n^{1/k})$ in the deterministic LOCAL model.
        \item $\Pi$ can be solved in $O(\log n)$ rounds in the deterministic LOCAL model.
    \end{itemize}
    Furthermore, which case applies can be computed solely based on the description of $\Pi$.
\end{restatable}

While the actual proof of these gap results for LFLs is in large parts based on the framework of the original results for LCLs, we believe that one of the central insights of this work is that the key question is not
\begin{center}
    \emph{``What is the correct generalization of LCLs?''}
\end{center}

but rather
\begin{center}
    \emph{``Which features of a problem definition give rise to interesting phenomena such as gap results?''}
\end{center}

From this perspective, \Cref{thm:polyProblems} establishes a fundamental barrier: any generalization expressive enough to encode the artificial problems constructed there necessarily admits arbitrarily dense complexity landscapes, and thus cannot support polynomial gap results.

On the other hand, LFLs together with \Cref{thm:GapResults} demonstrate how much expressive power can still be allowed while preserving strong gap phenomena. In this sense, LFLs sit at the boundary identified by \Cref{thm:polyProblems}. Taken together, these results aim to characterize when and how polynomial gap results emerge on trees.

We do not view LFLs as the unique or \emph{correct} generalization of LCLs to unbounded-degree graphs — indeed, they deliberately exclude natural problems such as $(\Delta+1)$-coloring. 
Rather, LFLs provide the right abstraction for understanding polynomial-time gap results on trees. 
As such, they form an important step toward a principled theory of distributed complexity on unbounded-degree trees, and we hope that this inspires new ways to tackle the systematic study of the lower $O(\log n)$ regime.

\subsection{Our Contributions in more Detail}
A striking feature of locally checkable problems is that purely local constraints can enforce global structure. The combinatorial explosion introduced by unbounded-degrees can easily lead to degenerate behavior.
We demonstrate this with a simple, insightful construction.

Let $\mathcal{T}_x$ denote the set of all rooted trees on exactly $x \in \mathbb{N}$ nodes and height~$2$. For each $x$, fix an arbitrary enumeration $(T_i)_{i \in [|\mathcal{T}_x|]}$ of the trees in $\mathcal{T}_x$.

Consider an isolated path $P = (v_1, \dots, v_L)$ and attach to each node $v_i$ the tree $T_i$, with $v_i$ serving as the root of $T_i$. We refer to such a path as an $x$-\emph{gadget path}. Both the topology and the correct order of trees can be enforced by local constraints.

We turn this construction into a labeling problem by requiring nodes that belong to an $x$-gadget path to output~$1$, and all other nodes to output~$0$. The problem is defined so that a $x$-gadget path is valid only if it starts with $T_1$. Consequently, the node at the opposite end of the path can output~$1$ if and only if the entire prefix $(T_1, T_2, \dots)$ is present. This forces any distributed algorithm to spend time linear in the length of the path in order to determine the correct output.

By requiring that only the first $L(x)$ nodes along the path output~$1$, we can control the maximum length of any gadget path in an instance with $n$ nodes. Choosing $L(x)$ appropriately then suffices to prove our theorem.

Instead of describing more details, we want to highlight the fact that all we are doing is making a lot of case distinctions (in a more or less arbitrary way). We then abuse these case distinctions by using them to count the length of a path. Our generalization will try to avoid exactly this behavior.

We defer the full technical details to \Cref{sec:partition-gadgets}.

\paragraph*{How we generalize LCLs:}
Based on the insights above, we attempt to define a generalization of LCLs to unbounded degrees, having the following goals in mind:
\begin{enumerate}
    \item Restricted to bounded-degree graphs, we should recover the class of LCLs (or at least have them as a subset of our problems).\label{req:subsetLCL}
    \item We want to ensure that the problem definition allows only for a finite number of configurations that describe the problem. \label{req:finite}
    \item \label{req:expressive} We want a generalization that captures as many problems of interest as possible.
    \item We want the formalism to feel natural and be easy to work with.\label{req:simple}
\end{enumerate}

Since what we are aiming for is a generalization of LCLs, Goal~\ref{req:subsetLCL} is only natural. To ensure that we do not run into the sort of degenerate behavior described above, we have Goal~\ref{req:finite}. Note that one immediate consequence of this second goal is that our problem description has to be independent of $\Delta$.

With Goal~\ref{req:expressive} we ensure that results about LFLs are interesting, since they apply to a large number of problems that we do actually care about. However, there is a fine line to be walked here: We also want the class of problems to have a clean and (somewhat) easy definition, so that we can reason about them. We capture this in Goal~\ref{req:simple}.

In the spirit of this trade-off between Goals~\ref{req:expressive} and \ref{req:simple}, we have omitted some possible extensions to LFLs. For example we could prove the same gap results even while allowing for constraints of the form "A node must have the same number of red neighbors as blue neighbors". Incorporating these sorts of constraints into our definition would strengthen our results, but at the same time it would make the definitions harder to understand and clutter the proofs with heavy notation and case distinctions. We believe that our notion of LFLs strikes a nice balance.

With the above goals in mind, we propose the class of Locally Finite Labelings (LFLs). An LFL $\Pi = (\Sigma_{\mathrm{in}}, \Sigma_{\mathrm{out}}, r, \mathcal{C})$ consists of an input alphabet, an output alphabet, a checking radius $r \in \mathbb{N}$, and a finite set of configurations $\mathcal{C}$.  

The key difference from LCLs is the notion of local correctness. In an LCL, the $r$-hop neighborhood of a node $v$ must be \emph{isomorphic} to one of the labeled $r$-hop neighborhoods in $\mathcal{C}$. In contrast, in an LFL each edge in a configuration is marked as either \emph{required} or \emph{optional}. A node $v$ proves local correctness by providing a graph \emph{homomorphism} from its $r$-hop neighborhood to a configuration such that each required edge is mapped into exactly once. Optional edges may be mapped into arbitrarily often (potentially 0 times).

As an additional example beyond MIS, we show how $3$-coloring can be expressed as an LFL.

\paragraph*{3-coloring}
The problem of 3-coloring can now be described in a more intuitive way, namely that $3COL = (\emptyset, \set{R,G,B}, 1, \C)$ with $\C$ containing three configurations depicted in \Cref{fig:3COL}. 

\begin{figure}[!ht]
    \centering
    \includegraphics[width=0.7\linewidth]{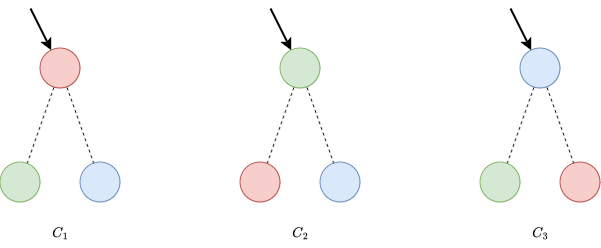}
    \caption{These three configurations suffice to encode the problem of computing a 3-coloring in unbounded-degree graphs. The arrows identify the center nodes and the edges are dashed to identify them as optional edges (there are no required edges for this problem).}
    \label{fig:3COL}
\end{figure}

\paragraph*{Recovering the gap results}

Our goal is to show that the polynomial gap results known for LCLs on bounded-degree trees extend to LFLs on trees of unbounded degree. The overall strategy closely follows the existing proofs for LCLs, but two conceptual obstacles must be addressed.

At a high level, the classical arguments break at two points: First, they rely on defining equivalence classes of subgraphs based on how a labeling on this subgraph interacts with the labeling of the remaining graph. In bounded-degree trees, these equivalence classes are finite because only finitely many local $r$-hop boundary neighborhoods can arise. In the unbounded-degree setting this finiteness breaks down, since there are infinitely many $r$-hop neighborhoods. The main task is to restore this finiteness in a principled way.

The second major problem is that we want to iterate all trees of a constant height\footnote{Intuitively a constant height is sufficient, because there is only a finite number of equivalence classes.}. This is no longer possible if the degrees are unbounded, as there are infinitely many trees of bounded height. We are tasked with defining a finite procedure that analyzes an infinite number of possible trees.

\paragraph*{Step 1: Finitely many Equivalence Classes via the node--edge checkable formalism}

The first step is to move from the standard neighborhood-based definition of LFLs to a node--edge checkable variant.
As it turns out, our notion of LFLs has a very natural translation into the world of node-edge checkable problems.

Conceptually, this reduction to the node-edge checkable setting corresponds to reducing the checking radius from an arbitrary constant \( r \) to a ``radius'' of \( 1/2 \), where correctness is verified solely by inspecting individual nodes and edges. This allows us to define the equivalence classes on subgraphs based on the edges connecting the subgraphs to the remaining graph. Since all considered subgraphs have at most two such edges, and the endpoint of such an edge is in one of finitely many cases, this recovers our finite equivalence classes.

For LCLs on bounded-degree trees, the equivalence between these two formalisms follows from a short and well-known argument. For LFLs, the situation is more delicate: although the set of configurations is finite, the set of admissible neighborhoods is infinite due to unbounded degrees, and correctness is witnessed by homomorphisms rather than isomorphisms.

We show that, nevertheless, for trees (and more generally for graphs of sufficiently large girth), every LFL can be encoded as a node--edge checkable LFL and vice versa. The key technical ingredient is a radius-reduction argument that allows us to focus on well-behaved 1-hop configurations. Repeating this reduction yields an equivalent LFL of radius one, which can then be turned node-edge checkable with a modified version of the original proof.

This establishes the following equivalence.

\begin{restatable}{theorem}{EquivNodeEdge}
    For every LFL $\Pi = (\Sin, \Sout, r, \C)$ in trees, there exists a node-edge checkable LFL $\Pi' = (\Sin', \Sout', \C_V, \C_E)$, such that:
    \begin{enumerate}[(1)]
        \item Any solution $(\sigma, F)$ to $\Pi$ can be turned to a valid solution for $\Pi'$ in $O(1)$ rounds of the LOCAL model.
        \item Any solution $\sigma'$ for $\Pi'$ can be turned into a valid solution $(\sigma, F)$ for $\Pi$ in $O(1)$ rounds of the LOCAL model.
    \end{enumerate}
\end{restatable}

Besides being essential for our gap results, this connection makes LFLs compatible with the powerful round elimination technique \cite{brandt-2019-an-automatic-speedup-theorem-for}. This opens up possibilities for future work, for example the gap between $\omega(1)$ and $o(\log^* n)$ for bounded-degree trees is based on round elimination. Does the same gap apply to LFLs? What is the dependence on $\Delta$ in that gap result?

\paragraph*{Step 2: Iterating unbounded-degree trees}

With the node--edge checkable formalism in place, we can handle the equivalence-class arguments underlying the polynomial gap results. However, a second obstacle is that the classical construction iterates over all trees of bounded height, which is no longer possible with unbounded-degrees. We resolve this by introducing \emph{virtual trees}: finite objects that represent entire families of trees with arbitrary degrees but identical local behavior. We show that only finitely many virtual trees need to be considered, and that each can be realized by an actual tree of bounded (though possibly large) degree.

The remaining parts of the gap results do not have major changes, but contain enough technical differences to warrant a full write-up. We provide a fully self-contained version of the proof, but restrict to just the deterministic case for brevity. Note that the result can be extended to randomized algorithms using the ideas from \cite[Lemma 15]{chang20}.

Overall, our contribution is not to introduce fundamentally new techniques, but to identify how far the polynomial gap results can be pushed beyond the bounded-degree setting. \Cref{thm:polyProblems} illustrates how allowing for an infinite number of case distinctions leads to degenerate behavior. Surprisingly, the fact that the gap results hold for LFLs suggests that this is indeed the defining characteristic of the polynomial gap results, which hold for a large class of problems, even in the unbounded-degree setting.

\subsection{A Tour of the Gap Results}
The proof of the gap results is quite long and technical, so we provide an overview of the main ideas.

\paragraph*{The node-edge checkable formalism}
In the node-edge formalism labels are put on halfedges, that is we cut every edge $e = \set{u,v}$ into two separate pieces $e_v = \set{e,v}$ and $e_u = \set{e,u}$. We then describe our problem by defining which combinations of labels are allowed on the halfedges around a node $\E_v = \set{e_v \mid e \in E : v\in e}$ and by defining which combinations of labels are allowed to be present on the two halves of an edge simultaneously. 

For LCLs in bounded-degree trees, there is a beautiful half-page argument, that any LCL can be expressed in the node-edge checkable formalism and vice versa. The key challenge here is that, in an LCL, correctness depends on the entire $r$-hop neighborhood of a node, while in the case of the node-edge checkable formalism the correctness must be checked by looking only at adjacent halfedges. Informally, the checking radius in the node-edge checkable formalism is $0.5$. 

The main reason why we can reduce the radius $r$ of an LCL to essentially $0.5$ is the fact that if we know what configuration $C$ a node $v$ uses, we know exactly how the neighborhood of $v$ looks like, as it must be isomorphic to $C$.

So we can have as an outputset for our radius $0.5$ LCL simply the set consisting of all possible configurations $C$, together with a sense of direction (each halfedge around $v$ outputs which edge of $C$ it represents). For now, assume the halfedges around $v$ can only output $C$ if the neighborhood around $v$ is isomorphic to $C$. Then, we can check on the constraints of the edges whether or not the two neighboring configurations are locally consistent.

Initially, it seems that extending this approach to LFLs is straightforward, there is still just a finite set of allowed configurations $\C$. However, crucially, the $r$-hop neighborhood of a node $v$ is no longer isomorphic to one of these configurations $C \in \C$. 

Consider again a node $v$ with a neighbor $u$, from $v$'s perspective, it is no longer sufficient to know that $u$ uses some configuration $C_u \in \C$. $C_u$ might include optional neighbors with some output label $x$, such that if such a neighbor is present in the neighborhood of $u$, then the $r$-hop neighborhood of $v$ is not valid. But if such a neighbor with output $x$ is not present, then the $r$-hop neighborhood of $v$ is valid. To distinguish this case, just knowing the fact that $u$ uses configuration $C_u$, is not sufficient. 

On a high level the solution to this problem is to have two versions of $C_u$, one version where a neighbor with label $x$ is guaranteed to exist and one version where such a neighbor does not exist. However, doing this proves to be highly nontrivial, mostly because maintaining compatibility and correctness of these new versions of our LFL becomes hard if the radius is larger than $1$. We navigate around this issue, by giving a reduction that reduces the radius of our LFL by just 1. This allows us to mostly focus on \emph{nicely behaved} subgraphs of our configurations that have a radius of 1. We then just repeatedly apply this reduction to obtain an LFL of radius 1, at which point we can apply the original argument to this LFL of radius 1 with only minor modifications.

\paragraph*{Reducing the radius by 1}

Intuitively, LFLs with radius~1 are easy to work with, because correctness depends only on \emph{how often each kind of labeling can appear on neighbors}, where a labeling refers to an input--output tuple. This is information that can be tracked easily.

Consider a configuration $C$ of radius $r>1$. Inside this configuration, consider all nodes at distance exactly $r-1$ from the center; we call these nodes \emph{twigs}, since all nodes attached to them lie at distance $r$ and are therefore leaves of~$C$. Our goal is to reduce the checking radius from $r$ to $r-1$. The idea is to let each twig output a certificate asserting that its $1$-hop neighborhood matches its $1$-hop neighborhood in~$C$. Again 1-hop configurations are (reasonably) easy to work with. Once these certificates are present, the information contained within distance $r-1$ of the center already suffices to ensure correctness of the entire $r$-hop neighborhood.

Now consider a node $v$ in an actual instance of our LFL, and let $A$ be the set of nodes at distance exactly $r-1$ from~$v$. Each node $u \in A$ matches some configuration $C_u$, as can be seen in \Cref{fig:radiusRed} on the left. In each such configuration, the node $v$ is mapped to a twig (since it lies at distance $r-1$ from~$u$). Consequently, $v$ must provide a certificate that simultaneously certifies that it plays the role of a twig in all configurations $(C_u)_{u \in A}$. 

\begin{figure}[!ht]
    \centering
    \includesvg[width=\linewidth]{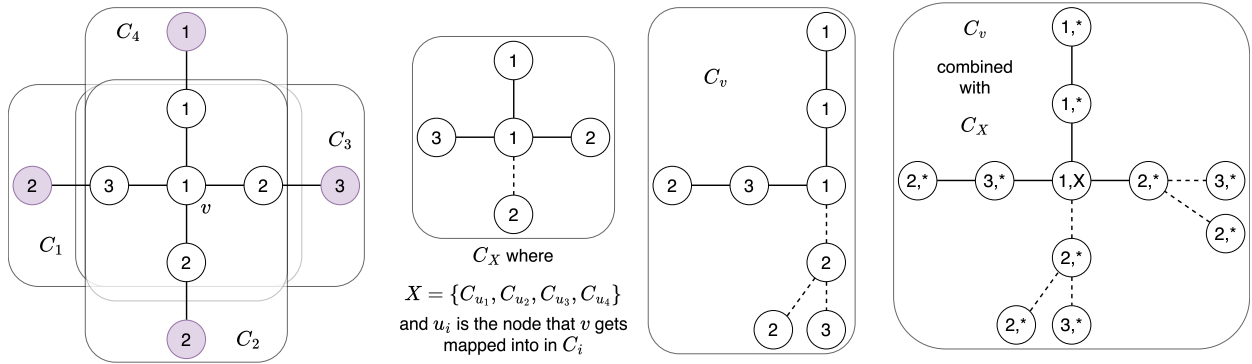}
    \caption{Illustration of the radius-reduction construction. 
The leftmost panel shows an instance of our LFL, with the set $A$ of nodes at distance $r\!-\!1$ from $v$ highlighted in purple. 
The node $v$ appears as a twig $u_i$ in each of the four configurations $C_1, C_2, C_3, C_4$. 
The second panel depicts the resulting configuration $C_X$ of the twigLFL, obtained by combining the four twig configurations $C_{u_1}, C_{u_2}, C_{u_3}, C_{u_4}$. 
The rightmost panel shows how $C_X$ is combined with the original configuration $C_v$ to obtain a version of the initial LFL augmented with certificates. 
In the final configuration, each node outputs both its original label and a certificate encoding the twig configurations it can match (a placeholder '*' is put everywhere, except the center node). 
In particular, the node $v$ outputs the set $X$, certifying that it can match all twig configurations $C_{u_1}, C_{u_2}, C_{u_3}, C_{u_4}$ \emph{simultaneously}, while still correctly solving the original problem.}
    \label{fig:radiusRed}
\end{figure}

We formalize this requirement as a separate problem, called the \emph{twigLFL}. For any twig $u$ in a configuration $C$, let $C_u$ denote the induced $1$-hop configuration of~$u$, with the property that if a node $v$ matches $C_u$, then its $1$-hop neighborhood can be correctly mapped to the $1$-hop neighborhood of $u$ in~$C$. Let $\mathcal{T}$ denote the set of all such twig configurations. In the twigLFL, each node outputs a set $X \subseteq \mathcal{T}$ of twig configurations such that its $1$-hop neighborhood can be matched to every configuration in~$X$. That is, $v$ simultaneously certifies that it is the correct twig for all nodes in $A$.
We provide an illustration of the setup in \Cref{fig:radiusRed}, second image from the left.

The final step is to combine the twigLFL with the original LFL that can be seen in \Cref{fig:radiusRed} on the right. Concretely, we modify the original LFL so that a node $v$ must simultaneously match one of the original configurations $C_v$ and a configuration of the twigLFL $C_X$. For every compatible pair of configurations $(C_v, C_X)$, we construct a combined configuration, in which each node outputs both an original output and a certificate in the form of a set of twig configurations. The resulting LFL can then be checked using only information from the $(r-1)$-hop neighborhood, by making sure that the nodes at distance $r-1$ output the correct certificates.

\paragraph*{The gap results}
The key idea behind the gap results is to understand how the topology of a tree (together with the input labels) restricts the set of output labels that may appear on a halfedge. Consider cutting a tree at some edge $e$, thereby splitting it into two parts. Let $T_1$ be one of the resulting components, still incident to a single halfedge $e_1$. We ask the following question: \emph{which output labels can be assigned to $e_1$ such that the labeling of $T_1$ can be completed to a valid solution?} The answer is a subset of the finite output alphabet $\Sout$. Hence, each rooted tree with a distinguished halfedge induces one of finitely many possible label sets, which we call the \emph{Type} of a tree.

Because LFLs in the node--edge checkable formalism have an effective checking radius of $0.5$, we can classify trees solely by the subsets of labels they permit on their boundary halfedges.

Our first goal is to compute all such subsets that may arise in any valid instance. In the bounded-degree setting, this can be done by enumerating all trees up to some constant height. In the unbounded-degree setting, this approach fails due to infiniteness. To overcome this, we introduce \emph{virtual trees} and show that it suffices to enumerate only finitely many of them to capture all possible boundary behaviors.

We then extend this idea from single edges to paths. Instead of cutting the tree at one edge, we cut out a path $P$ by separating it from the rest of the tree at two edges. Let $H$ denote the induced subgraph consisting of $P$ and all components attached to it, and let $e_s$ and $e_t$ be the two remaining boundary halfedges of $H$. This is depicted in \Cref{fig:padding}. We now ask: \emph{which pairs of labels can be assigned to $e_s$ and $e_t$ such that the labeling of $H$ can be completed?} Going one step further, we ask whether there exist sets $X_s, X_t \subseteq \Sout$ such that we can draw labels for $e_s$ from $X_s$ and $e_t$ from $X_t$ independently and can then still obtain a valid labeling of $H$ for any two choices.

\begin{figure}[!ht]
    \centering
    \includesvg[width=.7\linewidth]{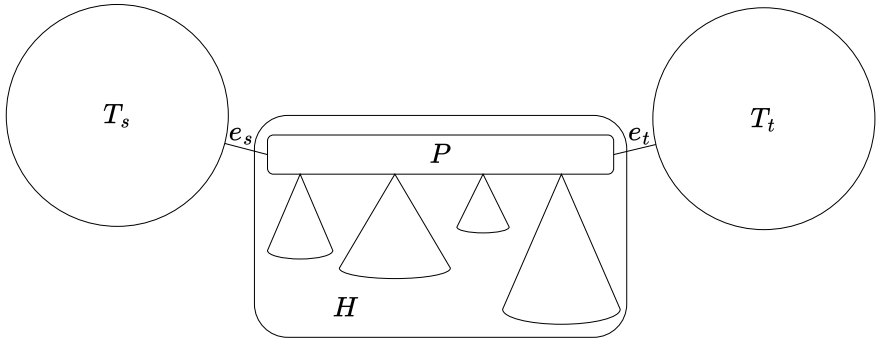}
    \caption{By cutting out a path $P$ from a tree, we split the tree into three parts $T_s, H,T_t$. The parts are connected by the edges $e_s$ and $e_t$. If we could pick labels for $e_s$ and $e_t$ without coordination, we could find labelings for $T_s$ and $T_t$ without the need for communication.}
    \label{fig:padding}
\end{figure}

We illustrate the idea in \Cref{fig:padding}, cutting $H$ out of the tree yields three components: $T_s$, $H$, and $T_t$. If such sets $X_s$ and $X_t$ exist, and if every label in $X_s$ (resp., $X_t$) can be extended to a valid labeling of $T_s$ (resp., $T_t$), then $H$ acts as a buffer between $T_s$ and $T_t$. In this case, the two sides can be solved independently, and any two solutions can be combined via $H$. This eliminates the need for global coordination across the cut and is the core mechanism behind fast algorithms in the polynomial regime.

We say that $(X_s, X_t)$ is an \emph{independent class} for $H$, since the choices of labels from $X_s$ and $X_t$ can be made independently. We study the existence of such independent classes using the notion of \emph{types} of trees introduced earlier. As before, the finiteness of the output alphabet $\Sout$ implies that there are only finitely many possible types into which such subgraphs $H$ can fall.

Using classical automata-theoretic arguments together with a pumping-style lemma for these subgraphs, we show that if $H$ is sufficiently long, then there exists another subgraph $H'$ with the same type—and hence the same admissible labelings on its boundary—but of arbitrary length. In particular, $H'$ can be chosen long enough for our purposes.

We can use these longer versions together with a fast algorithm $\A$ in the following way. To simulate $\A$ on $H$, we first replace it by a pumped version $H'$ with identical boundary behavior, chosen long enough so that when simulating $\A$ at the midpoint of $H'$, the algorithm cannot see outside of $H'$. Since $\A$ must be correct, the labeling it produces in the middle of $H'$ induces a non-empty independent class for $H'$, and by extension for $H$. To see this, note that $\A$ has to make its decisions without seeing outside of the path, so whatever output $\A$ produces must work no matter which labelings are present outside of the path.
In this way, fast algorithms give rise to independent classes.

The crucial question is how often we can use subgraphs like $H$ for padding. Suppose we have cut $T$ into $T_s,H,T_t$, what we really want to do is apply the same decomposition recursively inside $T_s$, cutting out another subgraph $H'$. If we also had an independent class for $H'$, we could reduce the amount of coordination even more. However, this is not as easy as before. The available labels on $T_s$ are not only constrained by the topology of $T_s$, but also on the boundary to $H_s$, where we must respect $X_s$. So the possible outputlabels that the nodes in $T_s$ may use are restricted by the fact that the halfedge $e_s$ must use an outputlabel from the previous independent class $X_s$. Note that for $X_s$ and $X_t$ to give the strong guarantee that we can label $e_s$ and $e_t$ independently it might have been necessary to severely restrict the available labels in $X_s,X_t$.

So if we try to find an independent class for some $H'$ in $T_s$, it is not clear that we can do so while \emph{respecting the restrictions imposed on $e_s$}.
The number of times we can find such independent classes even after having restricted our outputsets is what determines the polynomial complexity of the problem: if it can be done once, the complexity is $\Theta(\sqrt{n})$; if it can be done $k$ times, we obtain complexity $\Theta(n^{1/(k+1)})$.

The final step of our gap result is to show that if a problem admits a deterministic algorithm with complexity $o(n^{1/k})$, then this recursive decoupling must indeed be possible $k$ times. So we can still obtain good independent classes, even when we have restricted the possible labelings $k-1$ times already.

When using an algorithm to obtain an independent class as described above we pumped up a path so that the algorithm cannot see outside. If the algorithm has complexity $o(n^{1/k})$, then we needed to pump the path to length $\Omega(n^{1/k})$ with a suitably small constant. 

Consider a path $P'$ that we now want to find an independent class for, but where the nodes may potentially already be restricted.
We can turn this into an instance, where every node of the path  already is attached with a pumped path of length $\Omega(n^{1/k})$ each. When simulating our algorithm later, the outputs given by $\A$ on these paths will impose exactly the restrictions of the independent classes. So when we pump $P'$ to length $\Omega(n^{1/k})$, we will need $\Omega(n^{2/k})$ many nodes total. 

When we run our algorithm on this instance of roughly $n^{2/k}$ nodes, we still guarantee that the algorithm cannot see outside the paths. As a result we also obtain an independent class for $P'$ in the same way as before, but with potentially even more severe restrictions. If we did this construction for the third time, we would get an instance of size $~n^{3/k}$ nodes, $\ldots$. 
Hence, an algorithm with complexity $o(n^{1/k})$ can be used to repeat the above process $k$ times. By exploiting the fact that this recursive decoupling (as in \Cref{fig:padding}) can be done $k$ times, we obtain an algorithm with complexity $O(n^{1/(k+1)})$.
As a result there cannot be a problem with complexity in $\omega(n^{1/(k+1)}) \cap o(n^{1/k})$, proving the desired gap.

\section{Preliminaries}
We rely mostly on classical notation surrounding graphs, but to avoid confusion we state the following definitions explicitly. Formally, in a graph $G=(V,E)$ we fix the following notations:
\begin{itemize}
    \item For any two nodes $u,v \in V$ we denote the distance $d(u,v)$ as the length of the shortest path between these two nodes. We set $d(v,v) = 0$.
    \item For any edge $e \in E$ and any node $v \in V$, we define their distance $d(v,e) := 1+ \min\set{d(v,u)}_{u \in e}$.
    \item The neighborhood $\Nh(v)$ of $v$ is the subgraph using
    \begin{align*}
    \text{nodes }\set{u \in V \:|\: \text{dist}(u,v) \leq 1} \text{ and edges }\set{e \in E \:|\: \text{dist}(v, e) \leq 1}
    \end{align*}
    Similarly, for any $r \in \N$ we define the $r$-hop neighborhood $\Nh_r(v)$ of $v$ as the subgraph using
    \begin{align*}
    \text{nodes }\set{u \in V \:|\: \text{dist}(u,v) \leq r} \text{ and edges }\set{e \in E \:|\: \text{dist}(v, e) \leq r}
    \end{align*}
    Note that these definitions are different from the subgraph induced by the set\\ $\set{u \in V \:|\: \text{dist}(u,v) \leq 1}$. Let $u,w \in V$ be two neighbors of $v$, then if the edge $\set{u,w}$ exists, it will not be included in $\Nh(v)$.
    \item For any $r \in \N$ a centered radius $r$ ball is a connected graph $B=(V,E,\cen)$ with $\cen \in V$ being the center node such that for all nodes $v \in V$ it holds that $d(\cen,v) \leq r$ and similarly for all edges $e \in E$ it holds that $d(\cen, e) \leq r$.
    \item We denote the set of all half edges of a graph as $\hE := \set{(e,u), (e,v) \mid \forall \set{u,v} = e \in E}$. We will refer to these half edges mostly using $e_u=(\set{u,v},u)=(e,u)$ and $e_v = (\set{u,v},v) =(e,v)$
    \item For every node $v \in V$ we denote the set of adjacent halfedges as \\$\hE_v = \set{e' \in \hE \mid e' = (e,v) \text{ for } e\in E}$
    \item For two graphs $G=(V, E)$ and $G' = (V',E')$, $f:V \rightarrow V'$ is a graph homomorphism if for every edge $\set{u,w} \in E$ it holds that $\set{f(u), f(w)} \in E'$. Additionally $f$ is a graph isomorphism if it also is bijective and $f^{-1}$ is also a graph homomorphism. With abuse of notation, we additionally define $f(\set{u,v}) = \set{f(u),f(v)}$, which means that for $e\in E$, it holds that $f(e) \in E'$.
\end{itemize}

The class of Locally Checkable Labelings was first introduced by Naor and Stockmeyer \cite{NaorStockmeyer95} in their seminal work on local computation. It has since seen a vast amount of attention \cite{balliu19lcl-decidability,B0COSS22_LCLregularTrees,lcls_on_grids,chang20,CKP19exponential,CP19timeHierarchy,brandt21trees}. Intuitively to define an LCL, we just fix a set of labels and categorise all possible neighborhoods as either allowed or forbidden. We give a formal definition of the description in \cite{NaorStockmeyer95}.

\paragraph*{LCLs:} A Locally Checkable Labeling $\Pi = (\Sin, \Sout, r, \C)$ is a tuple consisting of the following:
\begin{itemize}
    \item A set of input labels $\Sin$.
    \item A set of output labels $\Sout$.
    \item The checkability radius $r \in \N$.
    \item A set of allowed configurations $\C$.
\end{itemize}

Where each configuration $C \in \C$ is a tuple $C = (\V, \E, \cen, \mu)$ describing a labeled centered ball $(\V, \E, \cen)$ of radius $r$ around center node $\cen$. Where the labels are defined by $\mu: \V \rightarrow \Sin \times \Sout$. 

An instance is given by a labeled graph $G= (V, E, \phi)$ where $\phi: V \rightarrow \Sin$ assigns each node an input label. 

A solution to $\Pi$ consists of a labeling $\sigma : V \to \Sout$ together with a family of isomorphisms 
\[
F = \{ f_v : \Nh_r(v) \to C \text{ ,where } C \in \mathcal{C} \}_{v \in V},
\]
These isomorphisms must map $v$ to the designated center $\cen$ of the chosen configuration, and they must preserve labels: for all $u \in \Nh_r(v)$ we require 
\[
(\phi(u), \sigma(u)) = \mu(f_v(u)).
\]

This completes the definition.  

\paragraph*{The LOCAL model:} We prove our results for the deterministic LOCAL model. Given a graph $G=(V,E)$, we view each node as a processor and each edge as a communication link. Nodes are assigned a unique id from the set $\set{1, \ldots, n^c}$ for some constant $c$. Initially nodes are only aware of their id, the exact\footnote{A linear upperbound in $n$ is sufficient (but necessary).}  number of nodes $n$ and its inputlabel. All nodes execute the same algorithm and communication proceeds in synchronous rounds. In each round nodes transmit messages of arbitrary size to all of their neighbors, receive messages and perform arbitrary deterministic computation on the available information. Note that because of the limited initial information and the communication in synchronous rounds, after $i$ rounds each node knows only the information in its local $i$-hop neighborhood, hence the name LOCAL model. Eventually each node must terminate and produce an output. The running time, or complexity of an algorithm is then defined as the maximum number of rounds any node requires to fix its output. An interesting observation to keep in mind is that a $T(n)$ round algorithm can be viewed as a (potentially very complicated) mapping from $T(n)$-hop neighborhoods to outputlabels.

\section{Locally Finite Labelings} \label{sec:LFLhE}
With the intuition explained in the introduction, we present the fully formal definition of Locally Finite Labelings.

A Locally Finite Labeling problem $\Pi=(\Sin, \Sout, r, \C)$ is a tuple consisting of the following:
\begin{itemize}
    \item $\Sin$ a finite set of input labels
    \item $\Sout$ a finite set of output labels
    \item $r \in \N$ a positive integer number called the radius
    \item $\C$ a set of $r$-hop configurations defining the constraints of the problem.
\end{itemize}

So far the definition is the same as for LCLs, but configurations for LFLs are (mildly) more complicated, requiring each edge to be labeled with one of \emph{required}, or \emph{optional}.

\paragraph*{Configuration}
An $r$-hop configuration $C=(\V,\E, \cen, \mu, \tau)$ is a tuple consisting of the following:
\begin{itemize}
    \item $(\V, \E, \cen)$ is a centered radius $r$ ball. ($\E \subset \set{\set{u,v} \:| \: u,v \in \V }$ allowing self loops)
    \item $\mu : \V \rightarrow \Sin \times \Sout$ is a map assigning every node of $C$ both an input and an output label.
    \item $\tau: \E \rightarrow \set{\text{required}, \text{optional}}$ is a map declaring every edge in $C$ either \emph{required}, or \emph{optional}.
\end{itemize}

An instance of $\Pi$ is then given by a labeled graph $G = (V,E, \phi)$ with $\phi: V \rightarrow \Sin$ giving every node an input label. A solution to $\Pi$ is given by a tuple $(\sigma, F)$ consisting of:
\begin{itemize}
    \item An output assignment $\sigma: V \rightarrow \Sout$ assigns every node an output label
    \item A family of graph homomorphisms $F = \set{f_v}_{v \in V}$ with $f_v: \Nh_r(v) \rightarrow \V$ for some configuration $C = (\V,\E, \cen, \mu, \tau) \in \C$.
\end{itemize}

Such a solution is correct \emph{if and only if} for every node $v \in V$ the corresponding map $f_v: \Nh_r(v) \rightarrow C$ \textbf{matches the configuration} $C$ under $\sigma$. This notion of \emph{matching} is the main difference to LCLs, where a graph isomorphism is required.

\paragraph*{Matching a configuration}
$f_v$ matches $C = (\V,\E, \cen, \mu, \tau)$ under $\sigma$ \emph{if and  only if} the following conditions are met:
\begin{enumerate}
    \item \textbf{Correctly centered}: $f_v(v) = \cen$, so $v$ must be mapped to the center node of the configuration.\label{rule:centered}
    \item \textbf{Labelings are respected}: For all $u \in \Nh_r(v)$ it holds that $(\phi(u),\sigma(u)) = \mu(f_v(u))$, so both inputs and outputs must match.\label{rule:labelings}
    \item \textbf{Requirements are satisfied:} Each required edge in the configuration exists exactly once in the neighborhood of $v$, or formally: Let $R = (\V_r, \E_r)$ be the subgraph induced by required edges, with $\E_r = \set{e \in \E \mid \tau(e) = \text{required}}$ and $\V_r = \bigcup_{e \in \E_r} e$. Then $f_v$ restricted to $f_v^{-1}(\V_r)$ is a graph isomorphism $f_v: f_v^{-1}(\V_r) \rightarrow \V_r$. \label{rule:required}
    \item \textbf{All edges are accounted for:} $f_v$ is a graph homomorphism, or in words: every edge in $v$'s neighborhood must be correctly mapped onto an edge existing in the configuration. \label{rule:homomorphism}
\end{enumerate}

Examples of MIS and 3-coloring as LFLs were already given in the introduction, but we additionally provide the problem of 3-cycle detection, as a more involved example in \Cref{apx:examples}.

\medskip

It is easy to see that our construction of LFLs satisfies Goal~\ref{req:subsetLCL}, that is on bounded-degree graphs, we recover LCLs.

\begin{lemma}
    On bounded-degree graphs, the set of problems that can be expressed as LCLs is the same as the set of problems that can be expressed as LFLs.
\end{lemma}
\begin{proof}
    Going from LCL to LFL, the only real challenge is to switch between the \emph{matching} conditions of LCLs and LFLs. LCLs require a mapping $f_v$ that is a graph isomorphism while LFLs only require that $f_v$ restricted to the required edges and nodes is an isomorphism. However, by simply using the same configurations and setting all edges to be required (so using the constant function $\tau(e) = \text{ required}$ for all $e \in \E$) the problem definitions become identical. 
    
    To go from LFL to LCL, we abuse the fact that we are in a bounded-degree graph. So we can enumerate all possible $r$-hop neighborhoods. We check for each such neighborhood, if it can be matched to one of the configurations of our LFL, if yes we add it to our LCL. Again the resulting problems will be identical.
\end{proof}

Additionally we provide one nontrivial extension to LFLs that we will use to make our proofs much more elegant. We call this extension partially ordered LFLs ($\poLFL$). In partially ordered LFLs we allow wildcard labels in the configurations. For example, we could then describe the problem of coloring using configurations where the center node has color e.g. \emph{Red} and it has a single optional neighbor with label \emph{Not Red}. Now any neighbor that does not output color \emph{Red} can be mapped onto this optional neighbor. We illustrate how to encode 3 coloring in just 3 intuitive configurations this way in \Cref{fig:3COLpOIntro}. This formalism is not just nice in the sense, that it gives us much more freedom to express problems in a concise and intuitive way, but it also makes some of our more involved proofs more manageable.

\begin{figure}[!ht]
    \centering
    \includesvg[width=0.7\linewidth]{img/3ColPO.svg}
    \caption{Using the ''not \_\_'' labels ($\{\bar{R}, \bar{G}, \bar{B}\}$) we can describe 3 coloring in the very intuitive way, where a node that outputs $R$ must have only neighbors that are ''not $R\,$''.}
    \label{fig:3COLpOIntro}
\end{figure}

In \Cref{sec:LFLpO} we give a formal definition and prove that partially ordered LFLs are equivalent to LFLs.

\subsection{LFLs in Trees - the Node-Edge Checkable Formalism}\label{sec:ReReduction}

To extend the gap results we have for LCLs in bounded-degree trees, we will rely on an alternative way of expressing LFLs. We extend the classical node--edge checkable formalism, known from the Round-Elimination technique, in a way that allows us to prove an equivalence between LFLs in trees and this enhanced formalism. 
The advantage of the node--edge checkable view is that it represents an LFL with an effective checking radius of $1/2$,
which makes it considerably cleaner to reason about the infinitely many possible $r$-hop neighborhoods that can appear in the unbounded-degree setting.

\paragraph*{Definition: LFLs -- node--edge checkable}
A Locally Finite Labeling in the node--edge checkable formalism is a tuple
\[
    \Pi = (\Sin, \Sout, \C_V, \C_E),
\]
where 
\begin{itemize}
    \item $\Sin$ is a set of input labels.
    \item $\Sout$ is a set of output labels, and we define $\Sout^* = \{a, a^* \mid a \in \Sout\}$. 
    Thus for every output label $a$ we include a designated ``optional'' version $a^*$.
    \item $\C_V$ is a set of allowed node configurations, each a multiset with elements from $\Sin \times \Sout^*$.
    \item $\C_E$ is a set of allowed edge configurations, each a multiset of the form $\{a,b\}$ with $a,b \in \Sin \times \Sout$.
\end{itemize}

An instance of $\Pi$ is a half-edge labeled graph $G=(V,E,\phi)$, and a solution is a mapping $\sigma:\hE \rightarrow \Sout$ satisfying all node and edge constraints. Formally:
\begin{itemize}
    \item For every node $v \in V$ with incident halfedges $\hE_v$, the multiset
    \[
        \La_v = \biguplus_{e_v \in \hE_v} \{(\phi(e_v),\,\sigma(e_v))\}
    \]
    must match some configuration $C \in \C_V$ in the sense that for every $(a,b)\in \Sin\times\Sout$ the multiplicity $m_{(a,b)}(\La_v)$ satisfies:
    \[
        m_{(a,b)}(\La_v)=m_{(a,b)}(C) \quad\text{or}\quad
        \bigl( m_{(a,b)}(\La_v) > m_{(a,b)}(C)\ \text{and}\ (a,b^*)\in C \bigr).
    \]
    \item For every edge $e=\{u,v\}\in E$, the multiset 
    \[
        \{(\phi(e_u),\sigma(e_u)),\,(\phi(e_v),\sigma(e_v))\}
    \]
    must belong to $\C_E$.
\end{itemize}

Intuitively, each node $v$ inspects the multiset of labeled halfedges around it and checks whether some node constraint matches. 
If a pair $(a,b)$ appears exactly as often as required in a configuration $C$, the constraint is satisfied; if it appears \emph{more} often, this is permitted only when $C$ contains $(a,b^*)$, which serves as the analog to an optional edge, allowing arbitrarily many occurrences.

\paragraph*{MIS as a node-edge checkable problem}
Using no input labels and the outputlabels $\set{M,O,P}$. The problem is given in the following way:
\begin{center}
\begin{tabular}{||c | c||} 
 \hline
 $\C_V$ & $\C_E$  \\ [0.5ex] 
 \hline\hline
 $M^*$ & $MM$   \\ 
 $PP^*O^*$ & $OO$  \\
  & $PM$\\
 \hline
\end{tabular}
\end{center}

The idea is that nodes in the MIS output label $M$ on all incident halfedges, while nodes not in the MIS must \emph{point} to at least one MIS neighbor via the label $P$. 
By the edge constraints, any halfedge labeled $P$ must face an $M$ on the opposite side, so the pointed-to neighbor indeed belongs to the MIS. 
Edges between two non-MIS nodes use label $O$ on both halfedges. 
Thus a non-MIS node may have arbitrarily many non-MIS neighbors using $O$ but must have at least one MIS neighbor using $P$, ensuring correctness of the MIS specification.

Again this presentation of MIS is much more intuitive, than explicitly writing down all of these combinations for each possible $\Delta$.

\paragraph*{Equivalence in trees}
Balliu et al.\ have shown that, in trees, the node--edge checkable formalism is equivalent to the standard definition of LCLs \cite{smallMessagesbcmos21}. This equivalence enables the use of the powerful round-elimination technique when reasoning about LCLs.

We obtain an analogous result for LFLs, but adapting their proof is not straightforward. Informally, their construction proceeds as follows. Given an LCL
\[
    \Pi = (\Sin, \Sout, r, \C),
\]
they define a node-edge checkable version
\[
    \Pi' = (\Sin, \Sout', \C_V, \C_E),
\]
where the output alphabet is
\[
    \Sout' := \{(C,e) \mid C \in \C,\ e \in \bar{\E}_C\}.
\]
Intuitively, each half-edge outputs exactly which half-edge of which configuration of $\Pi$ it represents. The node constraints $\C_V$ require all halfedges incident to a node $v$ to name the same configuration $C$, and jointly represent all halfedges adjacent to the center node of $C$. The edge constraints $\C_E$ require the halfedges $e_u, e_v$ of any edge $\{u,v\}$ to output labels $(C, e'_u)$ and $(C', e'_v)$ that are \emph{locally consistent}: the $(r-1)$-view of $u$ must be compatible with the configuration $C$ that $v$ selected.

This is illustrated in \Cref{fig:locallyConsistent}. There, $v$ chooses configuration $C_1$ and $u$ chooses configuration $C_2$. In the underlying instance, the $1$-hop view of $u$ is consistent both with the right-hand neighborhood of the center of $C_1$ and with the center node of $C_2$. Thus, if the halfedges $e_v,e_u$ output $(C_1,e')$ and $(C_2,e^*)$, respectively, the pair is locally consistent from $v$'s perspective. In essence, because $u$ identifies $v$ as its left neighbor in $C_2$, $v$ only needs to check $u$'s output to ensure that everything to the right of $v$ remains consistent with its own choice.

\begin{figure}[!ht]
    \centering
    \includesvg[width=0.8\linewidth]{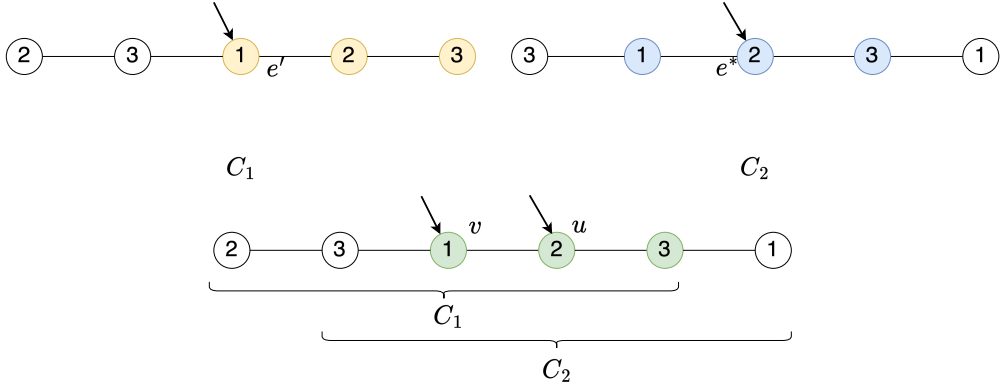}
    \caption{Illustrating local consistency. Because $e^*$ outputs $C_2$ as part of its output, $v$ can be sure that its neighborhood is consistent with $C_1$, without checking everything itself.}
    \label{fig:locallyConsistent}
\end{figure}

For LFLs this approach is insufficient. The presence of optional edges introduces additional ambiguity. Consider an LFL with two configurations as in \Cref{fig:inconsistentLFL}. The bottom path does not encode a valid solution: the neighborhood of $v$ does not match $C_1$. However, the $(r-1)$-view required for $v$ to justify matching $C_1$ also appears somewhere inside $C_2$. The node $u$ can therefore legitimately match $C_2$, even if doing so yields an output that is incompatible with $v$'s choice of $C_1$. Crucially, there are now \emph{multiple} ways to extend a partial neighborhood into a full configuration, and so the simple promise that a half-edge outputs $(C_2, e^*)$ is no longer strong enough to guarantee local consistency for $v$.

\begin{figure}[!ht]
    \centering
    \includesvg[width=0.8\linewidth]{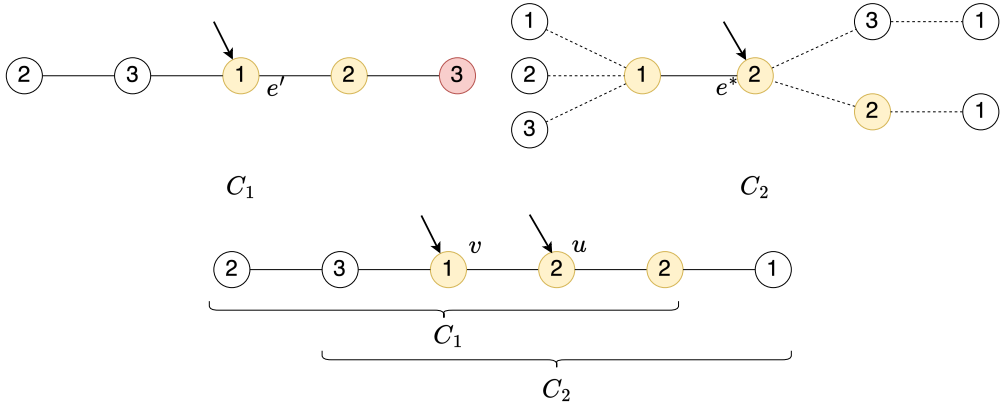}
    \caption{If $v$ tries to match $C_1$, but only knows that $u$ matches $C_2$, then $v$ will not see the error.}
    \label{fig:inconsistentLFL}
\end{figure}

One way to resolve this issue is to expand the configuration set $\C$ by introducing \emph{strengthened} versions of certain configurations. In the example above, we could add a modified configuration $C_2'$ that remains compatible with $C_2$ but also provides a sufficiently strong guarantee for $v$ to verify correctness. The strengthened configuration is shown in \Cref{fig:strengthened}. Any node that can match $C_2'$ could also have matched $C_2$, so this transformation does not make the problem harder. Instead, it introduces redundant but more informative configurations that allow us to encode the LFL correctly in the node-edge checkable formalism.

\begin{figure}[!ht]
    \centering
    \includegraphics[width=0.4\linewidth]{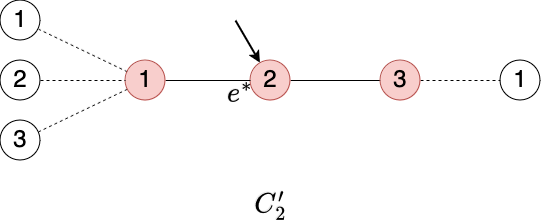}
    \caption{The new configuration $C_2'$ that $u$ can use to convince $v$ that the output is consistent. The nodes in red are the subgraph that is required in order for $v$ in \Cref{fig:inconsistentLFL} to be convinced of the correctness of the solution.}
    \label{fig:strengthened}
\end{figure}
However, while the underlying idea of the construction is essentially the same, the technical details become substantially more involved when $r>1$. Working directly with configurations of larger radius is cumbersome: handling all possible local cases and ensuring correctness and compatibility throughout turns out to be highly nontrivial. Instead, we employ an alternative approach where we first reduce the radius of a given LFL to 1.  We will now show that this is sufficient by proving that the construction of \cite{smallMessagesbcmos21} can be adapted---with minor modifications---to work for LFLs of radius~1.

\begin{lemma}\label{lem:rIsOne}
    Let $\Pi = (\Sout, \Sin, 1, \C)$ be an LFL. Then there exists a node--edge checkable LFL $\Pi' = (\Sout', \Sin', \C_V, \C_E)$ such that:
    \begin{enumerate}
        \item Any solution $(\sigma, F)$ to $\Pi$ can be transformed into a valid solution for $\Pi'$ in one round of the LOCAL model.
        \item Any solution $\sigma'$ to $\Pi'$ can be transformed into a valid solution $(\sigma, F)$ for $\Pi$ in one round of the LOCAL model.
    \end{enumerate}
\end{lemma}

\begin{proof}
If $r=1$, then each configuration of $\Pi$ is a star. Thus the correctness of the labeling at a node $v$ depends only on the labels it sees on its neighbors.  
To express such a problem in the node--edge checkable formalism, we must assign labels to halfedges. In order for $v$ to verify its correctness, every half-edge should encode both the input/output pair of $v$ and the input/output pair of the corresponding neighbor.  

Therefore we use the output alphabet
\[
    \Sout' = (\Sin \times \Sout)^2.
\]
These are ordered pairs: if a half-edge adjacent to $v$ outputs $(x_1,x_2)$, this means that $v$ has input/output label $x_1 \in \Sin \times \Sout$ and its neighbor has input/output label $x_2$.

\medskip
\noindent
For each configuration $C$, we map each neighbor $u$ of the center $\cen$ to a label:
\[
    f_C(u) = 
    \begin{cases}
        (x_{\cen}, x_u), & \tau(e)=\mathrm{required}, \\[4pt]
        (x_{\cen}, x_u)^*, & \tau(e)=\mathrm{optional},
    \end{cases}
\]
where $e$ is the half-edge between $\cen$ and $u$, and $x_{\cen}, x_u \in \Sin \times \Sout$ are their input/output labels in $C$.

Using this map, each configuration $C$ of $\Pi$ yields a node configuration (which is a multiset)
\[
    C_V = \{\, f_C(u) \mid u \in V_C \setminus \{\cen\} \,\}.
\]
Notice that given a valid solution to $\Pi$, we have every node $v$ assign its incident halfedges, leading to some node $u$, the labels $f(u)$. Then in this solution $v$ can use configuration $C_V$. 

We set
\[
    \C_V = \{\, C_V \mid C \in \C \,\}.
\]

\medskip
\noindent
Finally, the edge constraint is simply that if one half-edge is labeled $(x_1,x_2)$, then the opposite half-edge must be labeled $(x_2,x_1)$.

\paragraph*{Equivalence of problems:}

To see that the two problems are equivalent, observe that the only restriction placed on a node $v$ is the number of neighbors it has with each label $x \in \Sin \times \Sout$.  
A configuration $C$ in the original LFL specifies exactly how many neighbors of the center node $\cen$ may or must have label $x$.  
If $\cen$ has $k$ required neighbors with label $x$, then the multiset $C_V$ contains the element $(x_{\cen}, x)$ exactly $k$ times.  
If $\cen$ has an optional neighbor with label $x$ (and therefore can have an arbitrary number of additional such neighbors), then $C_V$ contains the element $(x_{\cen}, x)^*$.  
Thus $C_V$ encodes precisely the constraints imposed by $C$, and therefore the two formalisms define equivalent problems.
\end{proof}

Now we turn our attention to the radius reduction result. As in the original proof, our goal is to have certain nodes provide certificates guaranteeing correctness. However, reducing the radius all the way down to one in a single step would require handling highly complex certificates. Instead, we reduce the radius by exactly one, which allows us to work with significantly simpler radius-1 objects.

Fix some configuration $C$ of radius $r$. We want to derive from it a configuration of radius $r-1$. The core idea is to focus on the nodes at distance $r-1$ from the center of $C$; we call such nodes \emph{twigs}. Each twig node must output a certificate asserting that its neighborhood is compatible with $C$. To build such certificates, we use \emph{twig configurations} $C_T$, which have radius one and capture exactly what is required of a twig in $C$.

In a real instance of our problem, a node $u$ will typically serve as a twig in many different configurations simultaneously. Indeed, if we consider the set $D_{r-1}(u)$ of nodes at distance exactly $r-1$ from $u$, then for each $v \in D_{r-1}(u)$ there is some configuration $C_v$ whose center is $v$, and $u$ is a twig in $C_v$. Therefore, $u$ must certify that it is the correct twig for \emph{all} of these configurations at once.

To handle this, we introduce the \emph{TwigLFL}, in which every node outputs precisely the set of twig configurations it is compatible with. This TwigLFL serves as the foundation for the radius-$1$ reduction step, and will be the key tool in our overall radius reduction procedure.

This is the last definition that works for graphs with cycles.
\begin{definition}[subconfiguration]\label{def:subconfig}
    For any $r$-hop configuration $C = (\V,\E,\cen,\mu,\tau)$ a subconfiguration $C'=(\V', \E', \cen', \mu', \tau')$ is any connected subgraph $(\V',\E') \subset (\V, \E)$ of $C$ with some designated center node $\cen' \in \V'$. We always have that $\mu'$ and $\tau'$ are the restrictions of $\mu$ and $\tau$ to the subset $\V' \subset \V$, $\E' \subset \E$ respectively.
\end{definition}

From now on let $\Pi = (\Sin, \Sout, r, \C)$ be an LFL for graphs of girth at least $2r + 1$. As a result all configurations $C \in \C$ are acyclic and therefore are trees. Whenever we talk about the parent of a node for a configuration, we implicitly assume the configuration to be rooted at the center node.

\begin{definition}[twig-configuration]
    Let $\Pi = (\Sin, \Sout, r, \C)$ be an LFL and $C=(\V, \E, \cen, \mu, \tau)$ be an $r$-hop configuration.
    Let $u \in \V$ be any twig in $C$ and $p \in \V$ its parent.
    We define the twig configuration \[C_u = (\Nh_1(u), \E', u, \mu', \tau')\] as the subconfiguration implied by $\Nh_1(u)$.
    Additionally, in the twig configuration we set $\tau(\set{u,p}) = required$.
    We define $\T^\Pi$ as the set of all twig-configurations of $\Pi$.
\end{definition}

As discussed above, a node in a real solution may simultaneously match many different twig configurations. We now formalize this through the construction of the TwigLFL.

\subsection{The TwigLFL}

For any LFL $\Pi = (\Sin, \Sout, r, \C)$, the TwigLFL is an LFL
\[
\Pi^T = (\Sin, \Sout^T, 1, \C^T),
\]
which we will construct over the course of this section. We begin by specifying the first three components of its definition:
\begin{enumerate}
    \item $\Sin$ is unchanged.
    \item $\Sout^T = \Sout \times \Pow(\T^\Pi)$; that is, each node outputs a normal output label together with a subset of twig configurations indicating those it is compatible with.
    \item The radius is $r = 1$.
\end{enumerate}

The main challenge lies in defining the configuration set $\C^T$. The guiding intuition is that each node should output a set $X \in \Pow(\T^\Pi)$ of twig configurations such that it could, in principle, match \emph{all} configurations in $X$. Crucially, we must also ensure that if a node outputs the set $X$, then it is indeed capable of matching every twig configuration in $X$.

We realize this requirement at the level of configurations as follows. For each subset $X \in \Pow(\T^\Pi)$ of twig configurations, we attempt to construct a 1-hop configuration $C_X$ that represents the ``intersection'' of all configurations in $X$. (For some choices of $X$, this intersection may be empty, in which case no configuration $C_X$ will be added to~$\C^T$.)

Since $C_X$ is a radius-$1$ configuration, we construct it by taking a center node $\cen$ and adding required and optional neighbors. For this, we inspect every input/output labeling of a node, appearing in any configuration in $X$, and add to $C_X$ as many copies of such labeled edges as are demanded by the combined constraints. In this way, $C_X$ exactly captures the compatibility conditions encoded by all twig configurations in $X$.

Let $u$ be any non-center node appearing in any configuration in $X$, that is, $u \in C$ for some $C \in X \in \Pow(\T^\Pi)$.  
For each such node, we consider its input/output labeling $(x_u,y_u) = \ell$ and record how often this labeling appears in each configuration of $X$, together with whether any of these occurrences are optional.

\begin{definition}[Multiplicity of a labeling]\label{def:multiplicity}
    For any input/output labeled radius-$1$ ball 
    \[
        B = (V, E, \mu)
    \]
    centered at a node $v$, we define the \emph{multiplicity} $m_\ell(B) \in \N$ to be the number of neighbors of $v$ whose labeling is $\ell \in \Sin \times \Sout$.

    For any configuration 
    \[
        C = (V, E, \cen, \mu, \tau),
    \]
    we define $m_\ell(C)$ as the number of \emph{required} edges adjacent to $\cen$ whose endpoint label is $\ell$.
\end{definition}

The following observations serve as guidelines for how we construct the configuration $C_X$.

\begin{observation}\label{obs:optExists}
    Let $C_1, C_2 \in X$. If $m_\ell(C_1) < m_\ell(C_2)$, then either $C_1$ contains an optional edge with labeling $\ell$, or no neighborhood can match both $C_1$ and $C_2$ simultaneously.
\end{observation}

\begin{proof}
    Suppose $m_\ell(C_1) < m_\ell(C_2)$ and $C_1$ contains no optional edge with label $\ell$.  
    Then every neighborhood matching $C_1$ must contain \emph{exactly} $m_\ell(C_1)$ nodes labeled $\ell$, while any neighborhood matching $C_2$ must contain at least $m_\ell(C_2)$ nodes labeled~$\ell$.  
    Thus no single neighborhood can satisfy both constraints.
\end{proof}

As a consequence, if there exist $C_1, C_2 \in X$ such that no neighborhood can match both simultaneously, then no configuration $C_X$ exists for this~$X$.  
In this case we simply ignore such an $X$ and do not introduce a configuration for it.  
Henceforth we assume that we are never in this bad case.  
Under this assumption, we make the following key observation.

\begin{observation}\label{obs:nrOccLabel}
    Any neighborhood that can match all configurations in $X$ simultaneously must contain at least
    \[
        m_\ell(X) := \max_{C \in X} m_\ell(C)
    \]
    non-center nodes that are labeled with~$\ell$.

    Moreover, if at least one configuration $C \in X$ with $m_\ell(C) = m_\ell(X)$ does \emph{not} not have an optional edge leading to a node with label~$\ell$, then every neighborhood that matches all configurations in $X$ must contain \emph{exactly} $m_\ell(X)$ non-center nodes labeled with $\ell$.
\end{observation}

Based on this observation, we add exactly $m_\ell(X)$ \emph{required} edges leading to nodes with labeling $\ell$ to $C_X$.  
Additionally, we add an \emph{optional} edge leading to a node labeled~$\ell$ if and only if \emph{every} configuration $C \in X$ with $m_\ell(C) = m_\ell(X)$ contains such an optional edge.  
By \Cref{obs:optExists}, this condition implies that in fact \emph{all} configurations in $X$ contain such an optional edge.

By repeating this process for every labeling $\ell$ that appears in any configuration of~$X$, we obtain the configuration~$C_X$.  
We now show that this construction behaves as intended.

\begin{lemma}\label{lem:CxIsIntersection}
    Any neighborhood that can match all configurations in $X$ can match $C_X$, and any neighborhood that can match $C_X$ can match all configurations in $X$.
\end{lemma}

\begin{proof}
    Both directions follow the same structure, so we only prove the first.

    Let $N$ be any radius--1 input/output labeled neighborhood that matches all configurations in $X$.  
    We show that $N$ can be matched to $C_X$.

    Fix any labeling~$\ell$ that appears on a neighbor of the center of $N$.  
    Since $N$ matches every configuration $C \in X$, it must contain at least $m_\ell(X)$ nodes labeled~$\ell$.  
    We map exactly $m_\ell(X)$ of these nodes to the $m_\ell(X)$ required edges of $C_X$ with label~$\ell$.

    If $N$ contains more than $m_\ell(X)$ such nodes, then by \Cref{obs:optExists} every configuration in $X$ must contain an optional edge labeled~$\ell$, and hence $C_X$ contains exactly one such optional edge.  
    We map all remaining $\ell$-labeled nodes of $N$ to this optional edge.

    Repeating this procedure for each labeling~$\ell$ appearing in $N$ assigns all neighbors of $N$ to nodes of $C_X$.

    It remains to check that all requirements of $C_X$ are satisfied.  
    Let $e_r$ be any required edge of $C_X$ with labeling~$\ell$.  
    By construction of~$C_X$, there exists a configuration $C \in X$ with a required edge labeled~$\ell$, which implies that any neighborhood matching all of $X$ contains at least $m_\ell(X)$ nodes labeled~$\ell$.  
    Hence $N$ has enough $\ell$-labeled nodes to satisfy the requirement for~$e_r$.

    This completes the proof.
\end{proof}

We now obtain a configuration for the TwigLFL by assigning appropriate output labels\footnote{Recall that 
$\Sout^T = \Sout \times \Pow(\T^\Pi)$, so each node outputs both a normal output label and a subset of twig configurations.} 
from $\Sout^T$ to the nodes of~$C_X$.  
From the construction above, the components $\V_X$, $\E_X$, $\cen_X$, and $\tau_X$ are already fixed; the only remaining task is to define the labeling function $\mu_X$.

We assign labels as follows.  
The center node receives
\[
    \mu_X(\cen_X) = (x_{\cen}, (y_{\cen}, X)),
\]
where $(x_{\cen},y_{\cen})$ is the input/output labeling of $\cen_X$ determined during the construction of $C_X$.  
Every other node $v \in \V_X \setminus \{\cen_X\}$ receives
\[
    \mu_X(v) = (x_v, (y_v, *)),
\]
where $(x_v,y_v)$ is again its input/output label from the construction of $C_X$, and where the symbol $*$ denotes that \emph{any} subset $Y \subseteq \T^\Pi$ is permitted.  In other words, non-center nodes enforce the correct input/output labeling, but impose no restriction on the set of twig configurations they output. In \Cref{sec:LFLpO} we formally prove, that using such wildcard labels does not change the class of LFL problems.

We may now complete the definition of the TwigLFL by specifying its set of configurations:
\[
    \C^T := \{\, C_X \mid X \in \Pow(\T^\Pi) \,\}.
\]

From \Cref{lem:CxIsIntersection}, and using the fact that nodes may perform unbounded local computation, we obtain the following corollary.

\begin{corollary}
    Given a solution to $\Pi$, we can construct a solution to $\Pi^T$ in $O(1)$ rounds in the LOCAL model.
\end{corollary}

\subsection{An LFL with Radius \texorpdfstring{$r-1$}{r-1}}
In the next step, we want to combine the TwigLFL with the original LFL~$\Pi$.
By \emph{combine} we mean that each node is required to output a tuple of labels that simultaneously satisfies the constraints of both problems.
It is immediate that such a conjunction of tasks is again an LCL whenever the two starting problems are LCLs---we simply combine isomorphic neighborhoods and pair up outputs for each possible graph isomorphism between these neighborhoods.
One can quickly convince themselves that something similar should be possible for LFLs as well, but the details are somewhat unclear.

For this reason, we present an explicit construction that correctly combines the TwigLFL with the original LFL~$\Pi$, ensuring that the resulting specification is again an LFL.

\paragraph*{Combining configurations:}
Let $\Pi = (\Sin, \Sout, r, \C)$ be any LFL, $X \subset \T^\Pi$ any subset of twig configurations such that $C_X$ exists, and let $C = (\V, \E, \cen, \mu, \tau) \in \C$ be any configuration.
We construct a set of configurations $C^X$ such that any node that can match both $C_X$ and $C$ can match at least one configuration in $C^X$. \\
Again, we will construct these new configurations $C' \in C^X$ manually.

To this end, we first build a 1-hop neighborhood $C'$ around the center in a way that is compatible with $C_X$, and then extend $C'$ by attaching appropriate subtrees of~$C$.
To formalize this, we introduce \emph{configbranches}.

\begin{definition}[configbranch]\label{def:configBranch}
    For any $r$-hop configuration $C = (\V, \E, \cen, \mu, \tau)$ and any node $u \in \V \setminus \{\cen\}$, we define the \emph{configbranch} $T_u$ as the subtree rooted at $u$ when viewing $C$ as a tree rooted at $\cen$.
\end{definition}

For every node $v$ in $C_X$, we again define $\ell_v \in \Sin \times \Sout$ as the input/output pair of $v$ in $C_X$ (ignoring the part of the output corresponding to $X \subset \T^\Pi$).
We further define
\[
    \Nh_\ell(\cen) = \{\, u \in \Nh_1(\cen) \mid \mu(u) = \ell \,\},
\]
the set of neighbors with label~$\ell$, and
\[
    T_\ell = \{\, T_u \mid u \in \Nh_\ell(\cen) \,\},
\]
the set of all configbranches rooted at such neighbors.

We will see that there is not always a unique way to construct the combination of $C_X$ and $C$, so we keep track of all possible compatible constructions.
The collection $C^X$ is defined as the set of all these possible combined configurations.

We describe the construction from a bird's--eye view in \Cref{alg:computeCX}, 
where we suppress the details of how exactly extensions are performed for each 
labeling. Note that every time we extend, we potentially obtain multiple distinct 
extensions. For each such extension, we then continue extending with the next 
label, and so on.

\begin{algorithm2e}
\caption{Compute $C^X$}\label{alg:computeCX}
\KwIn{$C, C_X$}
\SetKwFunction{Extend}{extend}

$C' \gets$ initialize as just the center node $\cen'$ \\
$C^X \gets \{ C' \}$ \\

\For{each labeling $\ell$ that appears in $C_X$}{
    $B \gets \emptyset$ \Comment{temporary buffer}\\
    \For{each $C' \in C^X$}{
        $\mathcal{E} \gets$ \Extend{$C', C, C_X, \ell$} 
        \Comment{returns a set of extended neighborhoods.}\\
        $B \gets B \cup \mathcal{E}$ 
    }
    $C^X \gets B$ \Comment{update}
}
\end{algorithm2e}

\paragraph*{How the extending works:}
We start with a partially constructed configuration $C'$ whose center is
denoted by~$\cen'$.  
Since none of the already attached neighbors of $\cen'$ can have labeling~$\ell$,
we ignore all existing neighbors of~$\cen'$ for this step.  
We then add new $\ell$-labeled neighbors to~$\cen'$, 
extending each such neighbor by one of the configbranches in~$T_\ell$, as
described next.

\begin{enumerate}
    \item \textbf{Case:} $m_\ell(C) > m_\ell(C_X)$
    \begin{enumerate}
        \item $\cen_X$ \textbf{is not} adjacent to an optional node with label $\ell$.
        
        $\Rightarrow$ $C$ and $C_X$ are incompatible and so $C^X = \emptyset$. Abort the algorithm. \label{case:empty1}

        \item $\cen_X$ \textbf{is} adjacent to an optional node with label $\ell$.
        
        $\Rightarrow$ Add all neighbors $\Nh_\ell(\cen)$ and extend each $u \in \Nh_\ell(\cen)$ by its respective $T_u \in T_\ell$. \label{case:single1}
    \end{enumerate}
    \item \textbf{Case:} $m_\ell(C) = m_\ell(C_X)$
    \begin{enumerate}
        \item $\cen_X$ \textbf{is not} adjacent to an optional node with label $\ell$. 
        
        $\Rightarrow$ Only add a neighbor $u \in \Nh_\ell(\cen)$ if the edge $\set{\cen, u}$ is required. Again we extend each $u$ that we did add, with $T_u$. \label{case:single2}

        \item $\cen_X$ \textbf{is} adjacent to an optional node with label $\ell$.
        
        $\Rightarrow$ Add all neighbors $\Nh_\ell(\cen)$ and extend each $u \in \Nh_\ell(\cen)$ by its respective $T_u \in T_\ell$. \label{case:single3}
    \end{enumerate}
    \item If $m_\ell(C) < m_\ell(C_X)$ \label{case:tooLittle}
     \begin{enumerate}
        \item $\cen$ \textbf{is not} adjacent to an optional node with label $\ell$.
        
        $\Rightarrow$ $C$ and $C_X$ are incompatible and so $C^X = \emptyset$. Abort the algorithm.\label{case:empty2}

        \item $\cen$ \textbf{is} adjacent to an optional node with label $\ell$. \label{case:badCase}
        
        $\Rightarrow$ Create multiple different versions of $C'$. Handle this case separately.
    \end{enumerate}
\end{enumerate}
Note that in Cases~\ref{case:empty1}~and~\ref{case:empty2} we can abort the algorithm, as $C$ and $C_X$ are incompatible. In Cases~\ref{case:single1}, \ref{case:single2} and \ref{case:single3} we simply alter the config that we started with. It is only in the last Case~\ref{case:badCase} that we have to create multiple different configurations.

\paragraph*{The difficult case:}
If $m_\ell(C) < m_\ell(C_X)$ and $\cen_X$ is adjacent to an optional neighbor with
label~$\ell$, then we must ensure that $C'$ contains at least $m_\ell(C_X)$
required nodes labeled~$\ell$. To achieve this, we may use the optional neighbors present in~$C$.

We show a small example, suppose that $C$ has exactly two optional neighbors 
$u, v \in \Nh_\ell(\cen)$ and that 
\[
    m_\ell(C) = 2, \qquad m_\ell(C_X) = 4.
\]
Then we must add $d = m_\ell(C_X) - m_\ell(C) = 2$ additional required copies of
optional neighbors. There are several valid choices: we may add two copies of $u$
(extending each by $T_u$), or two copies of~$v$ (extending each by $T_v$), or one copy of each (extending
by $T_u$ and $T_v$, respectively). Each choice yields a different configuration,
but every resulting configuration has the property that any neighborhood matching
it can also match both $C_X$ and $C$. See \Cref{fig:combConfig,fig:resultingConfig} for an
illustration of this toy example.

In the general case, let
\[
    d := m_\ell(C_X) - m_\ell(C) > 0
\]
and let
\[
    A := \{\, u \in \Nh_\ell(\cen) \mid \tau(\{\cen, u\}) = \text{optional} \,\}
\]
be the set of optional $\ell$-neighbors of the center in~$C$.  
By the stars--and--bars theorem, there are
\[
    k := \binom{d + |A| - 1}{|A| - 1}
\]
distinct ways to choose $d$ additional required copies from the set~$A$.  
For each such choice, we construct a corresponding version of~$C'$, extending
every added copy of $u \in A$ with its configbranch $T_u$ as usual.

\begin{figure}[!ht]
    \centering
    \includesvg[width=.5\linewidth]{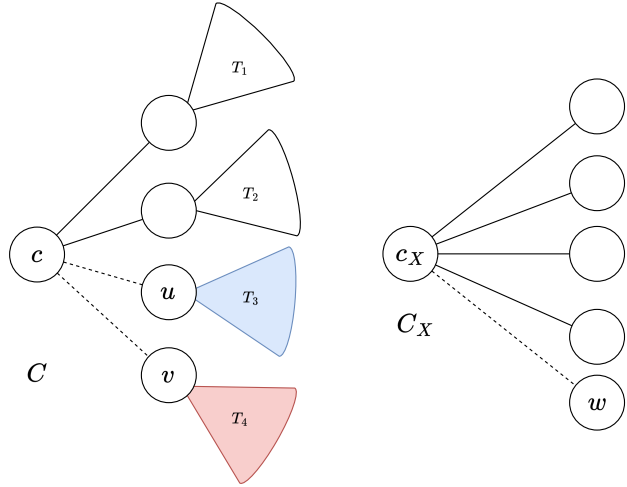}
    \caption{The case $m_\ell(C) < m_\ell(C_X)$ and $\cen$ is adjacent to 2 required neighbors with label $\ell$ and two optional neighbors $u,v$ both with label $\ell$.}
    \label{fig:combConfig}
\end{figure}

\begin{figure}[!ht]
    \centering
    \includesvg[width=\linewidth]{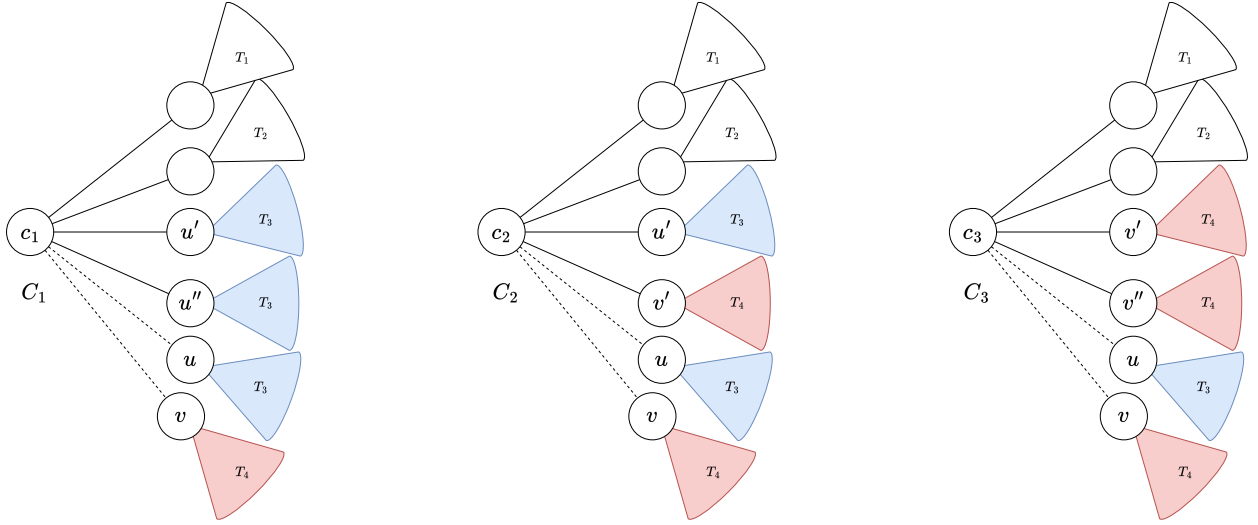}
    \caption{The resulting configurations from \Cref{fig:combConfig}, showing all possible combinations. They also contain optional copies of $u,v$ since $w$ is in the configuration $C_X$.}
    \label{fig:resultingConfig}
\end{figure}

Additionally if $\cen_X$ also has an adjacent optional edge with labeling $\ell$, then we also add optional copies of all nodes of $A$, extended by their respective configbranches, to $C'$. This finishes the description of the extend procedure.

\medskip

The construction is designed so that the following properties hold.
\begin{observation}\label{obs:cPrimeObs}
    If $C^X$ is non-empty, then for all $C' \in C^X$ and for every labeling $\ell$, the following statements are true:
    \begin{itemize}
        \item Every neighbor $u'$ of $\cen'$ is a copy of some neighbor $u$ of $\cen$, such that $T_u$ and $T_{u'}$ are isomorphic and identically labeled. Moreover, if $u$ is required, then there is exactly one copy $u'$ of~$u$.
        \item The center~$\cen'$ has an adjacent optional neighbor with label~$\ell$ if and only if both $\cen$ and $\cen_X$ do.
    \end{itemize}
\end{observation}

Using the same argument as in \Cref{lem:CxIsIntersection}, we obtain the following lemma.  
The only additional work compared to the construction of $C_X$ is that we have carefully attached the appropriate configbranches.

\begin{lemma}\label{lem:extendedConfigs}
    For any $C \in \C$ and $X \in \Pow(\T^\Pi)$ the following statements are true.
    \begin{itemize}
        \item For any $C' \in C^X$, constructed as above. If any labeled $r$-hop neighborhood $\Nh_r(v)$ around some node $v$, can match $C'$, then $\Nh_r(v)$ can match $C$ and $\Nh(v)$ can match $C_X$.
        \item If any labeled $r$-hop neighborhood $\Nh_r(v)$ around some node $v$, can match $C$ and $\Nh(v)$ can match $C_X$, then there exists some $C' \in C^X$, such that $\Nh_r(v)$ can match $C'$.
    \end{itemize}
\end{lemma}
\begin{proof}
    We prove the first implication; the converse is analogous.

    Let $\Nh_r(v)$ be any labeled $r$-hop neighborhood that matches some $C' \in C^X$,
    and let $f:\Nh_r(v)\to \V'$ be a witnessing matching of $C'=(\V',\E',\cen',\mu',\tau')$.
    First, by Observation~\ref{obs:cPrimeObs} and the same argument used in
    Lemma~\ref{lem:CxIsIntersection}, the neighborhood $\Nh(v)$ also matches $C_X$.

    It remains to show that $\Nh_r(v)$ matches $C=(\V,\E,\cen,\mu,\tau)$.  
    By Observation~\ref{obs:cPrimeObs}, every node of $C'$ is a copy of some node of $C$,
    and the corresponding configbranches are isomorphic. Hence there is a natural
    surjective map \(g:\V' \to \V\) that sends each copy in $C'$ to its original in $C$.
    Note that $g$ is a graph homomorphism (it preserves adjacency).

    Define $f' := g \circ f : \Nh_r(v) \to \V$. We claim that $f'$ is a correct matching
    of $\Nh_r(v)$ to $C$. We verify the matching conditions:
    \begin{enumerate}
        \item \emph{Centered.} By construction $f(v)=\cen'$, hence $f'(v)=g(\cen')=\cen$,
        so the center maps correctly.

        \item \emph{Labels respected.} The map $f$ respects node-labels of $C'$, and
        $g$ maps copies to their originals while preserving the node-labels assigned
        during the construction of $C'$. Therefore $f'$ respects the node-labels of $C$.

        \item \emph{Required-subgraph isomorphism.} View $C$ as a tree rooted at $\cen$.
        Let $e\in\E$ be any required edge of $C$, and let $u$ be the neighbor of $\cen$
        that lies on the unique path from $\cen$ toward $e$. Since $e$ is required,
        the edge $\{\cen,u\}$ is required as well (required edges form a connected
        subgraph containing the center). By Observation~\ref{obs:cPrimeObs} there is a
        unique copy $u'\in\V'$ of $u$ in $C'$, and the branch $T_{u'}$ contains a
        copy $e'\in\E'$ of $e$ which is required. Because $f$ matches $C'$, exactly
        one edge of $\Nh_r(v)$ is mapped to $e'$. Call that edge $e^*$. Then
        $f'(e^*) = g(f(e^*)) = g(e') = e$. Moreover, uniqueness of the copy $u'$ and
        the fact that $\{\cen',u'\}$ is required imply that $f'$ is injective on the
        required subgraph; hence $f'$ restricts to an isomorphism between the required
        subgraph of $\Nh_r(v)$ and that of $C$.
        
        \item \emph{Homomorphism.} Since both $f$ and $g$ are graph homomorphisms, so
        is their composition $f'$, and thus every edge of $\Nh_r(v)$ is mapped to an
        edge of $C$.
    \end{enumerate}

    This verifies that $f'$ correctly matches $\Nh_r(v)$ to $C$, completing the proof.
\end{proof}
The final step towards our radius $r-1$ LFL, is to convert each configuration in $C^X$ into an $(r-1)$-hop
configuration whose output labels guarantee correctness of the full $r$-hop
neighborhood.
To achieve this, we extend the output alphabet in the same manner as in the
TwigLFL: every node outputs a pair consisting of
\begin{enumerate}
    \item its ordinary output for $\Pi$, and
    \item a set $X \in \Pow(\T^\Pi)$ of twig-configurations that the node certifies it can match.
\end{enumerate}

Consider any configuration $C' \in C^X$, obtained from combining $C$ with
$C_X$.  
Let $u' \in C'$ be a twig node with associated\footnote{By associated, we mean the configuration attached to the original twig $u \in C$ of which $u'$ is a copy} twig-configuration $C_u$.  
We require that $u'$ outputs a set $X$ for which
$C_u \in X$.

Intuitively, this enforces that the $1$-hop neighborhood around $u'$ must
match the twig-configuration $C_u$.  
Since $u'$ is included in the $(r-1)$-hop neighborhood of the center, the center node sees that $u'$ outputs a set $X$ which contains $C_u$.
This information is precisely what is needed to determine---using only the
$(r-1)$-hop view---that the full $r$-hop neighborhood matches~$C'$.

In other words, by reading the set $X$ output by each twig, the verifier at
radius $(r-1)$ can verify that the missing outer layer (the $r$th hop) behaves exactly as
prescribed by $C'$.  
We now formalize this construction.

\paragraph*{A new labeling}
The final step in constructing our radius-$(r-1)$ LFL is to adjust the
labeling used on the configurations in $C^X$. We will again be using extra wildcard labels as before. A formal definition of LFLs with wildcard labels, together with an equivalence proof, can be found in \Cref{sec:LFLpO}.
We begin by extending the output alphabet in the same way as in the
TwigLFL:
\[
    \Sout^{(r-1)} := \Sout \times \Pow(\T^\Pi),
\]
so every node in an actual instance must output one of these labels.

For configurations, however, we will use a slightly richer set of
outputlabels.
Define
\[
    A := \{*\} \;\cup\; \{\,T^{+} \mid T \in \T^\Pi\,\},
\]
and set
\[
    \Sout^{*} := \Sout \times \bigl(\Pow(\T^\Pi) \cup A \bigr).
\]
Thus we introduce two kinds of ``variable'' outputlabels:
\begin{itemize}
    \item The wildcard label $*$, meaning that
    $(x,(y,Y)) \in \Sin \times \Sout^{(r-1)}$ is compatible with
    $(x,(y,*)) \in \Sin \times \Sout^{*}$ for all $x,y$.
    \item The labels $T^{+}$, indicating that the twig-configuration
    $T \in \T^\Pi$ must be included in the output set.
    Formally, $(x,(y,Y)) \in \Sin \times \Sout^{(r-1)}$ is compatible with
    $(x,(y,T^{+})) \in \Sin \times \Sout^{*}$ for all $x,y$
    if and only if $T \in Y$.
\end{itemize}

Now we describe how to assign the new outputlabels to the configurations.
For any $X \in \Pow(\T^\Pi)$ and any configuration $C \in \C$ for
which $C^X$ exists, we adjust the labeling on every
\[
    C' = (\V', \E', \cen', \mu', \tau') \in C^X
\]
as follows:
\begin{enumerate}
    \item The center receives the label
    \[
        \mu^{*}(\cen') = (x,(y,X)),
    \]
    where $\mu'(\cen') = (x,y)$ is its original label.
    \item For any twig node $u$ in $C'$ with original label
    $\mu'(u) = (x,y)$ and associated twig-configuration $T_u$, we set
    \[
        \mu^{*}(u) = (x,(y,T_u^{+})).
    \]
    \item For all remaining nodes $v$ with original label $\mu'(v) = (x,y)$,
    we assign the wildcard version
    \[
        \mu^{*}(v) = (x,(y,*)).
    \]
\end{enumerate}

We then define the full set of labeled configurations as
\[
    \C^{*}
    := \bigcup_{C \in \C} \;\bigcup_{X \in \Pow(\T^\Pi)}\;
       \bigcup_{C' \in C^{X}}
       \{\,(\V', \E', \cen', \mu^{*}, \tau')\,\}.
\]

Finally, we define the $(r-1)$--restriction of $\C^{*}$ as the set of all
$(r-1)$--hop subconfigurations of members of $\C^{*}$:
\[
    \C^{(r-1)}
    := \{\, C^{(r-1)} \mid C \in \C^{*} \,\}.
\]

With this, we have assembled all components needed to prove our main theorem.

\begin{theorem}\label{thm:reduceR}
    For any integer $r>1$ and any LFL $\Pi = (\Sin, \Sout, r, \C)$ in trees, there exists an LFL $\Pi^*=(\Sin, \Sout^{(r-1)}, r-1, \C^{(r-1)})$, such that:
    \begin{enumerate}[(1)]
        \item Any solution $(\sigma, F)$ to $\Pi$ can be turned to a valid solution for $\Pi^*$ in $r$ rounds of the LOCAL model. \label{enum:PiToPiStar}
        \item Any solution $(\sigma^*,F^*)$ for $\Pi^*$ can be turned into a valid solution $(\sigma, F)$ for $\Pi$ in $r$ rounds of the LOCAL model.\label{enum:PiStartToPi}
    \end{enumerate}
\end{theorem}
To keep the argument structured, we first prove the key technical statement.

\begin{lemma}\label{lem:radiusIncrease}
    Let $(\sigma^*, F^*)$ be any correct solution to $\Pi^*$, and let $v$ be any node whose
    $(r\!-\!1)$--hop neighborhood matches some configuration $C$.
    Let $C' \in \C^*$ be such that $C$ is the $(r\!-\!1)$--hop subconfiguration of $C'$.
    Then $v$ also matches~$C'$.
\end{lemma}

\begin{proof}
    Let
    \[
        f_v : \Nh_{r-1}(v) \to \V
    \]
    be the matching map witnessing that $\Nh_{r-1}(v)$ matches $C$.
    Our goal is to extend $f_v$ to a mapping
    \[
        g : \Nh_r(v) \to \V'
    \]
    that matches all of $C'$.

    The key observation is that every node $u$ at distance exactly $r-1$ from $v$ is,
    by construction, a twig in $C'$.
    Thus its label in $C$ is of the form $(x,(y,T_u^{+}))$, meaning that the associated
    twig configuration $T_u$ must be included in the output set $X_u$ of~$u$.
    Consequently, by Lemma~\ref{lem:CxIsIntersection}, the 1-hop neighborhood $\Nh(u)$
    can be matched to the twig configuration
    \[
        T_u = (\V_u, \E_u, \cen_u, \mu_u, \tau_u).
    \]
    Let $h : \Nh(u) \to \V_u$ be such a matching.

    Since $T_u$ is (by construction) isomorphic to the 1-hop neighborhood of $f_v(u)$
    in $C'$, we may regard $h$ as mapping $\Nh(u)$ into the appropriate copy of $T_u$
    inside $C'$.
    To extend $f_v$, we must ensure that $h$ is consistent with $f_v$ on the overlap
    $\Nh(u) \cap \Nh_{r-1}(v)$.

    The nodes in this intersection are $u$ itself and its unique neighbor $w$ that lies
    closer to $v$.
    The maps $f_v$ and $h$ already agree on $u$.
    If $h(w)$ is not equal to the corresponding copy $w'$ of $f_v(w)$ in $T_u$,
    then $h$ must map some other node $s$ (which necessarily has the same label as $w$)
    to~$w'$.
    In that case we define a modified map $h'$ obtained from $h$ by swapping the images
    of $w$ and~$s$.
    This ensures that $h'$ and $f_v$ now agree on the intersection
    $\Nh(u) \cap \Nh_{r-1}(v)$.
    Since the swap only exchanges two nodes with identical labels, $h'$ still correctly
    matches the configuration $T_u$.

    We may now extend $f_v$ over $\Nh(u)$ by setting
    \[
        f_v'(z) =
        \begin{cases}
            f_v(z), & z \in \Nh_{r-1}(v),\\[4pt]
            h'(z),  & z \in \Nh(u).
        \end{cases}
    \]
    Because the graph is a tree, the neighborhoods of two nodes $u,u'$ at distance exactly $r-1$ from $v$ are pairwise disjoint (except at their
    respective closest to $v$ nodes $w,w'$), so performing this extension one twig at a time produces
    a well-defined global mapping
    \[
        g : \Nh_r(v) \to \V'.
    \]

    Finally, since $f_v$ and each of the modified twig matchings $h'$ are correct
    configuration matchings, and they agree on overlaps by construction, the combined
    map $g$ is a valid matching from $\Nh_r(v)$ to $C'$.
    Hence $v$ matches $C'$, as required.
\end{proof}

We are now ready give the proof of our main theorem.

\begin{proof}[Proof of \Cref{thm:reduceR}]
    We first prove (\ref{enum:PiToPiStar}), so let $(\sigma, F)$ be a solution to $\Pi$. We first fix the outputlabels. For every node $v$, we collect all information in the $r$-hop neighborhood of $v$, this can be done in $r$ rounds of the local model. We get 
    \[
        X_v := \set{T_w \mid u \in \Nh_r(v) \text{ and } f_u:\Nh_r(u) \rightarrow C \text{ maps } v \text{ into a twig } w}.
    \]
    Here $T_w$ refers to the twig configuration of $w \in C$, so $X_v \subset \Pow(\T^\Pi)$ is the set of all twig configurations that $v$ is matched to. We define our output labeling with $\sigma^*(v) = (\sigma(v), X_v)$. This uniquely determines $\sigma^*$, so what is left is to match every $r-1$ hop neighborhood to some configuration in $\C^{r-1}$.

    Fix any node $v$, let $C$ be the configuration that $v$ is matched to and $X = X_v$ (to clean up notation).
    Since $\Nh_r(v)$ can match $C$, there must exist $f_v:\Nh_r(v)\rightarrow C$ that correctly matches $C$. Similarly, since $\Nh_1(v)$ can match all configurations in $X$, we get that $\Nh_1(v)$ can match $C_{X}$ and so by \Cref{lem:extendedConfigs} there exists $C' \in C^X$ such that $\Nh_r(v)$ can match $C' = (\V', \E', \cen', \mu^*, \tau')$.

    Since the $(r-1)$-hop restriction of $C'$ is a configuration in $\Pi^{(r-1)}$, it is sufficient to match $v$ to $C'$.

    Let $f': \Nh_r(v) \rightarrow \V'$ be the function that matches $\Nh_r(v)$ to $C'$. It is easy to see that the labeling $\mu^*$ was defined precisely to allow the labeling $\sigma^*$ (that is, every node outputs the set of all used twig-configurations).

    This finishes the first part.

    \medskip
    
    We now prove (\ref{enum:PiStartToPi}), so let $(\sigma^*, F^*)$ be a solution to $\Pi^*$. With the work we have already done in \Cref{lem:radiusIncrease} this is easy. The new output labeling $\sigma$ simply takes the $\Sout$ outputs assigned by $\sigma^*$, completely ignoring the sets of twig configurations. We claim that this already constitutes a valid solution to $\Pi$. What is left is to obtain the appropriate matching functions $F$.
    
    Let $v$ be any node. Since $(\sigma^*, F^*)$ is a correct solution, $f_v \in F^*$ correctly matches $v$ to some configuration $C_v \in \C^{(r-1)}$. By \Cref{lem:radiusIncrease} we get that $\Nh_r(v)$ can match some $C' \in \C^*$. By the construction of $\C^*$ this means that $C' \in C^X$ for some configuration $C \in \C$ and some set $X \in \Pow(\T^\Pi)$. Now by \Cref{lem:extendedConfigs} this means that $\Nh_r(v)$ can match $C = (\V, \E, \cen, \mu, \tau) \in \C$. So in $r$ rounds $v$ can collect all of the information in $\Nh_r(v)$ and brute force a mapping $f_v:\Nh_r(v) \rightarrow \V$ that correctly matches $\Nh_r(v)$ to $C$.
\end{proof}

We now simply repeat this $r-1$ times.

\EquivNodeEdge*
\begin{proof}
    Apply \Cref{thm:reduceR} $r-1$ times and then \Cref{lem:rIsOne}.
\end{proof}

\section{Polynomial Gaps in Trees}
For the remainder of this section we fix a node--edge checkable LFL
\(\Pi = (\Sin,\Sout,\C_V,\C_E)\) on trees.  We will prove the following gap:
for every integer \(k\ge 1\), either \(\Pi\) has deterministic complexity
\(\Omega\big(n^{1/k}\big)\) or it can be solved in
\(O\big(n^{1/(k+1)}\big)\).

The argument proceeds in two steps.  First we define a structural property of
\(\Pi\) that allows one to place outputs on the two ends of a long path
without global coordination.  If this operation can be iterated \(k\) times we
say that a \emph{\(k\)-good} function exists.  We then show that the
existence of a \(k\)-good function yields an \(O\big(n^{1/(k+1)}\big)\)-round
algorithm for~\(\Pi\).  Conversely, we show that any algorithm running in
$o\big(n^{1/k}\big)$ rounds implies the existence of a \(k\)-good function.

To make the property above precise we introduce the notion of \emph{types},
and then analyze how types behave on long paths.  We closely follow
the presentation in \cite{smallMessagesbcmos21}, however we must repeatedly accommodate the
possibility of unbounded node degree, which requires careful adjustments at
each stage of the proof.

\subsection{Computing Types of Trees}
Consider an input-labeled tree \(T = (V,E,\phi)\) and pick a connected
subgraph \(T_v = (V_v,E_v,\phi_v)\) that is attached to \(T \setminus T_v\)
by at most two edges, all incident to a single node \(v \in T_v\).
If there is exactly one such edge, let \(e_v\) denote the corresponding
adjacent half-edge.  If there are two such edges, let the corresponding
halfedges be \(e_{\mathrm{left}}\) and \(e_{\mathrm{right}}\).

\begin{definition}[Type]
    Let \(T, v, T_v, e_v\) (or \(e_{\mathrm{left}}, e_{\mathrm{right}}\))
    be as above.  We define the \emph{type} \(\Type(T_v)\subseteq \Sout\)
    (or \(\Type(T_v)\subseteq \Sout^2\)) as the maximal set of all
    output labels that may be assigned to \(e_v\) (or to
    \(e_{\mathrm{left}}, e_{\mathrm{right}}\)) such that there exists a
    labeling
    \[
        \sigma : \hE_v \cup \{e_v\} \to \Sout
    \]
    that is locally correct everywhere inside \(T_v\).  
    (In the two-pole case, the domain is
    \(\hE_v \cup \{e_{\mathrm{left}}, e_{\mathrm{right}}\}\).)

    With abuse of notation, we also let \(\Type\) denote the set of all
    types.
\end{definition}

Consider such a tree $T_v$, with $e_v$, as rooted at $v$, with $A$ being the set of $v$'s children.
We get that the type of $T$ can be computed as a function of $(\Type(a))_{a\in A}$ the types of the subtrees of $v$'s children. 
We will use this to compute all possible types that can actually appear in a real instance of $\Pi$.

However, there are already infinitely many such trees of height 1 in the unbounded-degree setting. To cope with this fact, we will need a concise representation of such a tree.

\begin{definition}[Virtual Trees]
    A virtual tree with 0,1 or 2 poles $\T=(\La_{in})$ \emph{(}, or $\T = (\La_{in}, x)$, or $\T = (\La_{in}, x_{left}, x_{right})$\emph{)}, consists of the following:
    \begin{itemize}
        \item $\La_{in}$ the \emph{set of incoming labelings} is a multiset of elements from $\Sin^2 \times \Type$.
        \item $x \in \Sin$ is a single inputlabel, or both $x_{left}, x_{right} \in \Sin$ are input labels 
    \end{itemize}
    Virtual trees represent a node $v$ with $|\La_{in}|$ many incoming edges and 0, 1 or 2 outgoing edges. Each $(x_{far},x_{adj},t) \in \La_{in}$ represents an incoming edge $e$ of $v$ with half edges labeled $x_{far}$ and $x_{adj}$ (adjacent) respectively, furthermore the halfedge not adjacent to $v$ may only output labels from $t$. 

    Virtual trees with 0 poles do not have any outgoing edges, virtual poles with one pole, have a single outgoing half-edge $e_{out}$ input-labeled with $x$ and virtual trees with 2 poles have two outgoing halfedges $e_{left},e_{right}$, input-labeled with $x_{left},x_{right}$ respectively.
\end{definition}
A visualization of a virtual tree with one pole is given in \Cref{fig:virtualTree}

\begin{figure}[!ht]
    \centering
    \includesvg[width=0.6\linewidth]{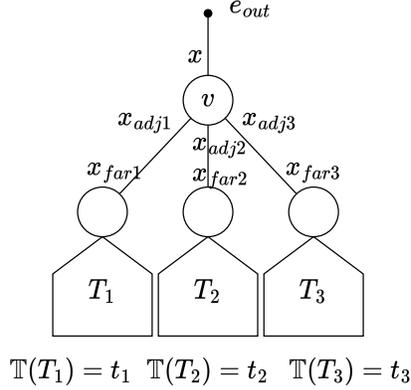}
    \caption{Representation of a virtual tree with 3 incoming edges. The types of the incoming edges are represented by subtrees $T_1, T_2, T_3$ that have the respective types $t_1, t_2, t_3$.}
    \label{fig:virtualTree}
\end{figure}

Similarly to types of actual trees, we define the type of such a virtual tree.
\begin{definition}[Type of a Virtual Tree] \label{def:virtTrees}
    For a virtual tree with 1, or 2 poles $\T = (\La_{in}, x)$ \emph{(},or $\T = (\La_{in}, x_{left}, x_{right})$\emph{)} we define the type $\Type(\T) \subset \Sout$ \emph{(}, or $\Type(\T) \subset \Sout^2$\emph{)} as the set of all admissible outputlabels. More precisely $y \in \Type(\T)$ \emph{(},respectively $(y_{left}, y_{right})$\emph{)} if and only if there exists a $\sigma: \La_{in} \rightarrow \Sout \times \Sout$ satisfying, for every $(x_{far}, x_{adj}, t) \in \La_{in}$ with  $\sigma((x_{far}, x_{adj}, t)) = (a,b)$:
    \begin{itemize}
        \item The types are respected, so $a \in t$. 
        \item Each edge satisfies the edge constraints, so $((x_{far},a),(x_{adj},b)) \in \C_E$
        \item We have a valid labeling around our node, so the multiset union $\set{(x,y)}\bigutimes_{\La_{in}} \set{(x_{adj},b)}$is in $\C_V$. 
        \emph{($\set{(x_{left},y_{left}), (x_{right},y_{right})} \bigutimes_{\La_{in}} \set{(x_{adj},b)} \in \C_V$)}
        \item $\Type(\T)$ is inclusion wise maximal.
    \end{itemize}
\end{definition}

In more handwavy terms: the type of $\T$ includes every outputlabel $y$ such that we can find a valid labeling of the virtual tree that satisfies the types. When thinking about \Cref{fig:virtualTree}, we notice that $\sigma$ in the definition of type just assigns every incoming edge of $v$ two outputlabels, one for each half edge. As a result every half edge is now input and output labeled. We require the outputlabels put on the \emph{far} edges to match their respective type. Additionally every such edge must have an input/outputlabeling that is allowed in our LFL $\Pi$. Lastly the labeling around $v$, where the outgoing half edge is labeled $(x,y)$ must also be a valid labeling according to $\Pi$.

Lastly, we complete the set of definitions, by also giving a statement about virtual trees with 0 poles, for which we just require that any valid assignment exists.
\begin{definition}
    A virtual tree with 0 poles is good, if there exists an assignment of outputlabels $\sigma: \La_{in} \rightarrow \Sout \times \Sout$ to the elements of $\La_{in}$, satisfying for every $(x_{far}, x_{adj}, t) \in \La_{in}$ with  $\sigma((x_{far}, x_{adj}, t)) = (a,b)$:
    \begin{enumerate}
        \item $a \in t$. 
        \item $((x_{far},a),(x_{adj},b)) \in \C_E$
        \item the multiset union $\bigutimes_{\La_{in}} \set{(x_{adj},b)}$is in $\C_V$.
    \end{enumerate}
\end{definition}

We show the power of these definitions, by giving a procedure that computes all of the types that can be encountered when solving $\Pi$ in a tree. For this we need to make sure that we do not have to look at an infinite amount of virtual trees. To that end, we show that for any very large virtual tree, there is a finite virtual tree with the same properties.

We define $s:= \max \set{|C| \mid C \in \C_V}$ as the maximum cardinality of any node configuration. Importantly, since each configuration is finite, we have $s \in O(1)$.

\begin{lemma}\label{lem:finiteVirtualTrees}
    For every virtual tree with 0,1 or 2 poles, $\T$ with set of incoming labelings $\La_{in}$, there exists a virtual subtree $\T'$ with set of incoming labelings $\La_{in}' \subset \La_{in}$, such that $|\La_{in}'| \le (s+1) \cdot |\Sin|^2 \cdot 2^{|\Sout|}$ and such that 
    \begin{itemize}
        \item If $\T$ has 0 poles, then $\T'$ is good if and only if $\T$ is good.
        \item If $\T$ has 1, or 2 poles, then $\Type(\T) = \Type(\T')$.
    \end{itemize}
    Moreover, $\T'$ can be constructed in time linear in the size of $\La_{in}$.
\end{lemma}
\begin{proof}
     We will prove the statement for virtual trees with 1 pole, the proof for 0 or 2 poles work exactly the same. So let $\T = (\La_{in},x)$.

     If $|\La_{in}| \le (s+1) \cdot |\Sin|^2 \cdot 2^{|\Sout|}$, the statement is trivial.

     So now assume that $|\La_{in}|>(s+1) \cdot |\Sin|^2 \cdot 2^{|\Sout|}$. Since there are only at most $2^{|\Sout|}$ many different types and only $|\Sin|^2$ many ways to input label an edge, there must be a type $t$ with input labels $x_{far},x_{adj}$ such that $\La_{in}$ contains more than $s+1$ copies of $(x_{far}, x_{adj}, t)$. We show that an equivalent virtual tree is reached if $(x_{far}, x_{adj}, t)$ appears only $s+1$ times. By repeating this for all different such elements of $\La_{in}$ that appear too often, we get some virtual tree with $|\La_{in}'| \le (s+1) \cdot |\Sin|^2 \cdot 2^{|\Sout|}$.
    
    So let $\La_{in}$ have more than $s+1$ copies of some $(x_{far}, x_{adj}, t)$ we argue that $\Type(\T') = \Type(\T)$, where $\T'$ is the virtual tree $\T$ with only $s+1$ copies of $(x_{far}, x_{adj}, t)$.
    
    To see that $\mathbf{\Type(\T) \subset \Type(\T')}$: Consider some $y \in \Type(\T)$ there must be an assignment of outputlabels to the elements of $\La_{in}$, such that the multiset union $\set{(x,y)}\bigutimes_{\La_{in}} \set{(x_{adj},b)}$ matches some $C \in \C_V$. $C$ contains two kinds of elements, required elements $a \in \Sout$ and respective optional elements of the form $a^*\in \Sout^* \setminus \Sout$. To show that also $y \in \Type(\T')$, by our definition of matching a configuration, it suffices to show, that $\T'$ contains enough copies of $(x_{far}, x_{adj}, t)$ to satisfy the requirements.
    Seeing as $s$ is the maximum cardinality of any configuration, it is also an upper bound on the number of required labels in any configuration. So since $\T'$ contains $s+1$ copies of $(x_{far}, x_{adj}, t)$ we can satisfy the required labels of $C$.
    
    To see that $\mathbf{\Type(\T') \subset \Type(\T)}$: The only difference between $\T'$ and $\T$ is that $\T$ contains more copies of $(x_{far}, x_{adj}, t)$.
    Since $\La_{in}'$ contains $(x_{far}, x_{adj}, t)$ exactly $s+1$ times, there must exist at least one $b^* \in C$, such that $b \in t$. So excess elements of the form $(x_{far}, x_{adj}, t)$ can output $b^*$.
\end{proof}

Since $|\La_{in}'| \le (s+1) \cdot |\Sin|^2 \cdot 2^{|\Sout|}$ is finite, the type of a virtual tree can just be bruteforced.

\begin{lemma}\label{lem:typeComputation}
    For a given virtual tree $\T$ with 1 or 2 poles the type $\Type(\T)$ can be computed. For a given virtual tree $\T$ with 0 poles whether $\T$ is good or not can be computed.
\end{lemma}
\begin{proof}
    We first invoke \Cref{lem:finiteVirtualTrees} to obtain a finite virtual tree with the same type.\\
    Since $\Sin, \Sout$ and $\La_{in}$ are finite, we can iterate all choices $y \in \Sout$ for virtual trees with 1 pole (,respectively iterate all $(y_{left}, y_{right})$ for trees with 2 poles) and check for all possible mappings $\sigma((x_{far}, x_{adj}, t)) = (a,b)$ if they satisfy the conditions in \cref{def:virtTrees}. The argument for 0 poles is the same.
\end{proof}

Now that we have shown that virtual trees allow us to restrict to finite representations, we also show that they are expressive enough to yield results for any instance.

\begin{lemma}\label{lem:TreesAreVirtualTrees}
    For every subtree $T_v$ as above, connected to $T\setminus T_v$ with one or two edges, there exists a virtual tree with one pole \emph{(}resp. two poles\emph{)} $\T = (\La_{in}, x)$ \emph{(}resp $\T = (\La_{in}, x_{left}, x_{right})$\emph{)}, such that $\Type(T_v) = \Type(\T)$.
\end{lemma}
\begin{proof}
    We only prove the statement for virtual trees with 1 pole. Either $v$ is a leaf, or $v$ is the root of some tree $T_v$.
    If $v$ is a leaf, use $\T = (\set{}, x)$, with $x$ being the input label on $v$'s adjacent halfedge.

    Otherwise, let $v$ have $k$ children $c_1, \ldots, c_k$. Each of the subtrees $T_{c_j}$ has some type  $\Type(T_{c_j})$. So we can use the virtual tree 
    \[
    \T = (\set{(x_{far1}, x_{adj1}, \Type(T_{c_1}), \ldots, (x_{fark}, x_{adjk}, \Type(T_{c_k}))}, x)
    \]
    , where the inputlabeling matches exactly the inputlabeling around $v$. By definition of types for trees and virtual trees, it holds that $\Type(\T) = \Type(T_v)$.
\end{proof}

In the same way as in the proof of \Cref{lem:TreesAreVirtualTrees} we can construct trees of any possible type iteratively, by using smaller trees. If we create a tree with a new type, that we have not seen before, we can try using it to create even more trees with different types. 

We now use this idea to compute all possible types.  The following
procedure starts from an initial set \(R = \emptyset\) and repeatedly
enlarges it until a fixed point is reached.

\paragraph*{Compute Types}
Given an LFL \(\Pi = (\Sigma_{\mathrm{in}}, \Sigma_{\mathrm{out}}, \mathcal{C}_V, \mathcal{C}_E)\) and an initial set
\(R \subseteq \Type \times \Sigma_{\mathrm{out}}\), the procedure computes the maximal
set of types that any subtree with a single outgoing labeled edge can
have.

Remember that $s \;:=\; \max\{\,|C| \mid C \in \mathcal{C}_V\,\}$.  

\smallskip
\noindent Repeat the following steps:

\begin{enumerate}
    \item For all values $i \in \set{0, \ldots, (s+1) \cdot |\Sin|^2 \cdot 2^{|\Sout|}}$ create all possible virtual trees with 0 and 1 poles $\T = (\La_{in}, x)$ (, or $\T = (\La_{in})$) with $|\La_{in}| = i$, by choosing $x \in \Sin$, $x_{adj} \in \Sin$ and $(t,x_{far})$ from $R$ (for the elements $(x_{far}, x_{adj}, t) \in \La_{in}$).
    \item Check that all virtual trees with 0 poles are good, otherwise output \emph{invalid problem}.
    \item Compute the types of all of the virtual trees $\T = (\La_{in}, x)$ with 1 pole and let $R'$ be the set of tuples $(\Type(\T), x)$. 
    \item If the type of any of these virtual trees is the empty set output \emph{invalid problem}
    \item If $R = R'$ terminate, else set $R \gets R \cup R'$ and go to step 1.
\end{enumerate}

Note that by \Cref{lem:finiteVirtualTrees} it is enough to consider only virtual trees of size at most $(s+1) \cdot |\Sin|^2 \cdot 2^{|\Sout|}$, to obtain all possible types.

Because $\Sin,\Sout$ are finite, $R \subset \Pow(\Sout)\times \Sin$ must be finite. But since we require $R$ to grow in each iteration, there can only be a finite number of iterations.
\begin{corollary}
    The procedure Compute Types terminates.
\end{corollary}

Now that we know that the procedure is well behaved, we argue that it is also computing what we want. For any subgraph $T_v$ in any instance the type $\Type(T_v)$ will be produced at some point.
\begin{lemma}\label{lem:constructTrees}
    Let $R$ be the set returned by \emph{Compute Types}, initialized with $R=\emptyset$.
Let $T_v$ be any subtree with one outgoing halfedge labeled $x$, rooted at some node $v$ in any valid instance of $\Pi$.
Then
\[
    (\Type(T_v),x) \in R.
\]
\end{lemma}
\begin{proof}
    We prove the claim by induction on the height. 
    To start, if $T_v$ has height 0, then in the first iteration, for $i=0$ the virtual tree $(\set{},x)$ is considered, so $(\Type(T_v),x) \in R$.

    For the inductive step, let $v$ have children $c_1, \ldots, c_k$. They are connected to $v$ with edges input labeled $x_{far1}, x_{adj1}, \ldots,x_{fark}, x_{adjk}$ and are roots of respective, height $\le i-1$, trees $T_1, \ldots, T_k$. By the induction hypothesis $(\Type(T_1),x_{adj1}),\ldots, (\Type(T_k),x_{adjk}) \in R$ in some iteration of \emph{Compute Types}. So in that iteration, the virtual tree 
    \[
    \T = (\set{(x_{far1}, x_{adj1}, \Type(T_1)),\ldots, (x_{fark}, x_{adjk}, \Type(T_k))}, x)
    \]
    is considered, by definition $\Type(\T) = \Type(T_v)$, so $\Type(T_v) \in R$.
\end{proof}

As a direct corollary of this it follows, that if ever $\emptyset \in R$, then the problem is unsolvable in some instances.

\begin{corollary}
    If Compute Types, with initial set $R = \emptyset$ outputs \emph{invalid problem}, then there exist some tree $T$ on which there are no solutions to $\Pi$.
\end{corollary}
\begin{proof}
    Start with the virtual tree $\T$ such that $\Type(\T) = \emptyset$, or with the virtual tree with 0 poles that is not good. We reconstruct a tree on which there is no solution for $\Pi$.
    To do so, we recursively replace the incoming edges of type $t$ with the first virtual tree $\T_t$ that created type $t$. Notice that $\T_t$ must have been created at least one iteration before $\T$, therefore this recursion will end. 
    
    If we started with a virtual tree with 1 pole, we add a direct parent node $p$ to $v$ and put an arbitrary input label on the halfedge $(\set{v,p},p)$. We have produced a labeled tree $T = (V, E, \phi)$, with $\Type(T) = \Type(\T) = \emptyset$.
    There cannot be an output labeling $\sigma$ of $T$ that is correct with respect to $\Pi$, as this labeling would imply that $\sigma(e_v) \in \Type(T_v)$ (,or that $\T$ was good).
\end{proof}

\paragraph{An $O(D)$ algorithm:}
Notice that \emph{Compute Types} immediately yields an $O(D)$-round algorithm, where $D$ is the diameter of the tree.
The algorithm proceeds by iteratively removing leaves.

If $v$ is a leaf, then in the first round it can locally compute its type and propagate this information to its unique neighbor.
More generally, suppose that $v$ is removed in some later round.
Using the types propagated to $v$ in previous rounds, node $v$ can construct the corresponding virtual tree $\T_v$, compute $\Type(\T_v)$, and propagate this type to its remaining neighbor.

Eventually, only a single node $r$ remains.
(If the last remaining structure is an edge with two nodes, ties are broken arbitrarily.)
The virtual tree of $r$ must be good, since otherwise $\Pi$ would not be solvable on this tree.
Node $r$ can then choose any output labeling that is compatible with the types of all its incident subtrees.
This choice is subsequently propagated backwards.

Once $v$'s unique upwards edge receives a label through this backwards propagation process, it is $v$'s turn to choose. The output label on its upward edge has already been fixed to some value, guaranteed to be in $\Type(T_v)$.
By definition of types, $v$ can therefore choose labels on its incident halfedges to form a valid local labeling that is consistent with the types of all its children.
Hence, the entire process results in a globally valid labeling.

While correct, this algorithm has round complexity $O(D)$, which is far too slow for our purposes.
Removing only leaves is therefore insufficient.
To obtain faster algorithms, we also need to remove long degree-$2$ paths.
The following definitions extend the notion of types to paths.

\begin{definition}[Compress Path]
    A compress path $H=(V, E, \phi, e_s, e_t)$ is formed by a core path $P=(v_1, \ldots, v_l)$ with the endpoints $s=v_1$ and $t=v_l$. Each node $v_i$ on the path might be the root of an arbitrary subtree $T_{v_i} $ attached to it. $e_s$ and $e_t$ are halfedges attached to $s$ and $t$ respectively. $e_s$ and $e_t$ do not have a corresponding half edge and also the other endpoint is not in $H$. All halfedges in $H$, including $e_s,e_t$ are inputlabeled by $\phi$.

    We say that the length of $H$ is the length of its core path.
\end{definition}

We now extend the notion of types to such paths.
Analogously to the previous definitions, the \emph{type} of a path $H$ is defined as the set of output labelings on its \emph{connecting edges} $e_s$ and $e_t$ that can be extended to a valid labeling of the entire path $H$.

\begin{definition}[Type of a compress path]
    The type $\Type(H)$ of a compress path $H = (V, E, \phi, e_s, e_t)$, is the maximal set of tuples of outputlabels $(y_s,y_t) \in \Sout^2$, satisfying the following property: For any $(y_s,y_t) \in \Type(H)$ there exists a valid labeling $\sigma$ of $H$, such that $e_s$ has outputlabel $y_s$ and $e_t$ has outputlabel $y_t$.
\end{definition}

Note that here a valid labeling means every node in $H$ must match a configuration of $\C_V$, where $s$ and $t$ include the labels from $e_s$ and $e_t$ respectively. Furthermore all edges in $H$, except $e_s$ and $e_t$ must match a configuration of $\C_E$.

Clearly, from the way types are defined, any type is a subset of $( \Sout)^2$, since this set is finite, also the number of types is finite.

\begin{corollary}
    The number of different types that a compress path can have is upper bounded by $2^{(|\Sout| )^2}$.
\end{corollary}

The difficulty with types for paths is that they effectively require us to choose the output labels on $e_s$ and $e_t$ \emph{simultaneously}. These choices generally depend on the structures of the subgraphs $G_s$ and $G_t$ attached to $H$ at either end. Consequently, the labeling chosen in $G_s$ must be coordinated with the labeling chosen in $G_t$.

To avoid this global coordination, we introduce the notion of \emph{independent classes}. Independent classes allow us to select the output labels on $e_s$ and $e_t$ \emph{independently}, while still guaranteeing that the resulting choices can be extended to a valid labeling of the entire path.

\begin{definition}[Independent Class]
    Let $H=(V, E, \phi, e_s, e_t)$ be a compress path, an independent class $\I$ for $H$ is a pair of sets of outputlabels $(X,Y) \in \Pow(\Sout)\times \Pow(\Sout)$, such that for any choice of outputlabels $x \in X, y \in Y$, there exists a valid labeling $\sigma$ of $H$ with $\sigma(e_s) = x, \sigma(e_t) = y$.\\
    If either $X$ or $Y$ is empty we call $\I$ empty.
\end{definition}
Again, let $G_s$ and $G_t$ be the graphs attached to $H$ through $e_s$ and $e_t$.
Under the restriction that $e_s$ is assinged a label from $X$ and $e_t$ is assigned a label from $Y$, the labeling in $G_s$ can be computed without regard to the choices made in $G_t$, and vice versa. For every such pair of choices, the labeling can be extended to a valid labeling of $H$. Thus, $H$ acts as padding between $G_s$ and $G_t$, allowing both sides to be solved independently. This padding will be what allows our algorithms to be fast, but it only exists, if we have ''good'' independent classes. 

We turn to studying these independent classes, again by first abstracting to virtual compress paths.

\begin{definition}[virtual compress path]
    A \emph{virtual compress path} $\Hp = (\T_i)_{1 \le i \le l}$ of length $l$ is a sequence of virtual trees, each with two poles.

Such a virtual compress path represents a core path
\[
P = (v_1, v_2, \ldots, v_l),
\]
where each node $v_i$ may have multiple incoming neighbors. For each $i$, the virtual tree
\[
\T_i = (\La_{in}^i, x_{left}^i, x_{right}^i)
\]
corresponds to the node $v_i$. Consecutive nodes $v_i$ and $v_{i+1}$ are connected by identifying the right pole of $\T_i$ with the left pole of $\T_{i+1}$, forming an edge $\{v_i, v_{i+1}\}$ whose input labels are $(x_{right}^i, x_{left}^{i+1})$.

The remaining poles $x_{left}^1$ and $x_{right}^l$ are assigned to two distinguished halfedges $e_s$ attached to $v_1$ and $e_t$ attached to $v_l$, respectively.
\end{definition}

In exactly the same manner as we already did for virtual trees, we can extend the definition of types to virtual compress paths. 

\begin{definition}[type of a virtual compress]
    For a virtual compress path $\Hp = (T_i)_{1 \le i \le l}$, the type $\Type(\Hp) \subset (\Sout)^2$ of $\Hp$ is a set of tuples of outputlabels.\\
    For every element $(y_s,y_t) \in \Type(\Hp)$ it holds that there exists a labeling $\sigma$ of the halfedges of the path $P=(v_1, \ldots, v_l)$, including the two halfedges $e_s, e_t$, such that:
    \begin{itemize}
        \item $\sigma(e_s) = y_s, \sigma(e_t)=y_t$.
        \item For every $T_i$, the labels assigned to the two poles of $T_i$ are in $\Type(T_i)$.
        \item $\sigma$ assigns all of the edges $\set{v_i,v_{i+1}}$ a valid labeling, for $1 \le i < l$.
        \item $\Type(\Hp)$ is inclusion wise maximal.
    \end{itemize}
\end{definition}

This is exactly the reason, why we also defined virtual trees with 2 poles and proved the previous results for these as well. It is now easy to conclude the following.

\begin{lemma}\label{lem:CompressToVirtual}
    For every compress path $H =(V,E, \phi, e_s, e_t)$, there exists a virtual compress path $\Hp=(\T_i)_{1\le i\le l}$, such that $\Type(H) = \Type(\Hp)$.
\end{lemma}
\begin{proof}
    Use \Cref{lem:TreesAreVirtualTrees} to replace every node $v_i$ in the core path of $H$ with a virtual tree with two poles $\T_i$. The fact that the types are the same follows from the fact, that the types of the $\T_i$ are identical to the types of the nodes on the core path.
\end{proof}

For a given compress path, there might be more than one reasonable independent class, which one to choose is not clear a priori. For now we assume that we are given a function, that maps virtual compress paths of a certain length to independent classes.

In \Cref{sec:algoImpliesFunc} we will turn a fast algorithm for $\Pi$ into an independent assigner for $\Pi$. There we will specify exactly which length $\ell$ we use. For now think of $\ell$ as some constant that depends only on $\Pi$ (and not on a specific independent assigner).

\begin{definition}[independent assigner]
    For every $\ell \in \N$ an independent assigner $f$ for $\Pi$ is a mapping that assigns every virtual compress path of length between $\ell$ and $2 \ell$ a (possibly empty) independent class with respect to $\Pi$.
\end{definition}

All virtual trees—and, by extension, all virtual compress paths—may have many incoming edges, each carrying a restriction in the form of a type that limits the output labels allowed on the far half-edge.

Initially, these types are determined solely by the topology of the tree and the constraints of the LCL, and are exactly the types $\Type_0$ computed by \emph{Compute Types}. However, once independent classes on compress paths are introduced, additional types may arise $\Type_0'$. We may now repeat the process of the \emph{Compute Types} procedure while including these new types. In particular, the multiset $\La_{in}$ of a virtual tree may now include constraints induced by independent classes, leading to new types $\Type_1$ for virtual trees with one pole.

These newly obtained types can again be used to form virtual trees with two poles, concatenate them into virtual compress paths, and apply independent classes, yielding further types $\Type_1'$. This induces a hierarchy of types: starting from the base set $\Type_0$ computed without independent classes, each application of independent classes may generate a new set of types, which can in turn be closed under \emph{Compute Types} to obtain the next level.

Crucially, the total number of possible types is finite. Therefore, this iterative process must eventually terminate. If it ever produces an empty type (or a virtual tree with zero poles that admits no valid labeling), the problem becomes unsolvable. Hence, it is essential that independent assigners avoid creating such empty types.

For a given independent assigner we can check how many such iterations we can do, before we obtain an empty type (if at all).

\paragraph*{Testing procedure:}
Given $\ell \in O(1)$, an LFL $\Pi$ and an independent assigner $f$ for $\Pi$, determine how often $f$ can be applied before obtaining an empty type.

First get all of the (type,input)-tuples that can be encountered naturally, by setting 
$R \gets$ \emph{Compute Types}($\Pi, \emptyset$).

Initialize $k = 0$, then repeat the following steps until termination:
\begin{enumerate}
    \item For all values $i \in \set{0, \ldots, (s+1) \cdot |\Sin|^2 \cdot 2^{|\Sout|}}$ create all possible virtual trees with two poles $\T = (\La_{in}, x, y)$ with $|\La_{in}| = i$,  by choosing $x,y \in \Sin$, $x_{adj} \in \Sin$ and $(t,x_{far})$ from $R$ (for the elements $(x_{far}, x_{adj}, t) \in \La_{in}$).\label{item:createBipolar}
    \item Create all possible virtual compress paths with lengths between $\ell$ and $2\ell$, using only the created virtual trees with two poles.
    \item For every such virtual compress path $\Hp = (\T_j)_{1\le j \le L}$ we apply $f$ to it, to obtain an independent class $I=(X,Y)$ consisting of two subsets of outputlabels $X,Y \subset \Pow(\Sout)$. We set $R' \gets R' \cup \set{(X,x_1),(Y,y_L)}$. ($x_1$, is the input label from $\T_1 = (\La_{in1}, x_1, y_1)$ and $y_L$ the input label from $\T_L = (\La_{inL}, x_L, y_L)$)\label{step:applyF}
    \item if any of the obtained set $X,Y$ are empty terminate and output $k$.
    \item $R' \gets $ \emph{Compute Types}($\Pi, R \cup R'$).\label{step:computeTypes2}
    \item If \emph{Compute Types} outputs \emph{invalid problem}, then terminate and output $k$.
    \item If $R' = R $ terminate and output $\infty$.
    \item Set $R \gets R'$, and $k \gets k+1$. 
\end{enumerate}

We say the initial call to \emph{Compute Types} is iteration 0 and after that we say that the execution of the loop with $k=0$ is the first iteration. The execution of the loop with $k=1$ is the second iteration and so on. The intuition behind this is that we set $k=i$ after the function has successfully passed the $i$-th iteration of the loop.

We argue that the testing procedure is well behaved. Notice that because there is a finite number of types we get the following.

\begin{corollary}\label{cor:constantlyManyCompresses}
    For a LFL $\Pi$ the number of different virtual compress paths considered in the testing procedure is a constant.
\end{corollary}

As a result also this procedure terminates in finite time.

\begin{corollary}
    The testing procedure terminates.
\end{corollary}
\begin{proof}
    In every iteration either the set $R$ grows by one, or (if $R' = R$) the procedure terminates. Since the number of possible types and the number of inputlabels are finite, the procedure must terminate.
\end{proof}

We may now use the testing procedure to define our notion of what a good independent assigner is.

\begin{definition}
    We say that an independent assigner $f$ for $\Pi$ is $k$-good (resp. $\infty$-good), if the testing procedure returns $k$ (resp. $\infty$) when testing $f$.
\end{definition}

We next describe how to solve any LFL given a good independent assigner. To do this we will assign virtual trees and virtual compress paths to the nodes of our input graph.
For this we will refer to what is known in the literature as a rake and compress decomposition. Each rake will correspond to a virtual tree with 1 or 0 poles and each compress will correspond to a virtual compress path.

In a Rake operation all nodes of degree $\le 1$ are removed, breaking ties arbitrarily when two (otherwise) isolated nodes have an edge. A Compress($\ell$) operation is parameterized by some value $\ell$ and removes all paths of degree 2 nodes of length at least $\ell$. 

On a high level a $(\gamma, \ell, k)$-decomposition is then obtained by $k$ repetitions of the following:
\begin{itemize}
    \item Perform $\gamma$ Rakes
    \item Perform a single Compress($\ell$)
\end{itemize}
After which we finish by performing another $\gamma$ rakes. 

The formal definition is analogous to the above algorithm, with the restriction that after the last $\gamma$ rakes the graph must be empty.
\begin{definition}[R\&C Decomposition]
A $(\gamma, \ell, k)$-decomposition of a graph $G=(V,E)$, is a partition of the nodes into $2k+1$ layers. The \emph{Rake Layers} $R_0, \ldots, R_{k}$ and the \emph{Compress Layers} $C_1, \ldots, C_k$. Each Rake Layer $R_i$ consists of $\gamma$ sublayers $R_{i,1}, \ldots, R_{i,\gamma}$. The layers satisfy the following constraints:
\begin{enumerate}
    \item The subgraph induced by the nodes in each $R_i$ has components of diameter at most $O(\gamma)$.
    \item Each sublayer $R_{i,j}$ consists only of isolated nodes. Each such isolated node has at most one neighbor in a higher layer.
    \item The subgraph induced by the nodes in each $C_i$ consists only of isolated paths of length between $\ell$ and $2\ell$. The endpoints of each such path have exactly one neighbor in a higher layer and no other nodes have neighbors in higher layers.
\end{enumerate}
Where higher and lower refers to the following order on layers: $R_0 < C_1 < R_1 < C_2 < \ldots < C_k < R_k$ ($R_{i,1} < R_{i,2} < \ldots < R_{i,\gamma}$).
\end{definition}

The following results due to \cite{CP19timeHierarchy} give optimal algorithms to compute these decompositions for different values of $k$ and any constant $\ell$.

\begin{lemma}[\cite{CP19timeHierarchy}]\label{lem:polyDecomp}
For $\ell \in O(1)$ and any positive integer $k$, we can set $\gamma \in O(n^{1/(k+1)})$ and compute a $(\gamma, \ell , k)$-decomposition in $O(kn^{1/(k+1)})$ rounds.
\end{lemma}  

\begin{lemma}[\cite{CP19timeHierarchy}]\label{lem:logDecomp}
For $\ell \in O(1)$ and $\gamma \in O(1)$ then by setting $L \in O(\log n)$ a $(\gamma, \ell , L)$-decomposition can be computed in $O(\log n)$ rounds.
\end{lemma}

Note that both results were originally proven for bounded-degree trees, but they extend to unbounded-degree trees using ideas from \cite{nodeAvgTrees} (Section 6).

\begin{corollary}
    \Cref{lem:polyDecomp,lem:logDecomp} also hold in the unbounded-degree setting.    
\end{corollary}

\begin{proof}
    The Compress procedure requires a distance-$\ell$ coloring with a constant number of colors. For bounded-degree graphs, this can be computed in $O(\log^* n)$ rounds. In the unbounded-degree setting such a coloring might not exist.

    However, the nodes removed by Compress do not depend on this coloring; only the remaining nodes require it. We can therefore compute the distance coloring locally in the subgraph induced by the removed nodes, which consists only of long paths. Computing a distance coloring on such paths still takes $O(\log^* n)$ rounds. 

    Since the rest of the graph does not depend on this coloring, this results only in an additive $O(\log^* n)$ overhead, leaving the asymptotic complexities of \Cref{lem:polyDecomp,lem:logDecomp} unchanged.
\end{proof}

With that out of the way, we can refocus on how we use these decompositions. We argue that given a $k$-good (,or $\infty$-good) independent assigner for $\Pi$ and a R\&C decomposition of $G$, we can solve $\Pi$ in $O(\gamma \cdot k)$ (, or $O(\log n)$) rounds. For this assume that all edges $\set{u,v}$, where $u$ has a strictly lower layer than $v$ are oriented from $u$ to $v$. Note that these are all edges, except the edges between nodes in a compress path. 

Based on the layer ordering, we inductively assign virtual trees to nodes and types to oriented edges. Specifically, every rake node is assigned a virtual tree with $0$ or $1$ poles, every compress path is assigned a virtual compress path, and every oriented edge $e$ is assigned a type $t_e \subset \Sout$.

\begin{itemize}
    \item \textbf{Rake layers $R_{0,i}$.}
    Consider a node $v$ in layer $R_{0,i}$ for some $i$. Let $T_v$ be the subtree rooted at $v$ consisting of all nodes reachable through strictly lower layers $R_{0,<i}$. Since $v \in R_{0,i}$, it has at most one neighbor in a higher layer.

    If $v$ has no outgoing edge, then $T_v$ is an isolated component of diameter at most $\gamma$. In this case, $v$ can brute-force a solution to $\Pi$, and we ignore such nodes.

    Otherwise, let $e_{out}$ be the unique outgoing edge of $v$. By \Cref{lem:TreesAreVirtualTrees}, the subtree $T_v$ induces a virtual tree $\T_v$ with one pole. We assign $\T_v$ to $v$ and assign the type $\Type(\T_v)$ to the oriented edge $e_{out}$.

    \item \textbf{Compress layers $C_i$.}
    Consider a maximal connected path $P \subset C_i$, and assume that all oriented edges whose source lies in a lower layer than $C_i$ already have types assigned.

    Each node $v \in P$ has exactly two neighbors that are not in lower layers. Let $e_{left}$ and $e_{right}$ be the corresponding halfedges, with input labels $x_{left}$ and $x_{right}$. Thus, $v$ corresponds to a virtual tree with two poles.

    For every edge $\{u,v\}$ where $u$ is a lower-layer neighbor of $v$, the oriented edge is already assigned a type $t_u$ by induction. Let $x^{(u)}_{far}$ and $x^{(u)}_{adj}$ be the input labels on the halfedges $(\{u,v\},u)$ and $(\{u,v\},v)$, respectively. Collect all tuples $(x^{(u)}_{adj}, x^{(u)}_{far}, t_u)$ into a multiset $\La_{in}$.

    Define the virtual tree
    \[
        \T_v = (\La_{in}, x_{left}, x_{right}).
    \]
    By \Cref{lem:finiteVirtualTrees}, there exists a virtual tree $\T'_v = (\La_{in}', x_{left}, x_{right})$ with $\La_{in}' \subset \La_{in}$ and
    \[
        |\La_{in}'| \le (s+1)\cdot|\Sin|^2\cdot 2^{|\Sout|}.
    \]
    Let $\Hp_P = (\T'_v)_{v \in P}$ be the resulting virtual compress path. Since $P \subset C_i$, the length of $\Hp_P$ lies between $\ell$ and $2\ell$. We therefore apply the independent assigner $f$ to $\Hp_P$ and assign the resulting types to the two boundary edges $e_{left}$ and $e_{right}$.

    \item \textbf{Higher rake layers $R_{i,j}$ with $i>0$.}
    Consider a node $v \in R_{i,j}$ with $i>0$. By definition, $v$ has at most one neighbor in a higher layer. Let $e_{out}$ be this edge if it exists, and let $x$ be the input label on the halfedge $(e_{out},v)$.

    All remaining neighbors of $v$ lie in lower layers, and the corresponding edges are oriented toward $v$. By induction, each such edge $\{u,v\}$ has an assigned type $t_u$. Let $x^{(u)}_{far}$ and $x^{(u)}_{adj}$ denote the input labels on $(\{u,v\},u)$ and $(\{u,v\},v)$, respectively, and collect all triples $(x^{(u)}_{far}, x^{(u)}_{adj}, t_u)$ into a multiset $\La_{in}$.

    If $e_{out}$ exists, we assign to $v$ the virtual tree $\T_v = (\La_{in}, x)$ and assign $\Type(\T_v)$ to $e_{out}$. If $e_{out}$ does not exist, we assign the virtual tree $\T_v = (\La_{in})$ with zero poles.
\end{itemize}

The assignment of virtual trees and edge types depends only on information from lower layers. Since there are at most $2k+1$ layers and each rake layer consists of at most $\gamma$ sublayers, the entire assignment process completes in $O(\gamma \cdot k)$ rounds.

\begin{observation}\label{obsv:runtimePropagation}
    The assignment of virtual trees and types can be performed in $O(\gamma \cdot k)$ rounds in the deterministic LOCAL model.
\end{observation}

We now give a propagation argument analogous to the one used in the $O(D)$ algorithm.
Virtual trees with $0$ poles can locally choose a labeling that respects all incident
types, and this choice can then be propagated downwards through the layers.
At each step, output labels are chosen from the types assigned to the corresponding
halfedges. To ensure that this process is always possible, all assigned types must
be non-empty. This is exactly where the $k$-goodness of the function $f$ is required.

\begin{lemma}\label{lem:assignedTreesAreGood}
If a $k$-good (or $\infty$-good) function $f$ is used to assign independent classes to
virtual compress paths, then all assigned types are non-empty and all virtual trees
with $0$ poles are good.
\end{lemma}

\begin{proof}
We prove by induction on $i$ that after iteration $i$ of the testing procedure, all
types assigned to oriented edges starting in layers $R_i$ or $C_i$, together with
their corresponding input labels, are contained in the set $R$.

For $i=0$, only rooted trees are considered. By \Cref{lem:typeComputation}, the assigned
types are exactly those computed by \emph{Compute Types}, which form the initial set $R$.

Assume the claim holds for all layers below $C_i$, that is any (type,input)-pair of a node with layer $<C_i$ is already contained in $R$ in the $i$-th iteration.
Since the virtual trees with two poles that are assigned to any node in $C_i$ use only these (type,input)-pairs, each such virtual tree is also constructed in Step~\ref{item:createBipolar} of the testing procedure loop.
Hence, all corresponding virtual compress paths are considered,
and for $f(\Hp_P) = (X,Y)$ the tuples $(X,x)$ and $(Y,y)$ are added to $R$ by the end of
iteration $i$.

Similarly, for nodes in $R_i$, \Cref{lem:finiteVirtualTrees} guarantees that each assigned virtual tree corresponds to one with at most $(s+1)\cdot |\Sin|^2 \cdot 2^{|\Sout|}$ incoming labelings. 
All such trees are considered by \emph{Compute Types} in Step~\ref{step:computeTypes2}, so their types are included in
$R$.

Thus, all assigned virtual trees and types are examined by the testing procedure within
the first $k$ iterations. Since $f$ is $k$-good, none of these types is empty and all
virtual trees with $0$ poles are good; otherwise the procedure would terminate earlier.
The argument for $\infty$-good functions follows analogously.
\end{proof}

We now formalize the downward propagation argument.

\begin{lemma}\label{lem:fastAlgo}
    If a $k$-good independent assigner exists for $\Pi$, for any constant $\ell \in O(1)$, then $\Pi$ can be solved in $O(n^{1/(k+1)})$ rounds in the deterministic LOCAL model.

    If a $\infty$-good independent assigner exists for $\Pi$, for any constant $\ell \in O(1)$, then $\Pi$ can be solved in $O(\log n)$ rounds in the deterministic LOCAL model.
\end{lemma}

\begin{proof}
We first compute the appropriate decomposition using \Cref{lem:polyDecomp}, or \Cref{lem:logDecomp} and assign virtual trees and types accordingly.

By \Cref{lem:assignedTreesAreGood}, every virtual tree with $0$ poles admits a valid
local labeling that respects all incident types. Each node that is assigned such a good virtual tree, picks such a labeling arbitrarily.

Consider a rake node $v$ with outgoing edge $(v,p)$. When $p$ assigns an output label
consistent with the type of $(v,p)$, the induced label on the half-edge at $v$ lies in
$\Type(\T_v)$. Therefore, this labeling can be extended to all edges incident to $v$
while respecting the assigned types of the lower layer neighbors.

For a compress path $P$, the two edges connecting $P$ to higher layers eventually
receive labels from the two types assigned by $f$. Since these types form an
independent class, any such choice can be extended to a valid labeling of the entire
path.

Because all assigned types are non-empty, this propagation continues until leaf nodes
are reached, yielding a correct solution to $\Pi$. As in
\Cref{obsv:runtimePropagation}, the total running time is $O(\gamma \cdot k)$ rounds, which yields $O(n^{1/(k+1)})$, if $f$ is $k$-good and $O(\log n)$, if $f$ is $\infty$-good.
\end{proof}

What is left is to specify which independent assigner to use, if any \emph{good} independent assigners even exist.

Note that knowing that a good independent assigner exists is actually sufficient. For any constant $\ell$, we can enumerate all possible independent assigners. Both the number of possible virtual compress paths of length between $\ell$ and $2\ell$, and the number of independent classes, are finite. So we can run the testing procedure for all possible independent assigners to determine what the largest $k$ is, such that a $k$-good independent assigner exists. 

In the next section, we will show that there is some universal constant $\ell_{pump}$ that is sufficient to check and that choice of $\ell_{pump}$ depends solely on the description of $\Pi$. 

Furthermore, we show that if $\Pi$ admits a fast algorithm, then also some good independent assigner must exist.

\subsection{A Pumping Lemma for Trees}
In this section we will recover a $k$-good independent assigner $f$ from an $o(n^{1/k})$-algorithm $\A$. 

\paragraph*{High-level plan}
Suppose we have a $T(n)$-round algorithm $\A$ for the LCL $\Pi$, and let
$H = (V,E,\phi,e_s,e_t)$ be a compress path of length at least $2T(n)+3$.
Let $v$ be the middle node of $H$ and simulate the execution of $\A$ at $v$.

Since $H$ is sufficiently long, the radius-$T(n)$ neighborhood of $v$ is
entirely contained in $H$. Therefore, the behavior of $\A$ at $v$ is
independent of the rest of the graph, and $\A$ fixes a valid output labeling
for all halfedges incident to $v$.

Fixing these labels induces restrictions on the possible output labels that
can be assigned to the boundary halfedges $e_s$ and $e_t$ connecting $H$ to
the remainder of the graph. Each restriction is a subset of the labels
allowed by $\Type(H)$ for the respective half-edge.

Since $\Pi$ is node-edge checkable and $v$ separates $H$ into a left and a
right subpath, the feasible labelings on the two sides are independent of
each other. Consequently, the induced restrictions on $e_s$ and $e_t$
together form an independent class for $H$.

\paragraph*{Getting long compress paths}
We will first recover a full compress path $H$ for every virtual compress path $\Hp$ considered in the testing procedure. After that we will see how to turn these into longer versions of the same compress paths, so that we may apply $\A$ to obtain an independent class, as described above.

The following definition formalizes this reconstruction, and the subsequent lemmas show that it is always possible. While the definition is somewhat technical, it is precisely the recursive construction one would expect based on the testing procedure.

\begin{definition}[Minimal Realization]\label{def:minimalReal}
For each tuple $(t,x)$ computed by the testing procedure, the \emph{minimal realization} is an input-labeled tree $T=(V,E,\phi,e)$, where $e$ denotes the special halfedge with input label $x$ and type $t$. It is defined inductively based on the iteration $k$ in which $(t,x)$ was first added to the set $R$:

\begin{enumerate}
    \item \textbf{Base case $k=0$.} These are the types produced by the initial execution of \emph{Compute Types(}$\emptyset$\emph{)}:
    \begin{itemize}
        \item If $(t,x)$ corresponds to a virtual tree $\T=(\emptyset,x)$, let $T$ consist of a leaf node $v$ and a parent $p$ with the edge $(\{v,p\},v)$ labeled $x$; the label on $(\{v,p\},p)$ can be arbitrary.
        \item For an iteration $i>0$ of \emph{Compute Types}, let $(t,x)$ correspond to a virtual tree $\T=(\La_{in},x)$. Start with a parent $p$ connected to a node $v$ representing $\T$, assigning $x$ to $(\{p,v\},v)$ and an arbitrary label to $(\{p,v\},p)$. 
        \item For every $(x_{adj},x_{far},t') \in \La_{in}$, take the minimal realization of the corresponding $(t',x_{far})$ obtained in a previous iteration, let its nodes be $(v',p')$, and identify $p'$ with $v$. Assign $x_{adj}$ to $(\{v,v'\},v)$.
    \end{itemize}

    \item \textbf{Iteration $k\ge 1$.} For $(X,x_1), (Y,y_L)$ belonging to an independent class $I=(X,Y)$ (Step~\ref{step:applyF}), construct the compress path $H=(V,E,\phi,e_s,e_t)$ by replacing all virtual trees in the path with their minimal realizations. Attach new nodes $s$ and $t$ to $e_s$ and $e_t$, respectively, with arbitrary labels on the other halfedges. Then:
    \begin{itemize}
        \item The minimal realization for $(X,x_1)$ is the tree $T_X=(V,E,\phi,e_s)$, with $v=v_1$ and $p=s$.
        \item The minimal realization for $(Y,y_L)$ is $T_Y=(V,E,\phi,e_t)$, with $v=v_l$ and $p=t$.
    \end{itemize}
    \item For any new types created in Step~\ref{step:computeTypes2}, we apply the same inductive construction, replacing virtual trees with minimal realizations as needed.
\end{enumerate}
\end{definition}

The next lemma makes the previous, cumbersome definition of minimal realizations practical. It states that if all compress paths used in constructing the minimal realization respect their assigned independent classes, then the resulting instance enforces the computed type.

\begin{lemma}\label{lem:solSatisfiesTypes}
Let $(t,x)$ be any tuple obtained during the testing procedure, and let $T$ be its minimal realization with special halfedge $e$. Let $\sigma$ be any output labeling such that:
\begin{itemize}
    \item $\sigma$ is a valid solution to $\Pi$.
    \item For every virtual compress path $\Hp = (\T_j)_{1 \le j \le L}$ encountered during the construction of $T$, with $f(\Hp) = (X,Y)$, the labeling on the corresponding subgraph $H = (V',E',\phi', e_s, e_t)$ of $T$ respects the independent class: $\sigma(e_s) \in X$ and $\sigma(e_t) \in Y$.
\end{itemize}
Then $\sigma(e) \in t$.
\end{lemma}

\begin{proof}
The proof proceeds by induction along the definition of minimal realizations.

\paragraph*{Base case ($k=0$)}  
We prove this by induction over the iterations $i$ of the \emph{Compute Types} procedure:
\begin{itemize}
    \item If $(t,x)$ was obtained in iteration $i=0$, it corresponds to a leaf node $v$ attached to a parent $p$, with $e = (\{v,p\},v)$ labeled $x$. By definition, all valid labelings must assign $e$ a label in $t$.
    \item If $(t,x)$ was obtained in iteration $i>0$, it corresponds to a virtual tree $\T=(\La_{in},x)$. Let $v$ represent $\T$ and $p$ its parent. For each child $c$ of $v$ corresponding to $(x_{adj}, x_{far}, t') \in \La_{in}$, the minimal realization of $(t',x_{far})$ ensures that any valid labeling assigns $e_c=(\{c,v\},c)$ a label from $t'$. Thus, the set of possible labels on $e=(\{v,p\},v)$ is exactly $t$.
\end{itemize}

\paragraph*{Inductive step ($k\ge 1$)}  
For tuples $(X,x)$ obtained as part of an independent class $I=(X,Y)$ in Step~\ref{step:applyF}, the assumptions of the lemma ensure that $\sigma$ assigns the correct labels to the corresponding edges.  
Similarly, any new tuples produced by \emph{Compute Types} in Step~\ref{step:computeTypes2} follow the same inductive argument.
\end{proof}

The problem now of course is that our compress paths are not long enough. This is where classical automata theory comes into play, in the next lemma we prove that the type of a virtual compress problem can be computed by a finite automaton. From this we will be able to proof a sort of pumping lemma for compress paths, which will allow us to replace each compress path in the minimal realization by a much longer version of that same compress path.

\begin{lemma}\label{lem:finiteAutomaton}
Let $\Hp = (\T_i)_{1 \le i \le L}$ be a virtual compress path of length $L \ge 2$.
Then $\Type(\Hp)$ can be computed solely from $\Type(\Hp_{-1})$ and $\Type(\T_L)$,
where $\Hp_{-1} = (\T_i)_{1 \le i \le L-1}$.
\end{lemma}

\begin{proof}
We describe how to compute $\Type(\Hp)$ from $\Type(\Hp_{-1})$ and $\Type(\T_L)$.

Let
\[
\T_{L-1} = (\La_{in}^{L-1}, x_{left}^{L-1}, x_{right}^{L-1})
\quad\text{and}\quad
\T_L = (\La_{in}^{L}, x_{left}^{L}, x_{right}^{L}).
\]
For every $(y_s,y_t) \in \Type(\Hp_{-1})$, we check whether there exists
$(y_{left},y_{right}) \in \Type(\T_L)$ such that
\[
\bigl((x_{right}^{L-1},y_t),(x_{left}^{L},y_{left})\bigr) \in \C_E.
\]
If this is the case, then we include $(y_s,y_{right})$ in $\Type(\Hp)$.

To see correctness, let $\sigma_{-1}$ be a labeling certifying
$(y_s,y_t) \in \Type(\Hp_{-1})$ and $\sigma_L$ a labeling that certifies $(y_{left},y_{right}) \in \Type(\T_L)$.
By combining these two labelings into a labeling $\sigma$, of all of $H$,
yields a labeling certifying $(y_s,y_{right}) \in \Type(\Hp)$. The correctness follows from the correctness of $\sigma_{-1}$ and $\sigma_L$, except for the single edge that connects $\Hp_{-1}$ and $\Hp_L$, which is labeled
\[
\bigl((x_{right}^{L-1},y_t),(x_{left}^{L},y_{left})\bigr)
\]
but this is in $\C_E$ by definition.

Conversely, let $(y_s,y_t) \in \Type(\Hp)$.
Then there exists a labeling $\sigma$ of the halfedges of $\Hp$ certifying this.
Restricting $\sigma$ to $\Hp_{-1}$ yields
$(y_s,y_t') \in \Type(\Hp_{-1})$ for some output label $y_t'$.
Moreover, $\sigma$ assigns an output label $y_{left}$ to the half-edge
$x_{left}^{L}$ such that $(y_{left},y_t) \in \Type(\T_L)$.
Since $\sigma$ is consistent, it holds that
\[
\bigl((x_{right}^{L-1},y_t'),(x_{left}^{L},y_{left})\bigr) \in \C_E,
\]
and therefore the above procedure correctly reconstructs $(y_s,y_t) \in \Type{\Hp}$.
\end{proof}

Since $\Type(\Hp_{-1})$ and $\Type(\T_L)$ are both finite, \Cref{lem:finiteAutomaton} implies that the type $\Type(\Hp)$ can be computed by a finite automaton. As a result, we can apply the pumping lemma for regular languages to virtual compress paths. We obtain the following result.

\begin{lemma}\label{lem:pumpingLemma}
    For every LFL $\Pi$, there exists some constant $\ell_{pump} \in \N$, such that for any compress path $H$ of length at least $\ell_{pump}$ and and any integer $x\ge \ell_{pump}$, there exists a compress path $H'$, satisfying the following:
    \begin{itemize}
        \item The length of $H'$ is at least $x$ and at most $x+ \ell_{pump}$.
        \item $\Type(H) = \Type(H')$
        \item The sets of input labeled trees attached to the core paths of $H$ and $H'$ are the same.
    \end{itemize}
\end{lemma}
\begin{proof}
    Let the length of $H$ be $L$, then by \Cref{lem:CompressToVirtual}, there exists a virtual compress path $\Hp = (\T_i)_{1 \le i \le L}$, of length $L$, such that $\Type(H)= \Type(\Hp)$. Furthermore, each of the $\T_i$ is a virtual tree with two poles representing one of the nodes on the core path of $H$. By \Cref{lem:finiteAutomaton} there exists a finite automaton $\A$ that reads $\Hp$ one virtual tree with 2 poles at a time and computes the type $\Type(\Hp)$. Let $\ell_{pump}$ be the number of states of $\A$. If $L \ge \ell_{pump}$, then when computing the type of $\Hp$, the automaton must loop on some subsequence of $(T_i)_{1 \le i \le L}$, of length $l$, at most $\ell_{pump}$. We consider $\Hp'$, that contains $\lceil \frac{x-L}{l}\rceil$ additional repetitions of this sequence (at the position of the original sequence). The length of $\Hp'$ is then 
    \[
        x \le L + l\cdot\left\lceil \frac{x-L}{l}\right\rceil \le x + l \le x + \ell_{pump}
    \]
    as desired. When reading $\Hp$ and $\Hp'$, $\A$, will terminate in the same state and so $\Type(\Hp) = \Type(\Hp')$. \\
    We then reconstruct an actual compress path, by replacing each $\T_i$ in $\Hp'$, with an isomorphic copy of the corresponding tree from $H$ and so also the last property follows.
\end{proof}

\subsection{Fast Algorithms Imply Good Functions}\label{sec:algoImpliesFunc}
We will now use this pumping lemma to construct a $k$-good function out of a deterministic $o(n^{1/k})$ algorithm $\A$. 
Using \Cref{lem:fastAlgo}, this then yields an algorithm running in
$O(n^{1/(k+1)})$ rounds.

We fix $n_0>0$, such that the runtime $T(n)$ of $A$ is bounded by $T(n_0) \le \varepsilon \cdot n_0^{1/k}$ for some very small constant $\varepsilon$ to be fixed later.

We will run the testing procedure again, this time constructing a function as we go. When asked for the independent class of a given virtual compress path, we will use the pumping lemma to turn the minimal realization into a graph instance on $n_0$ nodes. In that pumped realization all compress paths have length at least $2 T(n_0) +3$. We can then run our algorithm on that instance to determine the independent class.

To ensure that our algorithm behaves nicely, we will have to fix ids in a controlled manner on the created graph instances, this will be the main difficulty in reaching our result.

\begin{definition}[pumped realization]\label{def:pumpedRealization}
    For each type $t$ and its respective minimal realization, we define the pumped realization again inductively on the iteration $k$. 
    \begin{enumerate}
        \item For $k=0$ we change nothing.
        \item For $k\ge1$, for each tuple $(X,x)$ (and $(Y,y)$) obtained as part of an independent class in Step~\ref{step:applyF}, we first use \Cref{lem:pumpingLemma} to obtain a compress path $\Hp'$ of length between $10 T(n_0)$ and $10 T(n_0)
        + \ell_{pump}$ and then use the compress path $H=(V,E,\phi, e_s,e_t)$ of the same length, obtained by replacing all of the virtual trees on two poles with the minimal realizations that exist inductively.

        For the new types in Step~\ref{step:computeTypes2} we again do not change anything.
    \end{enumerate}
\end{definition}

Since we want to embed these pumped realizations into actual $n_0$ node graphs, we need to make sure that their size stays within some reasonable limits. 

\begin{lemma}\label{lem:CTSizes}
    Let $R$ be any subset of types, each with a pumped realization of size at most $S$. There is some constant $\kappa$, such that after Compute Types($R$), we can construct a pumped realization of each of the newly obtained types of size at most $\kappa \cdot S$.
\end{lemma}
\begin{proof}
    Run the Compute Types procedure with every possible inputset and let $h$ be the maximum number of iterations of Compute Types in any execution.

    We recursively construct the pumped realization as in \Cref{def:minimalReal}. In the worst case, we have only newly computed tuples $(t,x)$ (that were not in the set $R$ at the start of Compute Types) for $h$ iterations. However, after $h$ steps, the virtual trees that are left to recursively replace, all must have been in the initial set $R$ and so we replace all of them with pumped realizations that already existed. The maximum fan-out in every step is $(s+1)\cdot |\Sin|^2 \cdot 2^{|\Sout|}$ and so the maximum number of pumped realizations we need to attach is at most 
    \[
    \left((s+1)\cdot |\Sin|^2 \cdot 2^{|\Sout|}\right)^h
    \]
    The total number of nodes is then all nodes in the tree plus the size of all of the pumped realizations.
    \[
    \left((s+1)\cdot |\Sin|^2 \cdot 2^{|\Sout|}\right)^{h-1} + \left((s+1)\cdot |\Sin|^2 \cdot 2^{|\Sout|}\right)^{h} \cdot S 
    \]
    Which is less than $\kappa \cdot S$ for some suitable constant $\kappa$, since $h,s,|\Sin|,|\Sout|$ are constants .
\end{proof}

\begin{lemma}\label{lem:boundedRealizations}
    For each tuple $(t,x)$ computed by the testing procedure in iteration $j$, the pumped realization has $\le \kappa +  \frac{1}{2} n_0^{j/k}$ nodes. For the constant $\kappa$ from \Cref{lem:CTSizes}.
\end{lemma}
\begin{proof}
    By induction on $j$, for $j=0$ we get by \Cref{lem:CTSizes} that the realizations are at most of size $\kappa$.\\
    Let the statement hold true for all values $<j$, we prove that the pumped realizations of tuples $(t,x)$ created in iteration $j$ are small enough. There are two steps, where tuples get added, Step~\ref{step:applyF} and Step~\ref{step:computeTypes2}.\\
    For Step~\ref{step:applyF}, every virtual compress path $\Hp$ gets pumped to length $L$ between $10 T(n_0)$ and $10 T(n_0) + \ell_{pump}$. We chose $n_0$ such that $T(n_0)\le \varepsilon \cdot n_0^{1/k}$ for some very small constant $\varepsilon$. Let that constant be at most $\frac{1}{40\kappa^2}$, so 
    \[
    L \le 10 T(n_0) + \ell_{pump} \le 10 \frac{1}{40\kappa^2} n_0^{1/k} + \ell_{pump} \le \frac{1}{4\kappa^2} n_0^{1/k} + \ell_{pump} \le \frac{1}{2\kappa^2} n_0^{1/k}
    \]
    Now we attach to each node of that path at most $(s+1) \cdot |\Sin|^2 \cdot 2^{|\Sout|}$ many pumped realizations of tuples from previous iterations. By the induction hypothesis, all of these are of size at most $\kappa + \frac{1}{2} n_0^{(j-1)/k}$. So the total number of nodes in these pumped realizations is at most the sum of sizes of all of these realizations, plus the number of nodes in the path.
    \begin{align*}
        &\frac{1}{2\kappa^2} n_0^{1/k} \cdot ((s+1) \cdot |\Sin|^2 \cdot 2^{|\Sout|}) \cdot \left(\kappa + \frac{1}{2} n_0^{(j-1)/k}\right) + \frac{1}{2\kappa^2} n_0^{1/k} \\ 
        \le&  \frac{1}{2} n_0^{1/k} + \frac{1}{2\kappa^2} n_0^{1/k} + \frac{1}{4\kappa}n_0^{j/k}  \le \frac{1}{2\kappa}n_0^{j/k} 
    \end{align*}
    By remembering from the proof of \Cref{lem:CTSizes} that $(s+1) \cdot |\Sin|^2 \cdot 2^{|\Sout|}) < \kappa$ and by assuming that $n_0$ is sufficiently large ($\kappa$ does not depend on $n_0$).
    So the size of these pumped realizations of the tuples created in Step~\ref{step:applyF} is at most $\frac{1}{2\kappa}n_0^{j/k}$. So before Step~\ref{step:computeTypes2} the size of all pumped realizations is at most $\frac{1}{2\kappa}n_0^{j/k}$, so by \Cref{lem:CTSizes} after the execution of Compute Types in Step~\ref{step:computeTypes2} the sizes of all pumped realizations is at most $\kappa \cdot \frac{1}{2\kappa}n_0^{j/k} = \frac{1}{2}n_0^{j/k}$.
\end{proof}

Now that all pumped realizations have size less than $n_0$, we can pad them
arbitrarily to obtain graphs of exactly $n_0$ nodes. If we then simulate $\A$
on a node in the middle of a pumped path, $\A$ must commit to some output
label $y$ there without seeing the rest of the graph. This fixed choice
induces an independent class in the following sense.

\begin{lemma}\label{lem:independentClassFromLabel}
Let $\Hp = (\T_i)_{1 \le i \le L}$ be a virtual compress path of length $L$,
let $e$ be a half-edge of the core path of $\Hp$, and let $y$ be an output
label. Let $F$ be the set of all valid labelings of $\Hp$ that assign $y$ to
$e$. Then
\[
\Bigl(\{f(e_{left}) \mid f \in F\},\; \{f(e_{right}) \mid f \in F\}\Bigr)
\]
is an independent class for $\Hp$.
\end{lemma}

\begin{proof}
Fixing the output label $y$ on a half-edge $e$ of the core path separates
$\Hp$ into two subgraphs, $T_{left}$ and $T_{right}$. In any valid labeling,
the feasibility of the labeling on $T_{left}$ depends only on the labels
within $T_{left}$ and on the fixed label $y$ at $e$, and is independent of
the labeling on $T_{right}$. The same holds symmetrically for $T_{right}$.

Hence, for any two valid labelings $f_1,f_2 \in F$, the labeling that agrees
with $f_1$ on $T_{left}$ and with $f_2$ on $T_{right}$ is again a valid
labeling of $\Hp$. This is exactly the defining property of an independent
class.
\end{proof}

Using this idea, it is conceptually straightforward to derive an independent assigner from an algorithm $\A$. However, the output that $\A$ produces at the middle of a compress path may depend on the specific ID assignment in its $T(n_0)$-hop neighborhood.

Because of the recursive nature of our constructed instances, they may contain many copies of the same virtual compress path. To induce all of the types of the testing procedure, we must be able to reproduce the same independent class on each copy at the same time. To achieve this, we collect sets of ID assignments that all induce the same independent class.

\begin{definition}\label{def:impliedClass}
Let $\Hp$ be a virtual compress path and consider a pumped realization of $\Hp$. Let $v$ be the node in the middle of the core path of this realization, and let $e$ be a half-edge incident to $v$ that lies on the core path. A set of IDs $D$ \emph{implies} an independent class $I$ for $\Hp$ if there exists an assignment of the IDs in $D$ to the $T(n_0)$-hop neighborhood of $v$ such that, when $\A$ is run for $n_0$ rounds, the output label $y$ produced by $\A$ on $e$ induces the independent class $I$ as defined in \Cref{lem:independentClassFromLabel}.
\end{definition}

With this in place, we can now describe how an independent class is derived from $\A$.

\paragraph*{IndependentClass procedure}\hfill
\break
\emph{Input:}
\begin{itemize}
    \item A compress path $\Hp$ of length between $\ell_{\mathrm{pump}}$ and $2\ell_{\mathrm{pump}}$.
    \item A deterministic algorithm $\A$ and a parameter $n_0$.
    \item A set of available IDs $D$ of size $|\Sout| \cdot n_0$.
\end{itemize}

\emph{Output:}
\begin{itemize}
    \item An independent class $I$ for $\Hp$.
    \item A collection $\mathcal{D}$ of disjoint ID sets such that each $D' \in \mathcal{D}$ implies $I$ on $\Hp$.
\end{itemize}

Consider a pumped realization of $\Hp$, let $v$ be the node in the middle of its core path, and let $e$ be a half-edge incident to $v$ along the core path. Let $N$ denote the number of nodes in the $T(n_0)$-hop neighborhood of $v$.

Partition the ID set $D$ into disjoint subsets of size $N$. For each such subset, assign its IDs arbitrarily to the $T(n_0)$-neighborhood of $v$, attach a sufficiently long path at a node far from the simulated neighborhood to reach size $n_0$, and simulate $\A$ for $n_0$ rounds. Let $y$ be the output label assigned to $e$ most frequently over all simulations.

Fixing the label $y$ on $e$, apply \Cref{lem:independentClassFromLabel} to obtain an independent class $I$ for $\Hp$. Output $I$ together with all ID subsets that produced output $y$ on $e$.

\begin{lemma}\label{lem:enoughIds}
Let $\mathcal{D}$ be the collection of ID sets produced by the \emph{IndependentClass} procedure. Then the sets in $\mathcal{D}$ are disjoint and
\[
\left| \bigcup_{D \in \mathcal{D}} D \right| \ge n_0 .
\]
\end{lemma}

\begin{proof}
At least a $1/|\Sout|$ fraction of the tested ID assignments must produce the most frequent output $y$. Since the initial ID set has size $|\Sout| \cdot n_0$ and is partitioned into disjoint subsets, the claim follows.
\end{proof}

\paragraph*{Modifying the testing procedure}
We now modify the testing procedure to generate a function, using an algorithm, thereby defining the independent assigner \(f\) on the fly rather than testing a fixed function. The IndependentClass procedure relies on pumped realizations. By \Cref{lem:boundedRealizations}, any pumped realization constructed in iteration $j$ has size at most
\[
\kappa + \tfrac{1}{2}\, n_0^{\,j/k} .
\]
This bound is at most $n_0$ only for $j \le k$. Therefore, we restrict the construction of the function $f$ to the first $k$ iterations (inclusive).

Recall that in the testing procedure, whenever a virtual compress path \(\Hp\) is encountered, an independent class must be assigned to it. Instead of applying a pre-defined function \(f\), the first time a new virtual compress path \(\Hp = (\T_1,\dots,\T_L)\) appears, we invoke the \emph{IndependentClass} procedure using the ID set \(D_\Hp\), obtaining an independent class
\[
I_\Hp = (X,Y).
\]
We then define \(f(\Hp) := I_\Hp\) and extend the set of known types by adding the tuples \((X,x_1)\) and \((Y,y_L)\), where \(x_1\) is the left input label of \(\T_1\) and \(y_L\) the right input label of \(\T_L\).

From this point on, the value \(f(\Hp)\) is fixed and reused whenever \(\Hp\) reappears. In this way, the testing procedure incrementally constructs the function \(f\).

Since the testing procedure may encounter several distinct virtual compress paths, this step must be performed for each new such path. To ensure that sufficient IDs are available, we assume access to at least
\[
(z+1)\cdot |\Sout| \cdot n_0 \in O(n_0)
\]
distinct IDs, where \(z\) denotes the maximum number of compress paths that can be encountered during the testing procedure. By \Cref{cor:constantlyManyCompresses}, \(z\) is a constant. We partition the ID space into \(z\) disjoint sets of size \(|\Sout| \cdot n_0\), assigning one set \(D_\Hp\) to each possible virtual compress path of length between \(\ell_{\mathrm{pump}}\) and \(2\ell_{\mathrm{pump}}\).

The notions of minimal and pumped realizations are defined exactly as in the testing procedure. By \Cref{lem:boundedRealizations}, all pumped realizations arising from types computed in the first \(k\) iterations have size at most \(n_0\). Consequently, on each such realization we can apply our ID assignments to force the algorithm to respect the corresponding independent classes. 

We now formalize this argument and show that the correctness of the algorithm implies that the function constructed in this way is \(k\)-good.

\begin{lemma}\label{lem:instanceForcesType}
Let $(t,x)$ be any tuple obtained during the construction of the function $f$ from the algorithm $\A$. Then there exist
\begin{itemize}
    \item an input-labeled tree $T=(V,E,\phi)$ with $|V|=n_0$,
    \item a distinguished halfedge $e\in E$, and
    \item an ID assignment to the nodes of $T$,
\end{itemize}
such that when $\A$ is executed on $T$, either the output assigned to $e$ lies in $t$, or the produced labeling is invalid.
\end{lemma}

\begin{proof}
Since $(t,x)$ is produced during the construction of $f$, it arises either as the type of a virtual tree $\T$ with $\Type(\T)=t$, or from an independent class $I$ assigned to some virtual compress path $\Hp$ with $t\in I$.

Let $T=(V,E,\phi)$ be the pumped realization of $\T$ (or of $\Hp$), and let $v$ be the node representing $\T$, with special parent $p$. Let $e$ be the halfedge between $v$ and $p$. By \Cref{lem:boundedRealizations} the number of nodes in $T$ is at most $\kappa + \frac{1}{2} \cdot n_0^{(k)/k} \le n_0$. If the pumped realization has fewer than $n_0$ nodes, we attach a path of nodes to $p$ so that the resulting graph $T'$ has exactly $n_0$ nodes.

During the recursive construction of $T$, every encountered virtual compress path $\Hp'$ has already been processed by the construction procedure. For each such $\Hp'$, the IndependentClass procedure produced an independent class $f(\Hp')=(X,Y)$ together with a collection of disjoint ID sets, each implying this class. We assign one unused ID set to each occurrence of $\Hp'$ in the construction. By \Cref{lem:enoughIds}, the total number of available IDs suffices for the entire realization.

Moreover, IDs are assigned only within the middle $2T(n_0)$ nodes of each pumped compress path and their attached subtrees. Each compress path is pumped to length $10T(n_0)$, so distinct assigned regions are disjoint, ensuring that $\A$ behaves in each of these neighborhoods exactly as it does in the IndependentClass procedure. We assign arbitrary unused IDs to all other nodes.

Now consider any virtual compress path $\Hp'$ used in the construction, with corresponding subgraph $H=(V',E',\phi',e_s,e_t)$. By construction and \Cref{lem:independentClassFromLabel}, the ID assignment forces any valid labeling produced by $\A$ to assign $e_s$ an output in $X$ and $e_t$ an output in $Y$.

Thus the labeling produced by $\A$ satisfies the assumptions\footnote{Note that \Cref{lem:solSatisfiesTypes} technically requires minimal realizations, but since the pumped realizations contain compress paths of the same types as in the minimal realizations, the same proof works for pumped realizations as well.} of \Cref{lem:solSatisfiesTypes}, which implies that the output on $e$ must lie in $t$ unless the labeling is invalid.
\end{proof}

Now it is clear that $f$ must be $k$-good, as otherwise we would have an empty type, but then \Cref{lem:instanceForcesType} states that $\A$ assigns a label from an emptyset.

\begin{lemma}\label{lem:fIsGood}
    The independent assigner $f$ constructed by the construction procedure is $k$-good.
\end{lemma}
\begin{proof}
    Suppose not, then, in the first $k$ iterations, the testing procedure finds a virtual tree $\T$, such that the type $\Type(\T)$ is the empty set. Then \Cref{lem:instanceForcesType} gives us an instance on which $\A$ fails to compute a correct solution, a contradiction to the fact that $\A$ is a correct algorithm.
\end{proof}

These results now immediately imply our final result about the complexity gaps in the polynomial regime of the complexity landscape.

\begin{theorem}
    Let $\Pi$ be any LFL that admits an $o(n^{1/k})$ round deterministic LOCAL algorithm, for any integer $k \ge 1$.

    There exists a deterministic LOCAL algorithm that solves $\Pi$ in $O(n^{1/(k+1)})$ rounds.
\end{theorem}
\begin{proof}
    We construct an independent assigner $f$ that is $k$ good. Then by \Cref{lem:fastAlgo} there is an algorithm that solves $\Pi$ in $O(n^{1/(k+1)})$ rounds.
\end{proof}

\begin{theorem}
    Let $\Pi$ be any LFL that admits an $n^{o(1)}$ round deterministic LOCAL algorithm.

    There exists a deterministic LOCAL algorithm that solves $\Pi$ in $O(\log n)$ rounds.
\end{theorem}
\begin{proof}
    Let $k$ be the maximum number of iterations that the testing procedure can run (the set $R$ must grow in every iteration).

    Since the algorithm runs in $n^{o(1)}$ rounds, it can be made to run in time $o(n^{1/k})$ and so we can use it to construct a $(k+1)$ good function. Since $k$ is the maximum number of iterations the testing procedure can run, this means that $f$ must be $\infty$-good.

    Then by \Cref{lem:fastAlgo} there is an algorithm that solves $\Pi$ in $O(\log n)$ rounds.
\end{proof}

Together these two results imply \Cref{thm:GapResults}, what is left is to argue that we can determine the correct $k$ just based on the description of $\Pi$.

The constant $\ell_{\text{pump}}$ depends only on the description of $\Pi$ and with that value fixed, the number of possible independent assigners is a constant. We may now test all of them in the testing procedure to determine the maximum $k$ for which a $k$-good independent assigner exists. This means we can compute which theorem applies, based solely on the description of $\Pi$.  

\GapResults*

\section{Outlook and Open Questions}

The results in \Cref{thm:polyProblems,thm:GapResults} show that whether allowing unbounded degrees changes the complexity landscape on trees depends strongly on the chosen problem definition. In this work, we show that for LFLs the polynomial part of the complexity landscape coincides with that of LCLs on bounded-degree trees, while more generic definitions would not allow such gap results. This naturally raises the question of what happens below the $O(\log n)$ regime.

More specifically, which structural properties of a problem determine whether it can be solved efficiently when degrees are large? 

We believe that understanding the case of LFLs would be a big step in the right direction. Intuitively, we are interested in what complex behavior already emerges when restricting to a finite number of labels and a finite number of cases.

As discussed earlier, the complexity of MIS is $\Theta(\log n / \log \log n)$, or more precisely $\Theta(\min\{\Delta, \log_\Delta n\})$, showing that the lower regime of the complexity landscape must necessarily change in the unbounded-degree setting, even for LFLs. 
We know that the problem of Sinkless Orientation remains $\Theta(\log n)$ even in the unbounded-degree setting, whereas problems such as $(2,2)$-ruling sets—also encodable as LFLs—are known to have a $\sqrt{\Delta}$ dependence \cite{balliu2022hide}. This motivates the following question.

\vspace{0.15cm}
\begin{tcolorbox}
\textbf{Question:} What dependencies on $\Delta$ are possible for LFLs, and can one prove corresponding complexity gaps?
\end{tcolorbox}

One of the strongest gap results, due to \cite{CKP19exponential}, has as a starting point problems with complexity $o(\log_\Delta(n))$ and readily extends to LFLs.

\begin{theorem}[\cite{CKP19exponential}] \label{thm:logGap}
    Let $\Pi$ be an LFL on hereditary graphs and let $\A$ be a deterministic LOCAL algorithm for $\Pi$. If $\A$ solves $\Pi$ in $f(\Delta) + o(\log_\Delta n)$ rounds for some function $f$, then there exists a deterministic LOCAL algorithm solving $\Pi$ on the same class of graphs in
    \[
        O\bigl((1+f(\Delta))(\log^* n - \log^* \Delta)\bigr)
    \]
    rounds.
\end{theorem}

For bounded-degree graphs, this result recovers the $o(\log n)$ versus $\omega(\log^* n)$ gap, because the multiplicative $f(\Delta)$ dependencies become negligible.  

While the multiplicative $f(\Delta)$ dependencies in \Cref{thm:logGap} are unavoidable for some problems, such as $\Delta$-ruling sets or distance-$\Delta$ coloring, these problems cannot be encoded as LFLs. They require referencing something at distance $\Delta$ and hence can only be encoded using \emph{pointer chains} of length $\Delta$. These can only be realized with $\Omega(\Delta)$ distinct labels\footnote{As long as the checking radius needs to be a constant.}. In contrast, LFLs require a finite label set independent of the degree.

We believe that \Cref{thm:logGap} is not tight for LFLs. In particular, the multiplicative $f(\Delta)$ dependence may be unnecessary and could potentially be replaced by an additive one.

\vspace{0.15cm}
\begin{tcolorbox}
\textbf{Conjecture:} Any LFL $\Pi$ that admits an $f(\Delta) + o(\log_\Delta n)$-round LOCAL algorithm also admits an $O(1)$-locality SLOCAL algorithm, and hence an $O(g(\Delta) + \log^* n)$-round LOCAL algorithm for some function $g(\Delta)$.
\end{tcolorbox}

The above result by \cite{CKP19exponential} already gives us an idea of what the correct structural property is for problems that are $o(\log_\Delta n)$. Namely that these problems admit fast SLOCAL algorithms.

The SLOCAL model is a sequential analogue of the LOCAL model introduced in \cite{ghaffari2017complexity}. We process nodes in an arbitrary (adversarially chosen) order. When node $v$ is processed, $v$ can read its $r$-hop neighborhood and it computes and locally stores its output $y$ and potentially additional information.
When reading its $r$-hop neighborhood, $v$ also reads all the information that has been locally stored by the
previously processed nodes there. We call the parameter $r$ the locality of an SLOCAL algorithm.

Finally, for bounded-degree trees, the remaining part of the complexity landscape is the gap between $\omega(1)$ and $o(\log^* n)$. It is currently unclear whether this result—and its proof based on round elimination \cite{brandt21trees}—can be extended to LFLs. The node-edge checkable formalism for LFLs appears sufficiently expressive to potentially support such an extension, but it is unclear in which way a dependence on $\Delta$ might emerge, if at all.

\vspace{0.15cm}
\begin{tcolorbox}
\textbf{Question:} What complexities are possible for LFLs in the range $\omega(1)$ to $o(\log^* n)$?
\end{tcolorbox}

\appendix
\crefalias{section}{appendix}

\section{3-cycle Detection as an Elaborate Example}\label{apx:examples}
We finish our suite of examples with a more complex one: detecting 3-cycles. Formally every node must output 1 if and only if it is part of a 3-cycle and 0 otherwise. This is also the first case where it is nice to use a self-loop inside of a configuration. We define $3CYCLE = (\emptyset, \set{0,1}, 2, \C)$ with $\C$ being given in \Cref{fig:3Cycle}.

\begin{figure}[!ht]
    \centering
    \includesvg[width=0.6\linewidth]{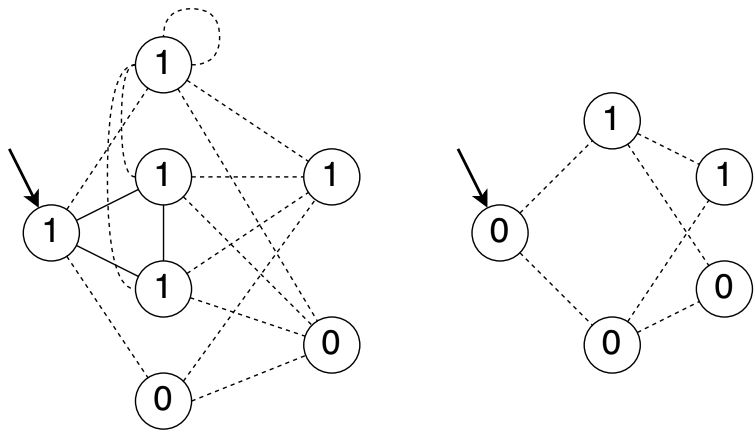}
    \caption{Two configurations to encode 3-cycle detection.}
    \label{fig:3Cycle}
\end{figure}

We will now argue that $3CYCLE$ accurately captures the problem. We will consider some node $v$ and show that it can output $1$ if and only if it is in a 3-cycle. 
We will only argue the case in which $v$ is not in a 3-cycle and the other case follows similarly. For this we show that $v$ cannot use the left configuration and therefore cannot output 1. Any two neighbors $u,w \in \Nh(v)$ that both output 1, cannot be mapped to the required nodes in the left configuration. Since $v$ is not in a 3-cycle $u,w$ cannot have an edge between them. But by Rule \ref{rule:required} this edge must exist if $v$ matches configuration 1. As a result $v$ cannot use this configuration. We now argue that $v$ can always use the second configuration. So $v$ outputs 0 and maps any neighbors at distance 1 to its immediate neighbors in the configuration. It then maps all of its distance 2 neighbors onto the respective distance 2 neighbors in the configuration, this is always possible, and so $v$ can be matched to the second configuration.

\section{LCLs with Infinitely Many Configurations}\label{sec:partition-gadgets}
We present a formal version of the construction described in the introduction. 
For this construction, we retain the \emph{Local} of Locally Checkable Labelings in the sense that solutions depend only on the local information around a node. To this end, we call a problem \emph{$r$-checkable} if 
\begin{enumerate}
    \item the local correctness of a solution at a node $v$ depends only on the $r$-hop neighborhood of $v$, and
    \item a solution is globally correct if and only if it is locally correct at every node.
\end{enumerate} 

While in the introduction we enumerated all possible rooted trees with a certain number of nodes, we do the same here formally by using partitions of integers. In principle each such partition represents a possible height 2 tree, which we will call a partition gadget.

These partition gadgets provide many distinct local configurations of the same overall size, allowing us to make nodes locally distinguishable without rapidly consuming nodes.

We begin by recalling some basic facts about integer partitions, which form the combinatorial foundation of our gadgets.

\begin{definition}[Integer Partition]
A \emph{partition} of a natural number $n \in \mathbb{N}$ is a multiset of positive integers 
\[
\La = \set{a_1, a_2, \ldots, a_k}
\]
such that 
\[
a_1 + a_2 + \cdots + a_k = n,
\]
where the ordering of the summands is irrelevant. 
We denote by $p(n)$ the number of distinct partitions of $n$.
\end{definition}

The function $p(n)$ grows quickly, providing the combinatorial explosion that we will exploit shortly. It is a classical result of Hardy and Ramanujan~\cite{HardyRamanujan1918} that the partition function satisfies the asymptotic estimate:
\begin{lemma}[\cite{HardyRamanujan1918}]\label{lem:boundPartitions}
\[
p(n) \sim \frac{1}{4n\sqrt{3}} \exp\!\Bigl(\pi \sqrt{2n/3}\Bigr).
\]
\end{lemma}

\medskip
With this background in place, we now describe how to turn a partition into a local structure that can be attached to a node.

\begin{definition}[Partition Gadget]
    Let $\La = (a_1, a_2, \ldots, a_k)$ be a partition of some integer.  
    To construct the \emph{$\La$-gadget} attached to a node $v$, we proceed as follows. For each $1 \le i \le k$: 
    \begin{enumerate}
        \item if $a_i = 1$, attach a single leaf node directly to $v$;
        \item if $a_i > 1$, attach a new node $u_i$ to $v$, and then attach $a_i - 1$ leaf nodes to $u_i$.
    \end{enumerate}
\end{definition}

In other words, each part $a_i$ of the partition corresponds either to a single leaf or to a small star centered at a new node $u_i$. Distinct partitions thus lead to structurally different neighborhoods of radius $2$, while keeping the overall gadget size bounded.  

\medskip
Next, we explain how to combine partition gadgets into paths.

Let $0 < \alpha < 1$ and fix a natural number $g$. We fix an ordering $(\La_i)_{i \in [p(g)]}$ of all partitions of $g$ for each possible value of $g$. A \emph{gadget path} is then a path whose nodes $v_i$ carry these gadgets in sequence. Because $p(g)$ is exponentially large in $\sqrt{g}$, we have more than enough gadgets available to sustain long paths. However, to keep the total number of nodes under control, we truncate the path after $\tfrac{1}{2}\lfloor g^{\frac{1}{\alpha}- 1} \rfloor$ gadgets.  

If we choose $g$ large enough with respect to $\alpha$, then we always have enough gadgets. The following is an immediate consequence of \Cref{lem:boundPartitions}.

\begin{corollary}\label{cor:largeEnoughG}
Let $0<\alpha<1$, there exists a constant $g_\alpha$, such that all integers $g>g_\alpha$ satisfy   
\[
    g^{\frac{1}{\alpha}-1} < p(g)
    \]  
\end{corollary}

Now that we know that we have enough gadgets, we can define our length constrained gadget paths. 

\begin{definition}[Gadget Path]
    Fix a unique order $(\La_i)_{i \in [p(n)]}$ for the partitions of each natural number $n$.

    Let $0 < \alpha < 1$ be any real number and let $g$ be any integer large enough s.t.
    \[
    g^{\frac{1}{\alpha}-1} < p(g)
    \]
    then a path 
    \[
    P = (v_1, v_2, \ldots, v_L)
    \]
    in any graph is an \emph{$(\alpha,g)$ gadget path} if each $v_i$ has the $i$-th partition gadget for $g$ attached and the length $L$ is bounded by 
    \[
    1 < L \le \tfrac{1}{2} g^{\frac{1}{\alpha}-1}
    \]
\end{definition}

The key property of this definition is that despite the enormous variety of gadgets, the total number of nodes remains bounded with respect to $g$. In any graph on $n$ nodes, the length of any $(\alpha,g)$ gadget path is restricted by (i) the inherent length restriction of the definition (ii) the fact that for large values of $g$ the graph construction uses a lot of nodes. We obtain the following corollary. 

\begin{corollary}\label{cor:gadgetPaths}
    An $(\alpha,g)$-gadget path of length $L$ contains exactly 
    \[
    L + L \cdot g
    \]
    nodes. Consequently for any fixed $0 < \alpha < 1$ and all allowed integers $g$, the length of the longest $(\alpha,g)$-gadget path in an $n$-node graph is in $\Theta(n^{1-\alpha})$.
\end{corollary}
The second part comes from the fact that, by definition, the lengths $L$ of the gadget paths are bounded by $\tfrac{1}{2}g^{\frac{1}{\alpha}-1}$.

\subsection*{The Family of Problems via Gadget Paths}

We are now ready to define the actual locally checkable problem: a node must output $1$ if and only if it belongs to a valid gadget path, and must output $0$ otherwise. 

\begin{definition}[Gadget-Path Problem $\Pi(\alpha)$]
    An instance of $\Pi(\alpha)$ is a tree $T=(V,E)$ where nodes are required to output one of $\{0,1\}$. The labeling is \emph{locally correct} if and only if every node $v \in V$ satisfies the following:  
    \begin{itemize}
        \item If $v$ is contained in a $(\alpha, g)$-gadget path for some valid $g$, then $v$ must output $1$.  
        \item If $v$ is not contained in any valid gadget path, then $v$ must output $0$.  
    \end{itemize}
    A labeling is \emph{globally correct} if it is locally correct at every node.
\end{definition}

we make the following observations:
\begin{observation}\mbox{}
    \begin{enumerate}
        \item The problem is $4$ checkable.
        \item The problem can be represented in local configurations.
    \end{enumerate}
\end{observation}
We now explain these two points in more detail.  

First, each node $v$ can verify whether or not it has a gadget attached using information from within distance~3. The nodes of $v$'s gadget are at distance 2 from $v$, so checking whether these nodes are leafes requires distance 3 information\footnote{Refer to the definitions in the preliminaries to verify exactly how these distances are defined.}.
Moreover, to identify the gadgets attached to its neighbors, $v$ requires information up to distance~4. From this information, $v$ can determine the corresponding parameter $g$ and which gadget is attached to itself and its neighbor. 

Second, to represent the problem as local configurations, we need a configuration that indicates the start of such a path and one which handles the propagation along the path. We require these configurations for every possible choice of $g$. This generic construction is illustrated in \Cref{fig:partitionProblem}. Importantly, the arbitrary extensions of these configurations (i.e., the dotted nodes marked “arbitrary non-gadget neighbors”) must be constrained so that they cannot themselves form part of a gadget.

To see why this is necessary, consider the top configuration in \Cref{fig:partitionProblem}, and assume that the center node $v$ (marked by the arrow) has a gadget $G_g(\La_i)$ attached. Suppose further that $v$ also has a neighbor $u$ which itself has only nine degree-one neighbors (i.e., leaves). In this case, $v$ would not actually have gadget $G_g(\La_i)$ attached; instead, the combined structure would correspond to a partition of a number that is larger by ten. To avoid such ambiguities, all neighbors of $v$ that are either leaves themselves or that have only leaves (other than $v$) are considered part of a single gadget.  

\begin{figure}[!ht]
    \centering
    \includesvg[width=0.8\linewidth]{img/partitionProblemConfigs.svg}
    \caption{ Key local configurations of $\Pi(\alpha)$, with center nodes indicated by arrows. Nodes without an explicitly specified output may take an arbitrary label (0 or 1). The $G_g(i)$ cones indicate that the gadget of the $i$-th partition of $g$ is attached to a node.\\
    \hspace{1.5cm}\textbf{Bottom :} Start of an $(\alpha,g)$ gadget path. The center node must output 1 if it has gadget $G_g(1)$ attached and it has a neighbor with the $G_g(2)$ gadget attached.\\
    \hspace{1.5cm}\textbf{Top:} Propagation within an $(\alpha,g)$ gadget path. The center node has gadget $G_g(i)$ attached (for some $i \in \mathbb{N}$), and one of its neighbors must have gadget $G_g(i-1)$ attached and output 1. So the configuration essentially requires a predecessor. This configuration only exists up until $i \le \frac{1}{2}g^{\frac{1}{\alpha}-1}$ for each $g$.}
    \label{fig:partitionProblem}
\end{figure}

Now the only thing left to do, is proof that our construction allows us to obtain any polynomial complexity.
\PolyProblems*
\begin{proof}
    Let $0 < r < 1$ be arbitrary, and set $\alpha = 1-r$.  
    We argue that the Gadget-Path Problem $\Pi(\alpha)$ has complexity $\Theta(n^{1-\alpha}) = \Theta(n^r)$.  
    As before, the upper and lower bounds follow from standard distributed arguments.

    \smallskip
    \noindent\textbf{Upper bound.}  
    Consider the deterministic algorithm where each node $v$ first determines the largest partition gadget attached to it.  
    By counting the size $g$ of this gadget, it can determine the $(\alpha,g)$-path constraints that it needs to satisfy (and exactly where its attached gadget lies in the fixed ordering of all partitions of $g$).  
    By \Cref{cor:gadgetPaths}, it is sufficient to explore a neighborhood of radius $O(n^{1-\alpha})$ to decide whether $v$ belongs to such a path.  
    Therefore, after at most $O(n^{1-\alpha})$ rounds in the LOCAL model, node $v$ can correctly output $1$ if it lies on a valid gadget path, and $0$ otherwise.  
    By construction, this yields a correct solution.

    \smallskip
    \noindent\textbf{Lower bound.}  
    For the lower bound, we use the standard indistinguishability argument.  
    Suppose a randomized LOCAL algorithm runs in time $T(n)$.  
    Then for any two instances in which a node $v$ has the same $T(n)$-neighborhood, the output distribution of $v$ must be identical in both.

    Assume, for contradiction, that there exists a constant $\varepsilon > 0$ such that $T(n) \le \varepsilon n^{1-\alpha}$ and the algorithm succeeds with probability greater than $1/2$.  
    Choose $n$ large enough so that these conditions hold and such that $g:=\lfloor n^\alpha \rfloor$ is large enough for \Cref{cor:largeEnoughG}.

    Consider an $(\alpha, g)$-gadget path 
    \[
    P = (v_1, v_2, \ldots, v_L)
    \]
    of length 
    \[
    L = \tfrac{1}{4} \lfloor n^{1-\alpha}\rfloor < \frac{1}{2} g^{\frac{1}{\alpha}-1}
    \]
    , which by \Cref{cor:gadgetPaths} uses fewer than $n$ nodes.  
    We can extend $P$ to an $n$-node tree by, for example, attaching an additional path to $v_L$, which preserves the gadget-path property.

    Now construct two instances:
    \begin{itemize}
        \item $G_1$: the graph containing the full path $P$,
        \item $G_2$: the same graph with $v_1$ removed. In this case, $G_2$ contains no valid $(n,\alpha)$-gadget path (nor any valid $(x,\alpha)$-path for other $x$).
    \end{itemize}
    In $G_1$, the node $v_L$ must output $1$, while in $G_2$, the same node must output $0$.  
    However, since \[
    L = \tfrac{1}{4} \lfloor n^{1-\alpha} \rfloor > \varepsilon n^{1-\alpha} \ge T(n)
    \]
    , the $T(n)$-neighborhood of $v_L$ is identical in both graphs.  
    Thus $v_L$ follows the same output distribution in $G_1$ and $G_2$, and therefore with probability at least $1/2$ the algorithm errs on one of the two instances.

    Hence, the randomized complexity of $\Pi(\alpha)$ is $\Omega(n^{1-\alpha})$, which matches the deterministic upper bound.

    \medskip
    Combining both bounds, the distributed complexity of $\Pi(\alpha)$ is
    \[
        \Theta(n^{1-\alpha}) = \Theta(n^r),
    \]
    as claimed.
\end{proof}

\section{Partially Ordered Locally Finite Labelings } \label{sec:LFLpO}
In some of the proofs in \Cref{sec:ReReduction}, we need to work with more complex output labelings that consist of multiple components, i.e., output labels that are tuples. To keep the configurations manageable, we introduce auxiliary labels that serve as wildcards. For instance, one particularly useful auxiliary label is the special symbol “any label”. 

To put the proofs on a rigorous mathematical foundation, we explicitly define LFLs with partially ordered outputs and show that they are equivalent to our standard notion of LFLs.

The idea of “any label” is that it appears only inside configurations: if a configuration specifies that a node has output “any label”, then this matches any actual output assigned to that node in the solution. In other words, whatever output label a node $v$ carries, it can be matched to a configuration node labeled “any label”.  

We formalize this construction by defining a partial order that relates auxiliary labels to the ordinary output labels. Formally, we introduce partially ordered LFLs, whose definition is modified in the following way:

\begin{itemize}
    \item Definition of a problem: A partially ordered LFL is a tuple $\Pi=(\Sin, \Sout, A, \le, r,\C)$, where $\Sin$, $\Sout$, $r$, $\C$ remain unchanged, but there are now also a set $A$ of auxiliary labels and $\le$ a partial order on $\Sout \cup A$.
    \item Definition of Configurations: an $r$-hop configuration $C=(\V, \E, \cen, \mu, \tau)$ now requires $\mu:V \rightarrow \Sin \times (\Sout\cup A)$ to assign inputlabels together with an outputlabel, or an auxiliary label.
    \item The \textbf{Labelings are respected} rule now requires all of the labels to match in a different sense. We say $f_v:\Nh_r(v) \rightarrow C$ respects the labeling if and only if for any node $u \in \Nh_r(v)$, it holds that the input label assigned to $u$ is the same as the input label assigned to $f(u)$ by $\mu$ and the outputlabel assigned to $u$ is less than or equal to the label assigned to $f(u)$.
\end{itemize}

\paragraph*{3-coloring with partially ordered labels}
With a partial order, we can express the problem of 3-coloring as very intuitive configurations. Namely the LFL $3COL = (\emptyset, \set{R,G,B}, \{\bar{R}, \bar{G}, \bar{B}\}, \leq_3, 1, \C)$, where $\le_3$ is defined in the following way:
\begin{align*}
    G &\le_3 \bar{R}, B \le_3 \bar{R}\\
    R &\le_3 \bar{G}, B \le_3 \bar{G}\\
    R &\le_3 \bar{B}, G \le_3 \bar{B}
\end{align*}

Basically $\bar{R}$ encodes ''not R'' and therefore both $B$ and $G$ can match $\bar{R}$, this leads us to the nice representation of the configurations $\C$ depicted in \Cref{fig:3COLpO}. 

\begin{figure}[!ht]
    \centering
    \includesvg[width=0.7\linewidth]{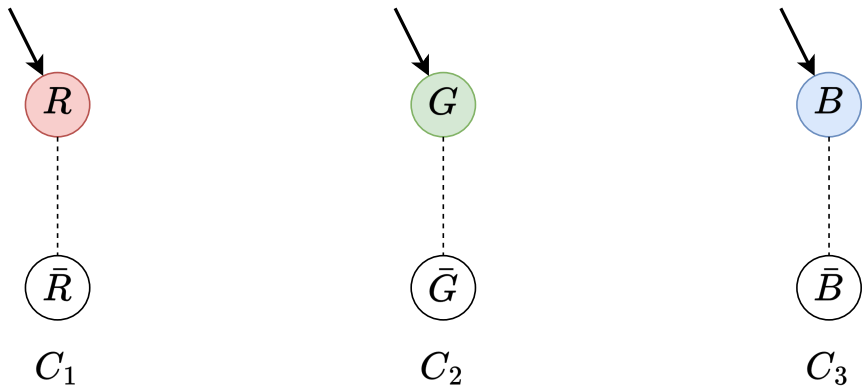}
    \caption{Using the ''not \_\_'' labels ($\{\bar{R}, \bar{G}, \bar{B}\}$) we can encode 3-coloring in the very intuitive way, where a node that outputs $R$ must have only neighbors that are ''not $R$''.}
    \label{fig:3COLpO}
\end{figure}

We want to show that allowing auxiliary labels with a partial order is just a notational convenience and does not actually change the set of problems. To achieve this, we first have to fix some notation.

\begin{definition}[required node]
    For any configuration $C = (\V, \E, \cen, \mu, \tau)$, we call a node $v \in \V\setminus \set{\cen}$ \emph{required} if and only if there exists an incident edge $e$, such that $\tau(e) = required$. In the same way, we call $v$ optional if and only if for all adjacent edges $e$ it holds that $\tau(e) = optional$.
\end{definition}

We prove the equivalence of LFLs with partially ordered outputs to regular LFLs, by removing auxiliary labels from our configurations one at a time. To that end we prove the following lemma.

\begin{lemma}\label{lem:redAux}
    Let $\Pi = (\Sin, \Sout, A, \le, r, \C)$ be an LFL with partially ordered outputs. For every $r$-hop configuration $C \in \C$ which contains at least one node with an auxiliary label, there exists a finite family $\C_C$ of $r$-hop configurations, such that 
    \begin{itemize}
        \item Every $C' \in \C_C$ contains strictly fewer auxiliary labels than $C$.
        \item If an input/output-labeled $r$-hop neighborhood $\Nh_r(v)$ matches $C$, then $\Nh_r(v)$ matches some configuration $C' \in \C_C$.
        \item Conversely, if an input/output-labeled $r$-hop neighborhood $\Nh_r(v)$ matches some configuration $C' \in \C_C$, then $\Nh_r(v)$ also matches $C$.
    \end{itemize}
\end{lemma}

\begin{proof}
    Let $\Pi=(\Sin,\Sout,A,\le,r,\C)$ be an LFL with partially ordered outputs, and let
    $C=(\V,\E,\cen,\mu,\tau)\in\C$ be a configuration that uses at least one auxiliary
    label. We pick one auxiliary-labeled node and eliminate its auxiliary label by
    replacing $C$ with a finite family $\C_C$ of configurations with strictly fewer auxiliary labels. 

    Let $v \in \V$ be a node that is assigned an auxiliary label $y$.
    Let 
    \[
        Y=\{y'\in\Sout \mid y'\le y\}
    \]
    be the set of all outputlabels that are allowed for $v$ under the partial
    order. We distinguish two cases depending on whether or not $v$ is required.

    \paragraph*{Case 1: $v$ is required}
    For every $y'\in Y$ we create a copy of $C$ in which $v$ is labeled with $y'$. Let $\C_C$ be the set of all such copies.

    \paragraph*{Case 2: $v$ is optional}
    We now form a single configuration $C'$.  
    Create $|Y|-1$ copies of $v$ (so we have $|Y|$ versions of $v$ total) and label one of them with $y'$, for every possible $y' \in Y$.  
    Since all replaced labels are in $\Sout$, the configuration $C'$ contains strictly
    fewer auxiliary labels.

    Define $\C_C=\{C'\}$ in this case.

    \medskip
    Clearly we satisfy the first part of the lemma, since each newly created configuration has exactly 1 auxiliary label less.

    \paragraph*{If $\Nh_r(v)$ matches $C$, then it matches some $C'\in\C_C$}
    Let $f:\Nh_r(v)\to \V$ be a matching of $C$.

    \begin{itemize}
        \item \emph{Case 1}  
        Since $v$ is required exactly one node is mapped to $v$ by $f$. Because $f$ respects labels, the label assigned to $v$ by $f$ is in $Y$ and hence a suitable version of $C$ exists in $\C_C$.

        \item \emph{Case 2:}  
        If $f$ maps some node $w$ to $v$, again because $f$ respects labels, $w$ has some output $y' \in Y$, so map $w$ to the suitable copy in $C'$.
    \end{itemize}

    The other direction follows analogously.
\end{proof}

Equipped with \Cref{lem:redAux} we can show that auxiliary variables are just a nice tool, but do not change the class of problems.

\begin{lemma}\label{lem:poLFlisLFL}
    For every LFL with partially ordered outputs, there exists an equivalent LFL without partially ordered outputs, and vice versa.
\end{lemma}
\begin{proof}
    Clearly, every LFL is also an LFL with partially ordered outputs, by setting $A=\emptyset$ and using the trivial partial order in which distinct elements are incomparable.
    
    \smallskip

    We therefore focus on the more substantial direction. 
    Let $\Pi = (\Sin, \Sout, A, \le, r, \C)$ be an LFL with partially ordered outputs. 
    For every configuration $C\in\C$, we iteratively apply \Cref{lem:redAux} until we obtain a set of configurations that contains no auxiliary labels. Let $\C'$ be the union of all configuration sets produced in this way. 
    Since no configuration in $\C'$ uses any auxiliary labels, the structure
    \[
        \Pi' = (\Sin, \Sout, r, \C')
    \]
    is a valid LFL \emph{without} partially ordered outputs.

    Finally, the two correctness guarantees of \Cref{lem:redAux}---that every neighborhood matching a configuration in $\C$ matches some configuration in the expanded family, and that every match to a configuration in the expanded family also matches the original one---imply that $\Pi$ and $\Pi'$ define exactly the same set of legal solutions.
\end{proof}

\bibliographystyle{plainurl}
\bibliography{biblio}

\end{document}